\documentclass[11pt,reqno]{amsart}
\usepackage{amsthm, amsmath,amsfonts,amssymb,euscript,hyperref,graphics,color}
\usepackage{graphicx}
\usepackage{comment}
\usepackage{import}
\usepackage{tikz}
\usepackage{latexsym}

\newcommand{\p}{\partial}

\newcommand{\di}{\mathrm{d}} 
\newcommand{\dt}{\partial_t}
\newcommand{\tr}{\text{tr}}
\newcommand{\dive}{\text{div}}
\newcommand{\curl}{\text{curl}}
\newcommand{\lie}{\mathcal{L}} 
\newcommand{\dtau}{\partial_{\bar \tau}}
\newcommand{\Tbar}{\bar T}
\newcommand{\D}{\mathcal{D}}

\newtheorem{theorem}{Theorem}[section]
\newtheorem{lemma}[theorem]{Lemma}
\newtheorem{proposition}[theorem]{Proposition}
\newtheorem{corollary}[theorem]{Corollary}
\newtheorem{definition}[theorem]{Definition}
\newtheorem{remark}[theorem]{Remark}

\numberwithin{equation}{section}

\begin{document}
\title[Future stability of $1+3$ Milne model for EKG]{Future stability of the $1+3$ Milne model for the Einstein-Klein-Gordon system}

\author[J. Wang]{Jinhua Wang} \email{wangjinhua@xmu.edu.cn}
\address{School of Mathematical Sciences, Xiamen University, Xiamen 31005, China}

\begin{abstract}
We study the small perturbations of the $1+3$-dimensional Milne model for the Einstein-Klein-Gordon (EKG) system. We prove the nonlinear future stability, and show that the perturbed spacetimes are future causally geodesically complete.  For the proof, we work within the constant mean curvature (CMC) gauge and focus on the $1+3$ splitting of the Bianchi-Klein-Gordon equations. Moreover, we treat the Bianchi-Klein-Gordon equations as evolution equations and establish the energy scheme in the sense that we only commute the Bianchi-Klein-Gordon equations with the spatially covariant derivatives while the normal derivative is not allowed. 
\end{abstract}
\maketitle
\tableofcontents

\section{Introduction}\label{sec-intro}

\subsection{The EKG system}\label{sec-einstein-kg-eq}
The Einstein-Klein-Gordon (EKG) system takes the form of
\begin{subequations}
\begin{equation}
\breve R_{\alpha \beta} - \frac{1}{2} \breve R \breve g_{\alpha \beta} = \mathcal{T}_{\alpha \beta}(\phi),   \label{eq-Einstein-source}
\end{equation}
\begin{equation}
 \mathcal{T}_{\alpha \beta}(\phi)= \breve D_\alpha \phi \breve D_\beta \phi - \frac{1}{2} \breve g_{\alpha \beta} \left( \breve D^\mu \phi \breve D_\mu \phi + m^2 \phi^2 \right), \label{def-energy-Mom-kg}
 \end{equation}
 \end{subequations}
where $\breve R_{\alpha \beta}$ and $\breve R$ denote the Ricci and scalar curvature of an unknown Lorentzian
metric $\breve g_{\alpha \beta}$ respectively, and $\mathcal{T}_{\alpha \beta}(\phi)$ is the energy momentum tensor for a massive scalar field $\phi$. The Bianchi identities 
imply that the scalar field $\phi$ satisfies the Klein-Gordon (KG) equation
\begin{equation} \label{eq-kg}
\Box_{\breve g} \phi - m^2 \phi = 0. 
\end{equation}
$\breve D$ denotes the covariant derivative associated to $\bar g_{\alpha \beta}$ and the geometric wave operator is $\Box_{\breve g} = \breve D_\alpha \breve D^\alpha$. We use $m$ to denote the mass of the KG field.

Our first motivation for studying the EKG system comes form the problem of the nonlinear stability of the supersymmetric backgrounds posed by S.T. Yau. Consider the Kaluza-Klein spacetime of the form $\mathcal{M} = \mathbb{R}^{1+3} \times S^1$. Then the five-dimensional metric depending on $S^1$ periodically can be Fourier-expanded. The projections of the five-dimensional Einstein equations onto the zero (massless) modes and the non-zero (massive) modes give rise to a wave-Klein-Gordon system. Subject to the zero mode perturbation, both of the work by Wyatt \cite{Wyatt-17}, and Branding-Fajman-Kroencke \cite{Branding-Fajman-Kroencke-18} have recently contributed on this aspect. \cite{Wyatt-17} is related to the nonlinear stability of the Kluza-Klein spacetime with the base on the Minkowski space, while \cite{Branding-Fajman-Kroencke-18} is concerning with the one base on the Milne model. There is also related result by Choquet-Bruhat and Moncrief \cite{Bruhat-Moncrief-U1} who establish the stability of $U(1)$-symmetric spacetimes to $1+3$-dimensions.

The backgrounds we consider are the following family of cosmological vacuum spacetimes. Let $\bar M$ be a $4$-manifold of the form $ I \times \Sigma,$ where $I \subset \mathbb{R}$ is an interval,  $\Sigma$ is a compact $3$-manifold admitting an Einstein metric $\gamma$ with negative Einstein constant $\lambda$. We choose $\lambda =-2$ so that the Einstein metric $\gamma$ is a hyperbolic metric with sectional curvature $-1$. Then $(\bar M, \bar\gamma)$ with $\bar\gamma$ given by $$\bar\gamma = -dt^2 + t^2 \gamma$$ is a solution to the vacuum Einstein equations and known as a $1+3$-dimensional \emph{Milne model}. Such model undergoing accelerated expansion in the future direction is locally isometric to the $k=-1$ vacuum Friedmann-Lem\^{i}atre-Robertson-Walker (FLRW) model. Andersson and Moncrief \cite{A-M-04} first consider the nonlinear stability of this model (called hyperbolic cone spacetimes there), and show that for the constant mean curvature (CMC) Cauchy data for the vacuum Einstein equations close to the standard data for the hyperbolic cone spacetimes, the maximal future Cauchy development is globally foliated by the CMC Cauchy surfaces and causally geodesically complete to the future. This also motivates the current project. We refer to \cite{Anderson-CMC, Andersson-CMC-02, Andersson-04-cmc, Gerhardt-CMC} and references therein for more backgrounds on the existence of CMC foliations.

There are several results concerning the decay of the KG field \cite{Katayama-wave-kg, Klainerman-KG-85, Shatah-KG, Sideris-KG}. In contrast to the wave equation, the KG equation is not conformal invariant and hence does not commute well with the symmetry of scaling. Due to this fact, Klainerman \cite{Klainerman-KG-85} employs the hyperbloid foliations which respect the Lorentz invariance of the KG operator and derives the asymptotic behavior of the KG field by the energy method alone. Furthermore, Katayama \cite{Katayama-wave-kg} somehow overcomes the incompatibility between the wave and KG field. Other method including Fourier analysis for the wave-Klein-Gordon system is established by Ionescu-Pausader \cite{Ionescu-Pausader-17}. Global solutions for the semilinear KG equations in the FLRW spacetimes are studied by Galstian-Yagdjian \cite{Galstian-Yagdjian}.

Let us review some nonlinear stability results of the Minkowski spacetime. Christodoulou and Klainerman \cite{Christodoulou-K-93} initiate a covariant proof based on the Bel-Robinson energy. Their proof relies upon the geometric foliations of spacetime, including maximal $t = \text{const}$ slices and a family outgoing null cones, and null condition hidden inside the Bianchi equations. Lindblad and Rodnianski \cite{Lind-Rod-05, Lind-Rod-10} devise another proof based on the wave coordinate gauge, under which the Einstein equations satisfy the weak null condition. This method is extended to the case of massive Einstein-Vlasov system by Lindblad and Taylor \cite{Lind-Taylor-17}.  All of these works require the full symmetries of the Minkowski spacetime, including the conformal symmetries, suggesting that these methods can not apply to the EKG system straightforwardly. Recently, inspired by \cite{Katayama-wave-kg} and \cite{Lind-Rod-05, Lind-Rod-10}, Ma and LeFloch \cite{Ma-Lefoch-16, Ma-Lefloch-EKG} make use of the hyperboloid method to address an energy argument for both of the wave equation and KG equation uniformly, and then prove the nonlinear stability for the EKG system. In addition to the $L^\infty-L^\infty$ estimate for the wave equations \cite{Katayama-wave-kg, Lind-Rod-05, Lind-Rod-10}, Ma and LeFloch \cite{Ma-Lefoch-16} also demonstrate an $L^\infty-L^\infty$ estimate for the KG equation, which is crucial in dealing with the nonlinear couplings \cite{Ma-Lefloch-EKG}. Via this hyperboloid method, the massive Einstein-Vlasov system is alternatively understood by Fajman, Joudioux and Smulevici \cite{F-J-S-17}. As the KG field and the massless Vlasov filed are lack of the conformal symmetries as well. Taylor \cite{Taylor-17} gives a proof of the stability for the massless Einstein-Vlasov system with compact support assumptions. The proof in fact reduces to a semi global problem since the matter is shown to be supported in a strip going to null infinity.
 
The works of Andersson and Moncrief \cite{A-M-04, A-M-11-cmc} also constitute nonlinear stability results for the vacuum Einstein equations. Compared to \cite{A-M-04}, the backgrounds considered in \cite{A-M-11-cmc} are extended to a family of $1+n$-dimensional Milne model. The proof of \cite{A-M-04} is based on the Bel-Robinson energy and its higher-order generalization. Our proof possesses some similarities to that of \cite{A-M-04} as far as we both treat the Bianchi equations as evolution equations to derive the energy estimates for the Weyl field.  While in \cite{A-M-11-cmc}, the authors turn to the Einstein evolution equations and use a wave equation type of energy introduced in \cite{A-M-03-local} for the energy estimates. Related to this, Andersson and Fajman \cite{Andersson-Fajman-17} show the future stability of the $1+3$-dimensional Milne model for the massive Einstein-Vlasov system. We remark that the work of \cite{A-M-04, A-M-11-cmc, Andersson-Fajman-17} all base on the constant mean curvature, spatially harmonic (CMCSH) gauge, with which the local existence theorem is proved in \cite{A-M-03-local}.  

In this paper, we are concerning the nonlinear stability of the $1+3$-dim Milne model for the EKG system. We work within the CMC gauge with zero shift. The local existence theorem then follows by adjusting the argument in \cite{Christodoulou-K-93} where the maximal gauge with zero shift is used. Our main difficulty lies in the fact that the KG field is not scale invariant and hence scaling does not qualify as a commutator vector field. We now recast the KG equation in a first order form and focus on the $1+3$ form of the Bianchi-Klein-Gordon equations. More precisely, we split the normal and tangential derivatives in the formulae of Bianchi-Klein-Gordon equations, treating them as evolution equations, and propose an energy argument in which the high order energies contain only high order of spatially covariant derivatives. Practically, we commute only the spatially covariant derivatives with the Bianchi-Klein-Gordon equations in $1+3$ from. As a result, the high order derivatives of the KG field along the normal direction possess no good estimates, and hence normal derivative is referred as a bad derivative. Due to this covariant proof, we avoid some unnecessarily nonlinear couplings that will arise when working in coordinates (see \cite{Ma-Lefloch-EKG} for the Minkowski case). Meanwhile, we find that the ``genuine'' coupling structures admit some cancellations, such that in the Bianchi equations, the nonlinear couplings involve no the high order of bad (normal) derivatives of the KG field. On the other hand, the difficulty does not diminish because if we proceed to the top order energy estimates, some nonlinear couplings become borderline terms. Another difficulty comes from the resulted borderline terms when commuting the spatially covariant derivatives with the KG equation. Fortunately, these borderline terms exhibit linearizable feature. Motivated by \cite{Christodoulou-08} (see also \cite{Miao-Pei-Yu, Wang-Wei-17, Wang-Yu-13, Wang-Yu-16}), we resolve these difficulties by the method of linearization and the hierarchy of energy. In particular, the $L^\infty-L^\infty$ estimate for the KG equation proposed in \cite{Ma-Lefoch-16} turns out to be helpful in the procedure of linearization. Besides, we reveal the subtle relation between the lapse and the KG field, with which we manage to improve the estimates step by step and finally reduce the nonlinear borderline terms into linear ones. 
Alternatively, if we use the CMCSH gauge, no new difficulty occurs and our method works as well. 

\subsection{Main result}\label{sec-main-result}
\begin{theorem}\label{main-thm-into}
Let $\Sigma$ be a compact hyperbolic $3$-manifold without boundary. Assume that $(\Sigma, \tilde g_0, \tilde k_0, \phi_0, \phi_1)$ is a CMC  data set satisfying the constraint equations of the EKG system \eqref{eq-Einstein-source}-\eqref{def-energy-Mom-kg}, \eqref{eq-kg}. 
Let $-3t_0^{-1}= \tau_0= \tr \tilde k_0 = \text{constant} <0,$ with  $t_0  > \max \{1, 9 m^{-1} \}$. Then there is an $\varepsilon > 0,$ so that if 
\begin{equation}\label{intro-initial-data}
\begin{split} 
&\|\tilde R_{0ij} + 2 t_0^{-2} \tilde g_{0ij} \|_{H^3(\Sigma, \tilde g_0)} + \| \tilde k_0 + t_0^{-1} \tilde g_{0ij}  \|_{H^4(\Sigma, \tilde g_0)}  \leq \varepsilon, \\
& \|\phi_0\|_{H^5(\Sigma, \tilde g_0)} + \|\phi_1\|_{H^4(\Sigma, \tilde g_0)}  \leq \varepsilon,
\end{split}
\end{equation}
where $\tilde R_{0ij}$ denotes the Ricci curvature of $\tilde g_0$, the maximal development $(\bar M, \breve g)$ of the EKG data set $(\Sigma, \tilde g_0, \tilde k_0, \phi_0, \phi_1)$ has a global CMC foliation in the expanding direction (to the future of $(\Sigma, \tilde g_0)$ in CMC time $\tau = \tr \tilde k$) and  $(\bar M, \breve g)$ is future causally geodesically complete.
\end{theorem}
\begin{remark}[The existence of CMC data]\label{rk-cmc}
We also refer to \cite{Constrain-Bruhat-Isenberg} for results concerning the existence of initially CMC data in the current setting. We assume $|\phi_0|_{L^\infty(\Sigma)}$ to be small enough, so that $\left( \frac{2}{3} \tau_0^2 - 2m^2 \phi_0^2\right) > 0$ (that is, $\mathcal{B}_{\tau_0, \phi_0} >0$ by the convention in \cite{Constrain-Bruhat-Isenberg}). In our case, $\Sigma$ has negative Yamabe invariant and $\mathcal{B}_{\tau_0, \phi_0} >0$, then the existence CMC  solutions of the constraint equations for the EKG system on such compact manifolds is covered by \cite{Constrain-Bruhat-Isenberg}.
\end{remark}

\begin{remark}[Generalize to non-CMC data]\label{rk-non-cmc}
Considering generalized data, which is not necessarily CMC, but close to the initially CMC data induced by the background solution, the existence of a maximal globally hyperbolic development of this initial data is assured in \cite{Ringstrom-09}. Under the smallness condition, the existence of a CMC surface in such a development is shown in the vacuum case by Fajman-Kroencke \cite{Fajman-Kroencke-15} (Theorem 2.9). Based on these results, we can now generalize our main theorem to the case of non-CMC data.

The rescaled version of data (non-CMC) \eqref{intro-initial-data} is $\|R_{0ij} + 2 g_{0}\|_{H^3(\Sigma, g_0)} + \|k_0+g_0\|_{H^{4}(\Sigma, g_0)} \lesssim \varepsilon$ (c.f. \eqref{rescale-metric-1}-\eqref{def-rescale-mass} for the rescaling).
To match the condition in \cite{Fajman-Kroencke-15}, one needs the smallness for $\|g_0-\gamma\|_{H^s(\Sigma, \gamma)} + \|k_0+\gamma\|_{H^{s-1}(\Sigma, \gamma)}$, $s>3/2+1$, where $\gamma$ is the hyperbolic metric.  We take $s=5$ and choose the harmonic coordinates with respect to $g_0$ and $\gamma$. Hence, in our setting, it involves controlling $\|g_{0ij}-\gamma_{ij}\|_{H^5(\Sigma, \gamma)}$ in terms of $\|R_{0ij} + 2 g_{0ij}\|_{H^3(\Sigma, g_0)}$. Denote by $\mathcal{A} = 2{\bf D}(R_{0ij} + 2g_{0ij})$, where ${\bf D}$ is the Frechet derivative. Recall that from \cite{A-M-04, A-M-11-cmc}, for a TT (transverse and traceless) tensor $h$,  $$\mathcal{A} h= -\Delta_{\gamma} h -2h.$$
Let $\lambda$ be the eigenvalues of $\mathcal{A}$. It is known that, the hyperbolic metric $\gamma$ ($n\geq 3$) is strictly stable, i.e. $\lambda >0$ (\cite{Besse-Einstein} $\S$ 12 {\bf H}). The operator $\mathcal{A}$ is an isomorphism.
Thus the smallness of $\|R_{0ij} + 2 g_{0ij}\|_{H^3(\Sigma, g_0)}$ implies smallness for $\|g_{0ij}-\gamma_{ij}\|_{H^5(\Sigma, \gamma)}$ and hence the rescaled version of \eqref{intro-initial-data} entails that $\|g_{0ij}-\gamma_{ij}\|_{H^5(\Sigma, \gamma)} + \|k_{0ij}+\gamma_{ij}\|_{H^{4}(\Sigma, \gamma)} \lesssim \varepsilon$. Now we can prove along the lines of the corresponding argument presented in \cite{Ringstrom-09, Fajman-Kroencke-15} to show the existence of a CMC surface in the non vacuum setting. Therefore, the data in Theorem \ref{main-thm-into} can be generalized to be non-CMC.
\end{remark} 
 
Before offering an overview for the proof, we introduce some notations.
Let $(\bar M, \breve g)$ be the maximal development of the EKG data given above.  $(\bar M, \breve g)$ is differmorphic to $\mathbb{R} \times \Sigma,$ where $\Sigma$ is a compact Cauchy surface and it is globally foliated by a CMC foliation $\{\Sigma_\tau, \tau \in I\}$. We rescale the CMC time function and define $t= -\frac{3}{\tau}.$ Then as $\tau \rightarrow 0^-, t \rightarrow +\infty$. Let $\p_t$ denote the coordinate vector field corresponding to $t$, then $\breve g(\p_t, \p_t) = -N^2$, where $N$ is called lapse, and we set $\hat N = N-1$. Let $\tilde g_{ij}$ be the induced Riemannian metric on $\Sigma_t$, $\tilde\nabla$ the corresponding covariant derivative. $T$ is the unit timelike vector field normal to $\Sigma_t$, and we have $\p_t = NT$. The second fundamental form is denoted by $\tilde k_{ij}$, while $\hat{\tilde k}_{ij}$ being the traceless part of $\tilde k_{ij}$. In the CMC gauge, $\tr \tilde k = \tau = -\frac{3}{t}.$  

Let $g_{ij}=t^{-2} \tilde g_{ij}$ be the spatially normalized metric, $\nabla$ the corresponding connection. 
As shown in \cite{A-M-04} ((4.20a)-(4.20c)), the following variables are scale-free (the indices refer to a coordinate frame).
\begin{equation}\label{rescale-metric-1}
\begin{split}
g_{ij} &=t^{-2} \tilde g_{ij} \quad\, g^{ij}=t^{2} \tilde g^{ij}  \\ 
k_{ij} &=t^{-1} \tilde k_{ij} \quad  \tr k = t \tr \tilde k\\
\Tbar &= t T \quad \quad\, \di \mu_g =t^{-3} \di\mu_{\tilde g}.
\end{split}
\end{equation}
It is convenient to introduce the logarithmic time $\bar \tau = \ln t$, which has the property that
$\p_{\bar \tau} = t\p_t$ is scale free.
We also define the rescaled mass of the KG field.
\begin{equation}\label{def-rescale-mass}
\bar m(t) = tm.
\end{equation}
We remark that such rescaling for the mass of the KG field is transitional, aiming to make the KG energy well-behaved with respect to the above scale-free variables.
Finally, we let $\phi$ be the KG field, and $E_{ij}, H_{ij}$ be the electric and magnetic parts of the Weyl field $W_{\mu \alpha \nu \beta}$. Note that, $E_{ij}, H_{ij}$ are scale-free as well.

\subsection{Strategy of the proof}\label{sec-main-proof}
{\bf Energy scheme}.   As is explained before, the scaling $\dtau$ (and hence $\Tbar$) does not qualify as a commutator vector field for the KG equation. Therefore, we only commute the KG equation with the spatial covariant derivative $\nabla$. Technically, we address an $1+3$ splitting for the KG equation, which takes the form of
\begin{align*}
& \lie_{\dtau} \Tbar \phi - \left( \tr k N +1 \right) \Tbar \phi  -N \nabla^i \nabla_i \phi + \bar m^2(t)  N \phi =  \nabla_i N  \nabla^i \phi.
\end{align*}
$\lie$ means taking Lie derivative.
Thus in such $1+3$ forms, the high order of the KG equation \eqref{eq-kg-N-l-1+3-general-simply} can be derived by calculating the commutator $[\lie_{\dtau}, \nabla]$ (c.f. Section \ref{sec-commuting-id}), which involves only $N$, $\hat k$ and their spatial derivatives. The Bianchi equation are treated in an analogous fashion. We then establish the energy inequalities for these $1+3$ evolution equations (c.f. Section \ref{sec-ee-kg-1+3} and \ref{sec-ee-Bianchi-1+3}). The energies are defined as below: Fixing an integer $I$, for $t>t_0 > 3 m^{-1}$, the energy norms for $\phi$ and $W$ ($E, H$) are
\begin{equation}\label{intro-def-energy-l-homo-kg}
E_l (\phi, t) = \|\nabla^l \Tbar \phi \|^2_{L^2(\Sigma_t)} +  \|\nabla^l \nabla \phi \|^2_{L^2(\Sigma_t)} +  \|\nabla^l (\bar m(t) \phi)\|^2_{L^2(\Sigma_t)}, \, l \leq I,
\end{equation}
\begin{equation}\label{intro-def-energy-l-homo-Weyl}
E_l (W, t) =   \|\nabla^l E\|^2_{L^2(\Sigma_t)}  + \|\nabla^l H \|^2_{L^2(\Sigma_t)}, \quad l \leq I-1.
\end{equation}
The notation $\nabla^l$ means taking the derivative $\nabla$ $l$ times. 

In other words, we give up the standard method of the energy-momentum tensor. Alternatively, we recast the KG equation in a first order form and treat it as a transport equation rather than a wave equation. This turns out to be more compatible with the energy scheme which commuting merely with the spatial derivatives.

{\bf The coupling structures}. In the strategy of the $1+3$ splitting, the normal derivative $\Tbar$ is referred as a bad derivative, while the spatial derivative $\nabla$ is a good derivative. This is implied by the following facts: The energies \eqref{intro-def-energy-l-homo-kg} do not control $\Tbar^2\phi$. However, making use of the KG equation which expresses $\Tbar^2 \phi$ in terms of $\Tbar\phi$ and $\nabla^l \phi, \, 0 \leq l \leq 2$, in turn yields almost $t^{\frac{1}{2}}$ growth for $\Tbar^2\phi$. On the other hand, implied by Sobolev inequalities, terms involving the spatial derivatives $\nabla^i \Tbar \phi, \nabla^j \phi$, $i\leq I-2, j \leq l-1$ decay almost like $t^{-\frac{1}{2}}$. Meanwhile, $\nabla^l \phi \sim t^{-1} \nabla^l (\bar m (t) \phi)$, $l\leq I-2$ decay almost like $t^{-\frac{3}{2}}$.

One of our key observations is that the nonlinear couplings avoid $\Tbar^2\phi$, due to some cancellations (see \eqref{stru-source-J-bianchi-0} and \eqref{stru-source-AJ-bianchi-0-1}). Therefore, we get couplings such as $\nabla \Tbar \phi \nabla \phi E$ and $\nabla^2\phi \Tbar\phi E$, which generally decay fast enough.  Now, let us pay more attention to the estimates for $\nabla^2 \phi \Tbar\phi E$ and its higher order versions $\nabla^l (\nabla^2 \phi \Tbar \phi E), l \leq I-1$. Notice that, for the lower order cases: $l \leq I-2$, $\|\nabla^l \nabla^2 \phi\|_{L^2(\Sigma_t)}$ is estimated as $t^{-1} \|\nabla^{l+2} (\bar m(t) \phi)\|_{L^2(\Sigma_t)}$. 
Unfortunately, such estimate is forbidden for the top order case: $l=I-1$, since it requires one more derivative than our purposes. That is to say,  $\|\nabla^{I-1} \nabla^2 \phi\|_{L^2(\Sigma_t)}$ has no additional $t^{-1}$ factor. As a result, the top order coupling $\nabla^{I-1} (\nabla^2 \phi \Tbar \phi E)$ will provide a borderline term $\nabla^{I-1}\nabla^2 \phi \Tbar\phi E$ which requires more work.

{\bf Hierarchy of energy}.
We start with some weak bootstrap assumptions (with $\delta (0<\delta < \frac{1}{6})$ loss in the decay rate of $t$), and finally close the bootstrap argument by sharpening the estimates (Section \ref{sec-close-b-t}). In particular, the estimate for $\|\hat k\|_{L^2(\Sigma_t)}$ is retrieved by the transport equation for $\hat k$ (Section \ref{sec-closing-bt-k}).

When commuting $\nabla^l$ $(1 \leq l \leq I)$ with the KG equation, we encounter the borderline terms $BL_l$, taking the forms of $BL_l = \sum_{a=1}^l BL_l^a$, with $1 \leq a \leq l$ and
\begin{equation}\label{def-BLl-123l}
\begin{split}
BL_l^a = & \bar m (t) \nabla^{l-a+1} N * \nabla^l (\bar m (t) \phi) * \nabla^{a-1} \bar T \phi \\
&\pm \bar m (t) \nabla^{l-a+1} N *  \nabla^l \Tbar \phi * \nabla^{a-1}( \bar m (t) \phi ).
\end{split}
\end{equation}
We confirm the inevitability of these borderline terms in Remark \ref{rk-borderline-term}. Meanwhile, we should notice that $BL_l$ will vanish if the scalar field is massless.

Because of the weak bootstrap assumptions, $BL_l, 1 \leq l \leq I$ bear $\delta$ loss in the decay rate.  We propose an $L^\infty - L^\infty$ estimate for the KG field, which is due to the technical ODE estimate in \cite{Ma-Lefoch-16}, to improve $\|T\phi\|_{L^\infty}$ and $\|\phi\|_{L^\infty}$ and then retrieve the $\delta$ loss in $BL_l, 1 \leq l \leq I$. But this costs the regularity: the $L^\infty - L^\infty$ estimate requires the strongest decay for $\nabla^2 \phi$, which further forces us to perform two more orders of energy estimates. That is, in \eqref{intro-def-energy-l-homo-kg}-\eqref{intro-def-energy-l-homo-Weyl}, we take $I=4$ rather than $I=2$.

We observe the linear structure in \eqref{def-BLl-123l}: $\nabla^{a-1} \phi$ is of lower order (less than $l^{\text{th}}$ order) and also $\nabla^{l-a+1} \hat N$ depends mainly on lower order energies of the KG field. In more details, the refined estimates for the lapse (Lemma \ref{lem-N-ineq-source}) tell that  $\|\nabla^a \hat N \|_{L^2(\Sigma_t)}, a \leq 2$ relies crucially on the zeroth order energy $E_0(\phi, t)$, and $\|\nabla^{a} \hat N\|_{L^2(\Sigma_t)}, 2 < a \leq l$ is primarily related to $E_{a-2}(\phi, t)$, which is again of lower order. Enlightened by these features, we proceed by an argument of hierarchy of energy estimates, which are carried out in Section \ref{sec-0-ee-kg} (for the KG field). The idea is that, the information already obtained for lower order energies will help to sharpen the estimates for $\nabla^{l-a+1} \hat N$, $\nabla^{a-1} \phi$ and reduce the nonlinear structure of  \eqref{def-BLl-123l} to a linear one. In addition, this idea of linearization also applies to the borderline terms (like $\nabla^{I-1}\nabla^2 \phi \Tbar \phi E$) that arise in the top order energy estimate for the Weyl field, since we can now take advantage of the improved estimates for $\phi$, c.f. \eqref{estimate-Sl11-2-1}.

We here sketch the linearization in the energy argument. Generically, the energy inequality for the KG field before linearization takes the form of
\begin{align*}
 \p_{\bar\tau} E_l(\phi, t) + E_l(\phi, t)  \lesssim   t^{-1} E_l(\phi, t) + t  \mathcal{E}_4(\phi, t) \mathcal{E}^{\frac{1}{2}}_4(\phi, t) E^{\frac{1}{2}}_l(\phi, t),
\end{align*}
where the second term on the right hand side is a rough bound for the borderline term \eqref{def-BLl-123l}, and $\mathcal{E}_4(\phi, t)$ means the energies of the KG field up to the fourth order. Here the first quadratic term $ \mathcal{E}_4(\phi, t)$ is due to the fact that $\nabla^i N$ depends quadratically on $\phi$ (c.f. Lemma \ref{lem-N-ineq-source} and Remark \ref{rk-rough-lapse}). We aims to linearize the borderline terms via the hierarchy of energy estimate and the $L^\infty - L^\infty$ estimate for the KG field, so that the energy inequality is improved as (c.f. Section \ref{sec-0-ee-kg})
\begin{align*}
\p_{\bar\tau} E_l(\phi, t) + E_l(\phi, t)  \lesssim & t^{-1}  E_l(\phi, t) + t^{-\frac{1}{2}} \varepsilon M \varepsilon I E^{\frac{1}{2}}_l(\phi, t),\\
  \lesssim & \left( t^{-1}+ \varepsilon M \right) E_l(\phi, t) +  \varepsilon^2 I^2 t^{-1}.
\end{align*}
Here $I$ is a constant depending on the initial data. In particular, the previously quadratic term is now replaced by a linear one. Changing to $t$, we have 
\begin{equation*}
\begin{split}
& \dt ( t  E_l(\phi, t)) \lesssim \left( t^{-2} + \varepsilon M t^{-1} \right) (t  E_l(\phi, t)) +  \varepsilon^2 I^2 t^{-1}.
\end{split}
\end{equation*} 
The Gr\"{o}nwall inequality then yields that $tE_l(\phi, t) \lesssim  \varepsilon^2 I^2 t^{C_M \varepsilon}$, where $C_M$ depends linearly on $M$. Then the energy argument are closed.

The paper is organized as follows. In Section \ref{preliminaries}, we describe some conventions and notations, and show the background materials. In Section \ref{sec-energy-estimate-1+3}, we conduct the $1+3$ splitting of the EKG system and establish the energy scheme.  In Section \ref{sec-preliminary-estimate}, we introduce the bootstrap assumptions, and present some preliminary estimates that follows from the Sobolev inequalities. The energy estimates for the KG field and the Weyl field are carried out in Section \ref{sec-energy-kg} and \ref{sec-energy-Bianchi}. Section \ref{sec-global-existence} is devoted to close the bootstrap argument and prove the global existence theorem, geodesic completeness. Finally, a number of useful definitions, identities and ODE estimates are collected in the appendix.

{\bf Acknowledgements}: The Einstein-Klein-Gordon project studied in this paper was suggested by L. Andersson. I am grateful to him
for many enlightening discussions. I thank V. Moncrief for many helpful discussions on this topic and suggesting the usage of CMC gauge in early 2014. I thank C. Liu and S. Ma for reading and useful comments as well.
The author was supported by a Humboldt Foundation post-doctoral fellowship at the Albert Einstein Institute during the period 2014-16, when part of this work was done.  She is also supported by NSFC (Grant No. 11701482), NSF of Fujian Province (Grant No. 2018J05010) and Fundamental Research Funds for the Central Universities (Grant No. 20720170002).

\section{Preliminaries}\label{preliminaries}
\subsection{Notations and conventions}\label{notations}
Indices:
\begin{itemize}
\item $\alpha, \beta, \cdots, \mu, \nu, \cdots$: Greek indices (spacetime indices) run over $0,\cdots,3.$
\item $i,j, \cdots, p,q, \cdots$: Latin indices (spatial indices) run over $1,\cdots,3.$
\end{itemize}
Geometry related:
\begin{itemize}
\item $(\bar M, \bar{g}_{\mu\nu}, D):$ the spacetime manifold with $D$ the connection.
\item $\left( \Sigma_t, g_{ij}, \nabla \right):$ the induced spatial manifold with $\nabla$ the connection.
\item $\Box_{\bar g}=D^\alpha D_{\alpha}:$ the spacetime wave operator. 
\item $\Delta = \nabla^i \nabla_i$: the spatial Laplacian operator.
\item $\di \mu_{\bar g}, \di \mu_g$: the volume form with respect to $\bar g$ or $g$. 
\item $\bar R_{\mu \alpha \nu \beta}, \bar R_{\mu \nu}, \bar R:$ the components of the Riemannian, Ricci, scalar curvature tensor of $\left( \bar{M}, \bar g_{\mu \nu} \right).$
\item $R_{imjn}, R_{ij}, R:$ the components of the Riemannian, Ricci, scalar  curvature tensor of $\left( \Sigma_t, g_{ij} \right).$
\item $\bar R_{imjn}, \bar R_{ij}:$ projection of $\bar R_{\mu \alpha \nu \beta}$, $\bar R_{\mu \nu}$ to $(\Sigma_t, g_{ij}).$
\end{itemize}
Foliation related:
\begin{itemize}
\item $\{T, e_i\}:$ $T$ denotes the future directed, unit vector normal to $\Sigma_t,$ and $\{ e_i \}$ an orthonormal frame tangent to $\Sigma_t$.  
\item $\tau, t, \bar \tau:$ $\tau$ is the CMC time $\tau = \tr k$; $t=-\frac{3}{\tau}$; $\bar \tau = \ln t$.
\item $\p_t, \dtau:$ coordinate vector field corresponding to $t, \bar \tau$.
\end{itemize}
Simplified conventions
\begin{itemize}
\item $\nabla_{I_l} \psi$: the $l^{\text{th}}$ order of covariant derivative $\nabla_{i_1} \cdots \nabla_{i_l}\psi$ and the multi index $I_l=\{ i_1 \cdots i_l\}$ is used.
\item $\nabla_{i_1} \cdots \stackrel{j}{\nabla}_p \cdots \nabla_{i_l} \psi $: $\nabla_{I_l} \psi$ with the $j^{\text{th}}$ $\nabla_{i_j}$ being replaced by $\nabla_p$.
\item $\D \phi$: any element in $\{\Tbar \phi, \nabla \phi, \bar m (t) \phi\}$.
\item $f_1 \lesssim  f_2$: $f_1 \leq C f_2$ with some {\it universal} constant $C$.
\item $f_1 \sim f_1$: $f_1$, $f_2$ are equivalent in the sense that $f_1 \lesssim f_2$ and $f_2 \lesssim f_1$.
\item $C_M$: a constant depending linearly on some controlling quantity $M$.
\item $A*B$: any finite sum of products of $A$ and $B$, with each product being a contraction (with respect to $g$) between two $\Sigma_t$-tensors $A$ and $B$\footnote{In the estimates, we will only employ the formula $\|A * B\| \lesssim \|A\| \|B\| $ which allows ourselves to ignore the detailed product structure at this point. $\|\cdot \|$ denotes the norm associated to $g.$}.
\end{itemize}

\subsection{The background materials}\label{sec-sec-background}
 Let $\Sigma$ be a spacelike hypersurface and let $\bar{M}= I \times \Sigma$ be an $n+1$-dim manifold with Lorentz metric $\breve{g}$ of signature $-\, + \cdots +$. We introduce local coordinates $(\tau, x^i, i=1, \cdots, n)$ on $\bar{M}$ so that $\tau$ is a time function and $x^i$  are coordinates on the level sets $\Sigma_\tau$ of $\tau$. 
 
 Let $\partial_\tau=\partial /\partial \tau$ be the coordinate vector fields corresponding to $\tau$. The lapse function ${}_{\tau}N$ and shift vector field $X$ of the foliation $\{\Sigma_\tau\}$ are defined by $\partial_\tau={}_{\tau}NT+\tilde X_\tau,$ where $T$  is the unit timelike vector field normal to $\Sigma_\tau$. Assume $T$ is future directed so that ${}_{\tau}N>0$.  We make the gauge choice $\tilde X_\tau=0,$ so that
the spacetime metric $\breve{g}$ with vanishing shift takes the form
\begin{equation}\label{metric-form-tau}
\breve g_{\mu\nu} = -{}_{\tau}N^2 \di \tau^2 + \tilde g_{ij} \di x^i \di x^j.
\end{equation}
The second fundamental form $\tilde k_{ij}$ of $\Sigma_\tau$ is given by $\tilde k_{ij} = - \frac{1}{2}\lie_{T} \tilde g_{ij}.$
The {\bf CMC} gauge is referred to as
\begin{equation}\label{cmc-gauge-tau}
\tr \tilde k = \tau.
\end{equation}
We call $\tau$ the CMC time function.
For notational convenience, we define a rescaled time function $t:=-\frac{3}{\tau}$ so that the corresponding lapse ${}_{t} N$
\begin{equation}\label{def-dt}
{}_{t}N =  \frac{3}{t^2}  {}_{\tau} N, \quad \partial_t = {}_{t}N T.
\end{equation}
We will drop the subscript $t$, and denote the lapse for the $t$-foliation by $N$. We also let
\begin{equation}\label{def-hat-N}
\hat N = N-1.
\end{equation}
In terms of the $t$-foliation,
\begin{equation}\label{k-decomp}
\tr \tilde k = - \frac{3}{t}, \quad \tilde k_{ij} = -\frac{1}{t} \tilde g_{ij} + \hat{\tilde k}_{ij},
\end{equation}
where $\hat{\tilde k}_{ij}$ is the trace free part of $\tilde k_{ij}$. 

Performing the rescaling \eqref{rescale-metric-1}, the rescaled curvature is related to the original one by
\begin{align*}
 R_{imjn} &= t^{-2} \tilde R_{imjn}, \quad  R_{ij } = \tilde R_{ij}, \quad   R = t^2 \tilde R.
\end{align*}
And we define the rescaled spacetime metric as well:
\begin{align*}
\bar g_{\mu \nu} &=t^{-2} \breve g_{\mu \nu}, \quad \, \bar g^{\mu \nu}=t^{2} \breve g^{\mu \nu}, \quad \bar T = t T, \\ 
\bar R_{\mu \alpha \nu \beta} &= t^{-2}  \breve R_{\mu \alpha \nu \beta}, \quad \bar R_{\mu  \nu } = \breve R_{\mu  \nu}, \quad  \bar R = t^2 \breve R.
\end{align*}

We remark that, if $\tau \in [\tau_0, 0)$ with $\tau_0 <0,$ then $t \in [ t_0, + \infty),$ where $ t_0=-\frac{3}{\tau_0}.$ Letting $m$ be the mass of the KG field \eqref{eq-kg}, we require that $|\tau_0| <  m$, i.e. $t_0 > 3 m^{-1}$.

\subsubsection{Lorenzian geometric equations}\label{sec-einstein-eq}
Recall the first and second variation equations for the rescaled metric
\begin{subequations}
\begin{equation}
\mathcal{L}_{\dtau} g_{ij} =2 \hat N g_{ij} -2N \hat k_{ij},  \label{eq-evolution-1}
\end{equation}
\begin{equation}
\mathcal{L}_{\dtau} \hat k_{ij} + \hat k_{ij} = \hat N g_{ij} -\nabla_i \nabla_j N + N \left( \bar{R}_{i \Tbar j \Tbar} - \hat k_{ip} \hat k_{j}^p \right). \label{eq-evolution-2}
\end{equation}
\end{subequations}
Here the indices $i, j$ refer to the frame $\{\p_i, \, i=1,2,3\}.$
And the Gauss-Codazzi equations are
\begin{equation}\label{Gauss-Riem-hat-k}
\begin{split}
R_{imjn} =&-  ( g \odot g )_{imjn}+ ( g \odot \hat k )_{imjn}  - \frac{1}{2} ( \hat k \odot \hat k )_{imjn}   + \bar{R}_{imjn},
\end{split}
\end{equation}
where the Kulkarni-Nomizu product $\odot$ is defined by: Let $\xi, \eta$ be symmetric $(0, 2)$-tensors,
\begin{equation}\label{def-odot}
( \xi \odot \eta )_{imjn} = \xi_{ij} \eta_{mn} - \xi_{in} \eta_{jm} + \eta_{ij} \xi_{mn} - \eta_{in} \xi_{jm}.
\end{equation}
\begin{equation}
\nabla_i \hat k_{jm} -\nabla_j \hat k_{im}= \bar R_{\Tbar mij}. \label{Codazzi-curl-k}
\end{equation}
The trace of the Gauss-Codazzi equations \eqref{Gauss-Riem-hat-k}-\eqref{Codazzi-curl-k} are
\begin{subequations}
\begin{equation}
R_{ij} = \hat k_{il} \hat k_{j}^l + \hat k_{ij} - 2 g_{ij}+ \bar{R}_{i \Tbar j \Tbar} + \bar{R}_{ij},\label{Gauss-Ricci-hat-k}
\end{equation}
\begin{equation}
\nabla^i \hat k_{ij} - \nabla_j \tr k= - \bar R_{\Tbar j}.\label{Codazzi-div-k}
\end{equation}
\end{subequations}
The double trace of the Gauss equations \eqref{Gauss-Riem-hat-k} is
\begin{equation}\label{Gauss-Ricci-trace-hat -k}
R - |\hat k|^2 +  6 = 2\bar{R}_{\Tbar \Tbar} + \bar{R},
\end{equation}

\subsubsection{The background spacetime}
Let $(\Sigma, \gamma)$ be a compact manifold without boundary which if of hyperbolic type, with $\gamma$ being the standard hyperbolic metric of sectional curvature  $-1$. The hyperbolic cone spacetimes (or $1+3$-dimensional Milne model) $(\bar M, \bar \gamma)$ is the Lorentzian cone over $(\Sigma, \gamma)$, i.e.
\begin{equation}\label{def-bar-gamma}
\bar M = (0, \infty) \times \Sigma, \quad \bar \gamma = - \di \rho^2 + \rho^2 \gamma.
\end{equation}
The family of hyperboloids $\Sigma_\rho$ given by $\rho = \text{constant}$ has normal $T= \p_{\rho}.$ The vector field $\rho \p_{\rho}$ is a timelike homothetic killing field. A calculation shows that $\tilde k_{ij} = - \frac{1}{\rho} \tilde g_{ij},$ where $\tilde g_{ij} = \rho^2 \gamma_{ij},$ and the mean curvature $\tr \tilde k =- \frac{3}{\rho}.$ The $\rho$-foliation has lapse $N=1$.

\subsection{Local existence}\label{sec-local}

It is known that in the CMC gauge with zero shift, the Einstein equations are non-strictly hyperbolic. However, following the work of Christodoulou-Klainerman \cite{Christodoulou-K-93}, where they work in maximal gauge with zero shift, we can prove the local existence theorem for EKG. The method is also analogous to the proof based on wavelike coordinates \cite{Bruhat} or CMCSH gauge \cite{A-M-03-local}, in the sense that we will show that to solve the local existence theorem for the original EKG system, it suffices to solve a ‘reduced’ system. Additionally,  this approach is irrelevant to the matter field. We sketch the proof in Appendix \ref{sec-local-cmc}.

\begin{theorem}\label{thm-local-existence}
Assume $\Sigma$ be a compact Riemannian manifold of hyperbolic type. Let $(g_0, k_0, \phi_0, \phi_1)$ be the rescaled CMC data (letting $\tr_{g_0} k_0 = -3$) verifying the constraint equations of EKG system and the following conditions:
\begin{itemize}
\item [1] The Ricci curvature of $g_{0}$, $R_{0ij}$ satisfies $R_{0ij} + 2 g_{0ij} \in H^3(\Sigma, g_0),$
\item [2] $k_0$ is a symmetric $(0,2)$-tensor such that the traceless part verifying $\hat k_0 \in H^4(\Sigma, g_0),$
\item [3] $(\phi_0, \phi_1)$  $\in H^5(\Sigma, g_0) \times H^4(\Sigma, g_0)$.
\end{itemize}
Then there are $0<t_{-} < t_0 < t_{+}$ such that there is a unique, local-in-time smooth development $(\bar M, \bar g)$ of $(g_0, k_0)$, foliated by CMC hypersurfaces, and $\bar M = (t_{-}, t_{+}) \times \Sigma$, $t =t_0$ corresponding to the initial slice $\Sigma$.
\end{theorem}

\section{$1+3$ splitting and the energy scheme}\label{sec-energy-estimate-1+3}

 In the quantitative computations throughout the paper, we employ the conventions: 
 $R_{imjn}$ whose four indices $i, m, j, n$ range over $1, \cdots, 3$ is viewed as a $(0,4)$-tensor (not the special component) on $\Sigma_t.$

\begin{definition}
A $(0, l)$-tensor $\Psi_{\alpha_1 \cdots  \alpha_l}$ is called $T$-tangent if 
\begin{equation}\label{def-T-tensor}
T^\beta \Psi_{\alpha_1\cdots \beta \cdots \alpha_l} =0.
\end{equation}
\end{definition}
If  $\Psi_{\alpha_1 \cdots  \alpha_l}$ is $T$-tangent, then $\lie_{\dtau} \Psi_{\alpha_1 \cdots  \alpha_l}$ and $\lie_{T} \Psi_{\alpha_1 \cdots  \alpha_l}$ are both $T$-tangent. But both $D_T \Psi_{\alpha_1 \cdots  \alpha_l}$ and $D_{\p_t} \Psi_{\alpha_1 \cdots  \alpha_l}$ are no longer $T$-tangent.  Any $T-$tangent tensors can be viewed as tensors on $\Sigma_t$.  Besides, we denote $\nabla_k \Psi_{i_1 \cdots i_n}$ the projection of $\nabla_k \Psi$ on $\Sigma_t$ (evaluated on the corresponding components).  Suppose $\Psi_{I_l}$ is a $(0, l)$-tensor on $\Sigma_t$, then the zero-extension of $\Psi_{I_l}$ becomes $T$-tangent tensor on $\bar M,$ which we will still denote by $\Psi_{I_l}$. The following commuting Lemma \ref{lem-commu-lie} holds for $T$-tangent tensor as well. 

Let $\Gamma_{ij}^a$ be the connection coefficient of $\nabla$, $\dtau$ be the normal vector field defined in Section \ref{sec-sec-background}. Then the Lie derivative $\lie_{\dtau} \Gamma_{ij}^a$ is a tensor field \eqref{dt-Gamma}, 
\begin{equation}\label{dtau-Gamma}
\begin{split}
\lie_{\dtau} \Gamma^a_{ij}& =\frac{1}{2}g^{ab} \left( \nabla_i \lie_{\dtau} g_{jb} + \nabla_j \lie_{\dtau} g_{ib} -\nabla_b \lie_{\dtau} g_{ij} \right).
\end{split}
\end{equation}
Remind that $\lie_{\dtau} g_{ij} = 2 \hat N g_{ij} -2N \hat k_{ij}.$

A commuting identity between $\nabla$ and $\lie_{\dtau}$ is given below.
\begin{lemma}\label{lem-commu-lie}
Let $V$ be an arbitrary $(0, l)$-tensor field on $(\Sigma, g_{ij})$. The following commuting formula is true:
\begin{equation}\label{commuting-lie-nabla}
\begin{split}
\lie_{\dtau} \nabla_j V_{a_1 \cdots a_l} =& \nabla_j \lie_{\dtau} V_{a_1 \cdots a_l} - \sum_{i=1}^l \lie_{\dtau} \Gamma^p_{j a_i} V_{a_1 \cdots p \cdots a_l}.
\end{split}
\end{equation}
\end{lemma}
This lemma can be proved by straightforward calculations.
Besides, we have for any $T-$tangent tensor $V_{a_1 \cdots a_l},$
\begin{equation}\label{commuting-lie-T-nabla}
\begin{split}
\lie_{\bar T} \nabla_j V_{a_1 \cdots a_l} =& \nabla_j \lie_{\bar T} V_{a_1 \cdots a_l} +N^{-1} \nabla_j N \lie_{\bar T} V_{a_1 \cdots a_l}\\
&+ N^{-1} \left( \lie_{\dtau} \nabla_j V_{a_1 \cdots a_l} - \nabla_j \lie_{\dtau} V_{a_1 \cdots a_l} \right).
\end{split}
\end{equation}

\subsection{Commuting identities}\label{sec-commuting-id}
 We shall introduce a commuting Lemma, based on which we can proceed the energy argument.

An application of Lemma \ref{lem-commu-lie} to $\nabla_{I_l}  \psi$ gives the following lemma.
\begin{lemma}[Commuting Lemma for Scalar Field]\label{lemma-commuting-application}
Let $l \geq 1$,  then 
\begin{subequations}
\begin{equation}\label{id-commuting-nabla-T-l-simplify}
\begin{split}
  \lie_{\Tbar}  \nabla_{I_l} \psi & = \nabla_{I_l}  \Tbar \psi  +  \mathcal{K}\mathcal{N}\mathcal{T}_{I_l}(\psi),
\end{split}
\end{equation}
\begin{equation}\label{id-commuting-nabla-N-T-l-simplify}
\begin{split}
&  \lie_{\dtau}  \nabla_{I_l} \psi =  \nabla_{I_l} \left( \dtau  \psi \right) +   \mathcal{K}\mathcal{N}_{I_l}(\psi),
\end{split}
\end{equation}
\end{subequations}
where $\hat{\mathcal{K}}_{I_l}(\psi)$ and $\mathcal{K}\mathcal{N}_{I_l}(\psi)$ are defined as
\begin{subequations}
\begin{equation}\label{def-commuting-KN-l}
\begin{split}
 \mathcal{K}\mathcal{N}_{I_l}(\psi) &= \pm \sum_{a+2 +b  = l}  \nabla_{I_a} \nabla \left( N \hat k + \hat N \right) * \nabla_{I_{b}} \nabla \psi, \,\, l \geq 2, 
\end{split}
\end{equation}
\begin{equation}\label{def-commuting-K-l}
\begin{split}
  \mathcal{K}\mathcal{N}\mathcal{T}_{I_l} (\psi) &= N^{-1} \mathcal{N}_{I_l} ( \Tbar \psi) +    N^{-1} \mathcal{K}\mathcal{N}_{I_l}(\psi), \,\, l \geq 1,
\end{split}
\end{equation}
\end{subequations}
with  $\mathcal{N}_{I_l} (\Tbar \psi)$ defined as below
\begin{equation}\label{def-N-l}
\begin{split}
  \mathcal{N}_{I_l} (\Tbar \psi) &= \sum_{a +b+1  = l} \nabla_{I_a} \nabla N * \nabla_{I_{b}} \Tbar \psi.
\end{split}
\end{equation}
And when $l=1$, $\mathcal{K}\mathcal{N}_{I_1}(\psi)=0$.
\end{lemma}

\begin{proof}
We will prove the following commuting identities.
\begin{equation}\label{commuting-lie-T-I}
\begin{split}
&  \lie_{\dtau} \nabla_{i_1} \cdots \nabla_{i_l} \psi -  \nabla_{i_1} \cdots \nabla_{i_l} \dtau \psi \\
=&- \sum_{a=0}^{l-2} \sum_{m=a+2}^l \nabla_{i_1} \cdots \nabla_{i_a} \left( \lie_{\dtau} \Gamma^p_{i_{a+1} i_m} \nabla_{i_{a+2}} \cdots \stackrel{m}{\nabla}_p \cdots \nabla_{i_l} \psi \right),
\end{split}
\end{equation}
where we use the convention for the right hand side: when $l=1$, it is identically zero;
when $l \geq 2$, the term corresponding to $a=0$ is interpreted as $- \sum_{m=2}^l  \lie_{\dtau} \Gamma^p_{i_{1} i_m} \nabla_{i_{2}} \cdots \stackrel{m}{\nabla}_p \cdots \nabla_{i_l} \psi.$ Namely, the index of $\nabla_{i_1} \cdots \nabla_{i_a}$, $\{i_1 \cdots i_a\}$ is always in the proper order. Otherwise, we let inverted one $\nabla_{i_{k+1}} \nabla_{i_{k}} = 1.$ 

For $l=1,$ by Lemma \ref{lem-commu-lie}, we have
\begin{equation}\label{commuting-id-nabla-T-induction-1}
\begin{split}
\lie_{\dtau} \nabla_{i_1} \psi  = \nabla_{i_1}  \dtau \psi.
 \end{split}
\end{equation}
Hence \eqref{commuting-lie-T-I} is true for $l=1.$ Next, we prove  \eqref{commuting-lie-T-I} by induction. Suppose  \eqref{commuting-lie-T-I} holds for $l\leq n-1,$ we wish to prove that it also holds for $l=n.$ 
By the commuting Lemma \ref{lem-commu-lie},
\begin{equation}\label{commuting-nabla-T-induction-n}
\begin{split}
& \lie_{\dtau} \nabla_{i_1} \nabla_{i_2} \cdots \nabla_{i_n} \psi \\
=&   \nabla_{i_1} \left( \lie_{\dtau} \nabla_{i_2} \cdots \nabla_{i_n} \psi \right) -  \sum_{m=2}^n \lie_{\dtau} \Gamma^p_{i_1 i_m} \nabla_{i_2} \cdots \stackrel{m}{\nabla}_p \cdots \nabla_{i_n} \psi.
\end{split}
\end{equation}
Apply \eqref{commuting-lie-T-I} to $l=n-1,$
\begin{equation}\label{commuting-nabla-T-induction-n-1}
\begin{split}
&  \lie_{\dtau} \nabla_{i_2} \cdots \nabla_{i_n} \psi -  \nabla_{i_2} \cdots \nabla_{i_n} \dtau \psi \\
=&- \sum_{a=1}^{n-2} \sum_{m=a+2}^n \nabla_{i_2} \cdots \nabla_{i_a} \left( \lie_{\dtau} \Gamma^p_{i_{a+1} i_m} \nabla_{i_{a+2}} \cdots \stackrel{m}{\nabla}_p \cdots \nabla_{i_n} \psi \right).
\end{split}
\end{equation}
Note that we set the inverted order one $\nabla_{i_2} \nabla_{i_1} =1$.
Substituting \eqref{commuting-nabla-T-induction-n-1} into \eqref{commuting-nabla-T-induction-n}, and making some rearrangements, we prove that  \eqref{commuting-lie-T-I} holds for $l=n.$ Noting the formula \eqref{dtau-Gamma} for $\lie_{\dtau} \Gamma_{ij}^a$, we derive \eqref{id-commuting-nabla-N-T-l-simplify}.

Based on \ref{commuting-lie-T-nabla}, we have for 
\begin{equation}\label{commuting-lie-T-nabla-1}
\begin{split}
\lie_{\Tbar} \nabla_j \nabla_{I_n} \psi =& \nabla_j \lie_{\Tbar}  \nabla_{I_n} \psi  +N^{-1} \nabla_j N \lie_{\Tbar}  \nabla_{I_n} \psi  \\
& - \sum_{k=1}^n N^{-1} \lie_{\dtau} \Gamma^p_{j i_k} \nabla_{i_1}  \cdots \nabla_{p} \cdots \nabla_{i_n} \psi.
\end{split}
\end{equation}
We will use \eqref{commuting-lie-T-nabla-1} to  prove \eqref{id-commuting-nabla-T-l-simplify}, \eqref{def-commuting-K-l} by induction. For $l=1,$ by \eqref{commuting-id-nabla-T-induction-1} and \eqref{commuting-lie-T-nabla}, we have
\begin{equation*}
\begin{split}
 \lie_{\Tbar} \nabla_{i_1} \psi &= \nabla_{i_1} \Tbar \psi + N^{-1} \nabla_{i_1} N * \Tbar \psi,
 \end{split}
\end{equation*}
That is  \eqref{id-commuting-nabla-T-l-simplify}, \eqref{def-commuting-K-l} holds for $l=1$. Suppose \eqref{id-commuting-nabla-T-l-simplify} holds for $l \leq n-1$, in views of \eqref{commuting-lie-T-nabla-1}, we now consider by induction
\begin{equation*}
\begin{split}
& \lie_{\Tbar} \nabla_{i_1} \nabla_{i_2} \cdots \nabla_{i_n} \psi =   \nabla_{i_1} \left( \lie_{\Tbar} \nabla_{i_2} \cdots \nabla_{i_n} \psi \right) \\
& \quad \quad \quad + N^{-1} \nabla_{i_1} N \lie_{\Tbar}  \nabla_{i_2} \cdots \nabla_{i_n}  \psi  - N^{-1} \nabla (\hat N+ N \hat k) \nabla^n \psi \\
=&  \nabla_{i_1} \left(  \nabla_{i_2} \cdots \nabla_{i_n} \lie_{\Tbar} \psi + \sum_{a +b  = n-1, a  \geq 1} N^{-1} \nabla_{I_a} N * \nabla_{I_{b}} \Tbar \psi +   N^{-1} \mathcal{K}\mathcal{N}_{I_{n-1}}(\psi) \right) \\
& + N^{-1} \nabla_{i_1} N \left(  \nabla_{i_2} \cdots \nabla_{i_n} \lie_{\Tbar} \psi + \sum_{a +b  = n-1, a  \geq 1} N^{-1} \nabla_{I_a} N * \nabla_{I_{b}} \Tbar \psi \right) \\
&- N^{-1} \nabla_{i_1} N * N^{-1} \mathcal{K}\mathcal{N}_{I_{n-1}}(\psi) - N^{-1} \nabla (\hat N + N \hat k) \nabla^n \psi.
\end{split}
\end{equation*}
That is,
\begin{equation}\label{commuting-nabla-T-induction-n}
\begin{split}
& \lie_{\Tbar} \nabla_{I_n}\psi =  \nabla_{I_n} \lie_{\Tbar} \psi + N^{-1} \nabla N * \nabla_{I_{n-1}} \lie_{\Tbar} \psi  \\
& \quad + \sum_{a +1 +b  = n-1} N^{-1}  \nabla_{i_1} \left( \nabla_{I_a} \nabla N * \nabla_{I_{b}} \Tbar \psi \right)  \\
&\quad +  N^{-1}  \nabla  \left( \mathcal{K}\mathcal{N}_{I_{n-1}}(\psi) \right) - N^{-1} \nabla (\hat N + N \hat k) \nabla^n \psi,
\end{split}
\end{equation}
which further implies  \eqref{id-commuting-nabla-T-l-simplify}, \eqref{def-commuting-K-l} .
\end{proof}

\begin{lemma}[Commuting Lemma for $T$-Tangent tensor]\label{lemma-commuting-application-tensor}
Let $\Psi_{J_n}$ be a $T$-tangent $(0, n)$-tensor field. 
Let $l\geq 1$,
\begin{subequations}
\begin{equation}\label{id-commuting-nabla-N-T-l-simplify-tensor}
\begin{split}
&\lie_{\dtau}  \nabla_{I_l} \Psi_{J_n}=  \nabla_{I_l} \lie_{\dtau} \Psi_{J_n} + \mathcal{K}\mathcal{N}_{I_l}(\Psi_{J_n}),
\end{split}
\end{equation}
\begin{equation}\label{id-commuting-nabla-T-l-simplify-tensor}
\begin{split}
&\lie_{\Tbar}  \nabla_{I_l} \Psi_{J_n}=  \nabla_{I_l} \lie_{\Tbar} \Psi_{J_n} +  \mathcal{K}\mathcal{N}\mathcal{T}_{I_l} (\Psi_{J_n}),
\end{split}
\end{equation}
\end{subequations}
where $\mathcal{K}\mathcal{N}\mathcal{T}_{I_l}  (\Psi_{J_n})$ and $\mathcal{K}\mathcal{N}_{I_l}(\Psi_{J_n})$ are defined as
\begin{equation}\label{def-commuting-K-R-T-l-tensor}
\begin{split}
\mathcal{K}\mathcal{N}_{I_l}(\Psi_{J_n}) &=\pm \sum_{a+1 +b  = l} \nabla_{I_a} \nabla \left( N \hat  k + \hat N \right) * \nabla_{I_{b}} \Psi_{J_n} , \\
 \mathcal{K}\mathcal{N}\mathcal{T}_{I_l} (\Psi_{J_n}) = &   N^{-1}\mathcal{N}_{I_l} ( \lie_{\Tbar} \Psi_{J_n}) +N^{-1} \mathcal{K}\mathcal{N}_{I_l}(\Psi),
\end{split}
\end{equation}
where
\begin{equation}\label{def-eq-kg-N-l-1+3-error}
\begin{split}
 \mathcal{N}_{I_l} (\Psi) &= \sum_{a +1+b = l}  \nabla_{I_{a}} \nabla \hat N * \nabla_{I_b} \Psi, 
\end{split}
\end{equation}
\end{lemma}
\begin{proof}
In fact, we can prove
\begin{equation}\label{commuting-lie-T-I-tensor}
\begin{split}
&  \lie_{\dtau} \nabla_{i_1} \cdots \nabla_{i_l} \Psi_{j_1 \cdots j_n} -  \nabla_{i_1} \cdots \nabla_{i_l} \lie_{\dtau} \Psi_{j_1 \cdots j_n} \\
=&- \sum_{a=0}^{l-2} \sum_{m=a+2}^l \nabla_{i_1} \cdots \nabla_{i_a} \left( \lie_{\dtau} \Gamma^p_{i_{a+1} i_m} \nabla_{i_{a+2}} \cdots \stackrel{m}{\nabla}_p \cdots \nabla_{i_l} \Psi_{j_1 \cdots j_n} \right)\\
&- \sum_{a=0}^{l-1} \sum_{b=1}^n \nabla_{i_1} \cdots \nabla_{i_a} \left( \lie_{\dtau} \Gamma^p_{i_{a+1} j_b} \nabla_{i_{a+2}} \cdots \nabla_{i_l} \Psi_{j_1 \cdots j_{b-1} p j_{b+1} \cdots j_n} \right).
\end{split}
\end{equation}
And similar commuting formula between $\lie_{\Tbar}$ and $\nabla_{I_l}$ holds as well.
Noting the formulae \eqref{dtau-Gamma} for $\lie_{\dtau} \Gamma_{ij}^a$ and \eqref{commuting-lie-T-nabla},
we can derive the proof for \eqref{id-commuting-nabla-N-T-l-simplify-tensor}-\eqref{id-commuting-nabla-T-l-simplify-tensor}.
\end{proof}

We additionally present a commuting identity between $\nabla$ and $\Delta$, which can be easily proved by induction. 
\begin{lemma}\label{lemma-commuting-nabla-laplacian}
 Let  $l \geq 0$, then for any scalar field $\psi$,
\begin{equation}\label{commuting-nabla-laplacian-l-simplify}
\Delta \nabla_{I_l} \psi =  \nabla_{I_l} \Delta \psi  + \mathcal{R}_{I_l}(\psi),
\end{equation}
where $\mathcal{R}_{I_l}(\psi), l \geq 1$ is defined as 
\begin{equation}\label{def-R-l-commute-nabla-laplacian}
\mathcal{R}_{I_l}(\psi) = \pm \sum_{a+b=l, \, a\leq l-1} \nabla_{I_a} R_{imjn} * \nabla_{I_b} \psi.
\end{equation}
Generally for any $T$-tangent $(0, n)$-tensor $\Psi=\Psi_{J_n},$
\begin{equation}\label{commuting-nabla-laplacian-l-Psi-i-ij}
 \Delta \nabla_{I_l} \Psi_{J_n}= \nabla_{I_l} \Delta \Psi_{J_n} + \mathcal{R}_{I_l}(\Psi_{J_n}),
\end{equation}
where $\mathcal{R}_l(\Psi_{J_n})$ is defined as
\begin{equation}\label{Def-R-Psi-ij-commute-nabla-laplacian}
 \mathcal{R}_{I_l}(\Psi_{J_n})=\pm \sum_{a+b = l} \nabla_{I_a} R_{imjn} * \nabla_{I_b} \Psi.
\end{equation}
We also define that when $l=0$,  $\mathcal{R}_{I_0}(\psi) = \mathcal{R}_{I_0}(\Psi_{J_n}) =0.$
\end{lemma}

We collect some definitions which will be used later. For any scalar function or $T$-tangent tensor $A$, we define
\begin{equation} \label{def-hat-k-N-A-Bianchi}
\begin{split}
  \hat{\mathcal{K}}\mathcal{N}_{I_l} (A) =&  \sum_{a+b +c=l } \nabla_{I_a} N * \nabla_{I_b} \hat k *\nabla_{I_c} A.
\end{split}
\end{equation}

At last, we show an identity that is useful in the calculation follows. 
\begin{proposition}
For any scalar function $f$, there is a identity
\begin{align}\label{id-div-T}
 \dtau \int_{\Sigma_{t}} f \di \mu_g =  \int_{\Sigma_{t}} \left( \dtau f  + 3 \hat N f \right) \di \mu_g.
\end{align}
\end{proposition}

\subsection{Energy identities for the $1+3$ KG equation}\label{sec-ee-kg-1+3}
The KG equation can be decomposed in the following $1+3$ form  
\begin{equation}\label{eq-rescale-kg-1+3-0}
 \lie_{\dtau} \Tbar \phi - \left( \tr k N +1 \right) \Tbar \phi  -N \nabla^i \nabla_i \phi + \bar m^2(t)  N \phi -  \nabla_i N  \nabla^i \phi =0.
\end{equation}

Define the $i^{\text{th}}$ order energy norm for the KG field $\phi$ as follows,
letting 
\begin{equation}\label{def-energy-density-l-homo-kg-g}
\rho^g_i (\phi) = |\nabla_{I_i} \Tbar\phi|_g^2 +  |\nabla_{I_{i}} \nabla \phi |_g^2 +  \bar m^2(t) |\nabla_{I_i} \phi|_g^2,
\end{equation}
where $ |\nabla_{I_l} \phi|_g^2 : =  g^{i_1 j_1} \cdots g^{i_l j_l} \nabla_{I_l} \phi \nabla_{J_l} \phi$. And
\begin{equation}\label{def-energy-l-homo-kg-g}
E_i (\phi, t) = \int_{\Sigma_{t}} \rho_i (\phi)  \di \mu_g.
\end{equation}
$E^{\tilde g}_i (\phi, t)$ is defined via replacing the metric $g$ and rescaled mass $\bar m(t)$ in $E^g_i (\phi, t)$ by $\tilde g$ and $m$.
Then, we have $$E_i(\phi, t) = t^{-1+2i} E^{\tilde g}_i (\phi, t).$$
Actually, the ($i^{\text{th}}$ order) $\Tbar$ energy of KG field is given by
\begin{equation}\label{def-tilde-energy-norm-l-kg}
\bar{E}_i(\phi, t) =\int_{\Sigma_t}  \left( \rho_i (\phi)  -  \tr k  \nabla_{I_i}  \Tbar\phi  \nabla^{I_i}\phi \right) \di \mu_g.
\end{equation}
For $|\tr  k|=3 <  \bar m(t)$ (or $t_0 > 3 m^{-1}$), these two energy norm are indeed equivalent: $\frac{1}{2} E_i (\phi, t) <  \bar E_i (\phi, t)< \frac{3}{2} E_i (\phi, t)$.

\subsubsection{The zero order case}\label{sec-0-ee-kg-1+3}
\begin{theorem}\label{Thm-energy-id-KG-0} 
Let $\phi$ be a solution to the KG equation \eqref{eq-kg}, then the $0^{\text{th}}$-order energy admits 
\begin{equation}\label{energy-id-0-kg}
\begin{split}
&\dtau  \tilde E_0(\phi, t) + \tilde E_0(\phi, t)+ \int_{\Sigma_{t}}  2 N |\nabla \phi|^2 + {}_0LK =0,
\end{split}
\end{equation}
where the error term ${}_0LK$ is given by
\begin{equation}\label{def-0-LK}
{}_0LK = 3 N  \phi \Tbar \phi - 4 \hat N \phi  \Tbar \phi - 2N \hat k_{ij} \nabla^i \phi \nabla^j \phi  -2  \nabla_i N  \nabla^i \phi  \Tbar \phi.
\end{equation}
\end{theorem}
\begin{remark}
Such an energy inequality for the KG equation is also introduced by Galstian-Yagdjian \cite{Galstian-Yagdjian} in the FLRW spacetime. However, they pursue global results with low regularity, and hence do not need the high order energy identities.
\end{remark}

\begin{proof}
We multiply $2\Tbar\phi$ on \eqref{eq-rescale-kg-1+3-0}, noticing that $$2N \bar m^2(t) \phi \Tbar \phi =\dtau (\bar m^2(t) \phi^2) - 2 \bar m^2(t) \phi^2,$$ and
\begin{align*}
-2N \nabla^i \nabla_i \phi \Tbar \phi &= -\nabla_i \left(2N \nabla^i \phi  \Tbar \phi  \right) +2 \nabla^i \phi  \nabla_i \left( N  \Tbar \phi \right) \\
&=-\nabla_i \left(2N \nabla^i \phi  \Tbar \phi  \right)  + 2 \nabla^i \phi  \lie_{\dtau} \nabla_i \phi,
\end{align*}
where in the second equality above we have used the commuting identity \eqref{id-commuting-nabla-T-l-simplify} with $l=1$. Additionally, making use of the identity $\lie_{\dtau} g^{ij} = -2\hat N g^{ij} + 2N \hat k^{ij}$, and hence $$2 \nabla^i \phi  \lie_{\dtau} \nabla_i \phi = \dtau  |\nabla \phi|^2-2N \hat k_{ij} \nabla^i \phi \nabla^j \phi + 2\hat N |\nabla \phi|^2,$$ we arrive at 
\begin{equation}\label{energy-estimate-1+3-step1-1}
\begin{split}
&\dtau \rho_0(\phi)  - 2 \bar m^2(t) \phi^2  + 4 N \left(\Tbar \phi\right)^2+ 2 \hat N (\bar T \phi)^2  + 2 \hat N  |\nabla \phi|^2 \\
&-2N \hat k_{ij} \nabla^i \phi \nabla^j \phi -2  \nabla_i N  \nabla^i \phi  \Tbar \phi -2\nabla_i \left(N \nabla^i \phi  \Tbar \phi  \right) =0.
\end{split}
\end{equation}

Now let us focus on the quadratic term $4 N \left(\Tbar \phi\right)^2 - 2 \bar m^2(t) \phi^2$ in \eqref{energy-estimate-1+3-step1-1}. We will additionally address an integration by part (or rather extract a divergence form) and use  the KG equation to make the sign right. Performing an integration by part,
 \begin{equation*}
\begin{split}
 N \left(\Tbar \phi\right)^2 =& \dtau \left( \phi \Tbar \phi \right) - \phi \lie_{\dtau} \Tbar \phi.
\end{split}
\end{equation*}
We then substitute the KG equation \eqref{eq-rescale-kg-1+3-0} into the following formula,
\begin{align*}
4 N \left(\Tbar \phi \right)^2 - 2 \bar m^2(t) \phi^2 =& N \left( \Tbar \phi \right)^2 + 3 \dtau ( \phi \Tbar \phi) - 3 \phi \lie_{\dtau} \Tbar \phi - 2 \bar m^2(t) \phi^2\\
=& N \left(\Tbar \phi \right)^2+ \bar m^2(t) N \phi^2 + \dtau (3 \phi \Tbar \phi) + 2 \bar m^2(t)  \hat N \phi \\
& + 3 \phi  \left( 2N \bar T \phi + \hat N \Tbar \phi  - N \Delta \phi  -  \nabla_i N  \nabla^i \phi \right).
\end{align*}
The first two term now has a good sign, and the extra quadratic term $- 3N \phi \Delta \phi$,
\begin{align*}
- 3N \phi \Delta \phi   &= -3 \nabla_i  \left( N \phi \nabla^i \phi \right) +3 N |\nabla \phi|^2 + 3  \phi \nabla^i N \nabla_i \phi.
\end{align*}
 We remark that the error terms arising in this procedure are all of lower order: they indeed have additional $t^{-1}$ factor, for instance, $\phi \cdot N \bar T \phi $ can be rewritten as $t^{-1} \bar m (t) \phi \cdot N \bar T \phi.$
Putting the two identities above together, we have
\begin{equation*}
\begin{split}
&\dtau \left( \rho_0(\phi) + 3 \phi \Tbar \phi \right)  + N \rho_0 (\phi) +2 N |\nabla \phi|^2 + 2 \hat N \rho_0 (\phi) \\
& -2 \nabla_i  \left( N \phi \nabla^i \phi \right) -2\nabla_i \left(N \nabla^i \phi  \Tbar \phi  \right)  + 6N  \phi \Tbar \phi \\
& + 2  \hat N \phi  \Tbar \phi -2N \hat k_{ij} \nabla^i \phi \nabla^j \phi  -2  \nabla_i N  \nabla^i \phi  \Tbar \phi  =0.
\end{split}
\end{equation*}
Integrating on $\Sigma_{t}$ and making use of the identity \eqref{id-div-T}, there is
\begin{equation}\label{energy-estimate-1+3-step3-1}
\begin{split}
& \dtau \int_{\Sigma_{t}} \left( \rho_0(\phi) + 3 \phi \Tbar \phi \right) \di \mu_g + \int_{\Sigma_{t}}  2 N |\nabla \phi|^2 \di \mu_g  \\
&+ \int_{\Sigma_{t}} N  \left(  \rho_0(\phi) +3 \phi \Tbar \phi \right) + 3 N  \phi \Tbar \phi - 7 \hat N \phi  \Tbar \phi  \\
&+ \int_{\Sigma_{t}} - \hat N \rho_0 (\phi)  -2N \hat k_{ij} \nabla^i \phi \nabla^j \phi  -2  \nabla_i N  \nabla^i \phi  \Tbar \phi  =0.
\end{split}
\end{equation}
We then achieve \eqref{energy-id-0-kg}. 
\end{proof}

As a remark, \eqref{energy-id-0-kg} can also be deduced by the method of energy-momentum tensor. However, if we proceed to the higher order case, it will be natural and simpler to commute the spatially covariant derivative $\nabla$ with the KG equation in $1+3$ form \eqref{eq-rescale-kg-1+3-0}. And this turns out to be more compatible with the energy with only high order of spatial energy. 

\subsubsection{The higher order case}\label{sec-high-ei-kg-1+3}

\begin{lemma}\label{lemma-High-KG-1+3}
Let $l \geq 1,$ the $l^{\text{th}}$-order KG equation $\nabla_{I_l} \left( N \left( \Box_{\bar g} \phi -\bar m^2 (t) \phi \right) \right) =0$ can be decomposed in the following $1+3$ form
\begin{equation}\label{eq-kg-N-l-1+3-general-simply}
\begin{split}
& \lie_{\dtau}  \nabla_{I_l} \Tbar \phi - \left( N \tr k +1 \right) \nabla_{I_l} \Tbar \phi -  N \nabla_{I_l} \Delta \phi + \bar m^2(t) N \nabla_{I_l} \phi \\
 &\quad + \bar m^2(t) \mathcal{N}_{I_l}(\phi)  - \mathcal{K}\mathcal{N}^\prime_{I_l}(\Tbar \phi) - \mathcal{N}_{I_{l+1}} (\nabla \phi)=0,
\end{split}
\end{equation}
where $\mathcal{N}_{I_l} (\phi), \mathcal{N}_{I_{l+1}} (\nabla \phi)$ are defined as \eqref{def-eq-kg-N-l-1+3-error}, and $\mathcal{K}\mathcal{N}^\prime_{I_l} (\Tbar \phi)$ is defined by 
\begin{equation}\label{def-KN-prime-l}
\begin{split}
 \mathcal{K}\mathcal{N}^\prime_{I_l}(\Tbar \phi)& =  \mathcal{K}\mathcal{N}_{I_l}(\Tbar \phi) + \sum_{a+1+b=l } \tr k  \nabla_{I_a} \nabla N * \nabla_{I_{b}} \Tbar \phi,
\end{split}
\end{equation}
and $\mathcal{K}\mathcal{N}_{I_l}(\Tbar \phi)$ is defined as in \eqref{id-commuting-nabla-N-T-l-simplify}.
\end{lemma}

\begin{remark}[Inevitability of the borderline terms in the high order KG equation]\label{rk-borderline-term}
We first remark that  
\begin{equation*}
\bar m^2(t) \mathcal{N}_{I_l} (\phi) =\bar m^2(t) \left( \nabla_{I_l} (N \phi) -  N\nabla_{I_l} \phi \right),
\end{equation*}
will be one of the borderline terms in the energy estimate for KG equation.

One may argue that this borderline term is artificial, since if we commuting $\nabla_{I_l}$ straightforwardly with the original KG equation, as
\begin{equation}\label{eq-high-KG-original-1+3}
\begin{split}
 \nabla_{I_l}  \left(  \lie_{\Tbar} \Tbar \phi - \left(N^{-1} + \tr k \right) \Tbar \phi  -  \Delta \phi + \bar m^2(t)  \phi -  N^{-1} \nabla_i N  \nabla^i \phi \right) =0,
\end{split}
\end{equation}
then the aforementioned borderline term $\bar m^2 (t) \mathcal{N}_{I_l} (\phi)$ does not exist. In fact, such borderline term will alternatively arise from the commuting term $ [\lie_{\Tbar}, \nabla_{I_l}] \Tbar \phi =   \mathcal{K}\mathcal{N}\mathcal{T}_{I_l} (\Tbar \phi)$, c.f. \eqref{def-commuting-K-l}. 
After commuting $\lie_{\Tbar}$ with $\nabla_{I_l}$ and multiplying by $N$ on \eqref{eq-high-KG-original-1+3}, one has
\begin{equation*}
\begin{split}
 & \lie_{\dtau}  \nabla_{I_l}  \Tbar \phi   -\left( N \tr k +1 \right) \nabla_{I_l} \Tbar \phi - N\nabla_{I_l} (\Delta \phi) + \bar m^2(t) N \nabla_{I_l}  \phi \\
   =& N [\lie_{\Tbar}, \nabla_{I_l}] (\Tbar \phi) +N \nabla_{I_l} \left( N^{-1} \nabla_i N  \nabla^i \phi \right).
\end{split}
\end{equation*}
This almost takes the same form as \eqref{eq-kg-N-l-1+3-general-simply} (ignoring the lower order terms), except that the previous $\bar m^2(t) \mathcal{N}_{I_l} (\phi)$ is now replaced by $-N [\lie_{\Tbar}, \nabla_{I_l}] (\Tbar \phi)$. However, taking a close look at $-N[\lie_{\Tbar}, \nabla_{I_l}] (\Tbar \phi)$, we will find out that this is exactly the missing borderline term, c.f. \eqref{def-commuting-K-l}, 
\begin{equation*}
\begin{split}
 -N [\lie_{\Tbar}, \nabla_{I_l}] (\Tbar \phi)
  =& -\mathcal{N}_{I_l} (\Tbar^2 \phi) - \mathcal{K}\mathcal{N}_{I_l} (\Tbar \phi),  \quad l \geq 1,
\end{split}
\end{equation*}
where
\begin{equation*}
\begin{split}
 -\mathcal{N}_{I_l} (\Tbar^2 \phi) =& \sum_{a +b+1  = l} \nabla_{I_a} \nabla N * \nabla_{I_{b}} \Tbar^2 \phi.
\end{split}
\end{equation*}
We further make use of the KG equation to substitute $\nabla_{I_{b}} \Tbar^2 \phi$,
\begin{equation*}
\begin{split}
- \mathcal{N}_{I_l} (\Tbar^2 \phi) = & \sum_{a +b+1  = l} \bar m^2(t) \nabla_{I_a} \nabla N * \nabla_{I_{b}}  \phi + \text{l.o.t.}.
\end{split}
\end{equation*}
The first one on the right hand side above is exactly the borderline term $\bar m^2 (t) \mathcal{N}_{I_l} (\phi)$. Note that, in the energy argument, this term will be multiplied with $2 \nabla^{I_l}  \Tbar \phi$, and hence the whole term takes the form of $2 \bar m^2(t)  \nabla^{I_l}  \Tbar \phi * \mathcal{N}_{I_l} (\phi)$.

Secondly, let us remark that the other borderline term will also originate from something related to $ [\lie_{\Tbar}, \nabla_{I_l} ] \phi$. As in the proof of Corollary \ref{coro-energy-id-KG-l}, to establish the energy identity, we multiply $2 \nabla^{I_l}  \Tbar \phi$ on the high order KG equation \eqref{eq-kg-N-l-1+3-general-simply}. Let us focus on the massive term \eqref{eq-1+3-KG-multiplying-3} for the moment,
\begin{equation*}
\begin{split}
 & 2 \bar m^2(t)  N \nabla_{I_l}  \phi  \nabla^{I_l} \Tbar \phi  \\
 =& \dtau \left( \bar m^2(t) |\nabla_{I_l} \phi|^2 \right)  - 2 \bar m^2(t) N \nabla^{I_{l}} \phi  [\lie_{\Tbar}, \nabla_{I_l} ] \phi \\
 & - 2 \bar m^2(t) |\nabla_{I_l} \phi|^2+ 2 \bar m^2(t) \hat N | \nabla_{I_l} \phi |^2 - 2 \bar m^2(t) N \hat k* \nabla^{I_l} \psi * \nabla_{I_l} \phi.
 \end{split}
\end{equation*}
In the above identity, $\dtau \left( \bar m^2 (t) | \nabla_{I_l} \phi|^2 \right) $ will contribute to the energy norm,  $- 2 \bar m^2(t) |\nabla_{I_l} \phi|^2$ will be absorbed (c.f. \eqref{eq-step-2-0-kg}), and the other terms in the $3^{\text{rd}}$ line are cubic lower order terms. While the final one $-2\bar m^2 (t) \nabla^{I_{l}} \phi  [\lie_{\Tbar}, \nabla_{I_l} ] \phi$, as shown before, leads to the borderline term $-2\bar m^2 (t) \nabla^{I_{l}} \phi * \mathcal{N}_{I_l} (\Tbar \phi),$ which can been seen as somehow ``dual'' to the previous one $2\bar m^2 (t) \nabla^{I_l}  \Tbar \phi * \mathcal{N}_{I_l} (\phi)$. 

The terms discussed above are the only two borderline cases. More explicitly, they take the forms of
\begin{equation*}
\begin{split}
 &\sum_{a +b+1  = l} 2\bar m^2 (t)  \nabla_{I_a} \nabla N * \left( \nabla_{I_{b}} \phi * \nabla^{I_l}  \Tbar \phi - \nabla_{I_{b}} \Tbar \phi * \nabla^{I_{l}} \phi \right),
 \end{split}
\end{equation*}
which can not be cancelled.
\end{remark}
\begin{proof}
Take the derivative $\nabla_{i_1}$ on the KG equation in $1+3$ form \eqref{energy-id-0-kg}, by the commuting Lemma \ref{lemma-commuting-application} (with $l=1$ in \eqref{id-commuting-nabla-N-T-l-simplify}),
\begin{equation*}
\begin{split}
& 0= \lie_{\dtau} \nabla_{i_1} \Tbar \phi -\nabla_{i_1} \Tbar \phi - N \tr k  \nabla_{i_1} \Tbar \phi + \bar m^2(t) N \nabla_{i_1} \phi -  N \nabla_{i_1}\nabla^p \nabla_p \phi  \\
 & +\bar m^2(t) \nabla_{i_1} N \phi - \tr k \nabla_{i_1} N  \Tbar \phi - \left(\nabla_{i_1} N \nabla^p \nabla_p \phi + \nabla_{i_1} \left(\nabla^p N \nabla_p \phi \right) \right).
\end{split}
\end{equation*}
Then \eqref{eq-kg-N-l-1+3-general-simply}-\eqref{def-eq-kg-N-l-1+3-error} with $l=1$ follows.

Suppose  the commuting Lemma \ref{lemma-commuting-application} (with $l=1$ in \eqref{id-commuting-nabla-N-T-l-simplify}) holds for $l=n-1,$ that is, the $(n-1)^{\text{th}}$ order KG equation in $1+3$ form takes the form
\begin{equation}\label{eq-kg-N-n-1-1+3-general}
\begin{split}
& \lie_{\dtau} \nabla_{i_2} \cdots \nabla_{i_n} \Tbar \phi - \left( \tr k N +1 \right) \nabla_{i_2} \cdots \nabla_{i_n} \Tbar \phi  - N\nabla_{i_2} \cdots \nabla_{i_n} \nabla^p \nabla_p \phi \\
 &+ \bar m^2(t) N \nabla_{i_2} \cdots \nabla_{i_n} \phi  + \bar m^2(t) \mathcal{N}_{I_{n-1}}(\phi)  - \mathcal{K}\mathcal{N}^\prime_{I_{n-1}}(\Tbar \phi) + \mathcal{N}_{I_{n}} (\nabla \phi) =0,
\end{split}
\end{equation}
Now taking the derivative $\nabla_{i_1}$ on \eqref{eq-kg-N-n-1-1+3-general}, and using  the commuting Lemma \ref{lem-commu-lie} with $l=n-1$, we have
\begin{equation*}
\begin{split}
 &\nabla_{i_1} (\lie_{\dtau} \nabla_{i_2} \cdots  \nabla_{i_{n}} \Tbar \phi) =\lie_{\dt} \nabla_{I_n}  \Tbar \phi + \sum_{m=2}^n \lie_{\dtau} \Gamma^p_{i_1 i_m} \nabla_{i_2} \cdots  \stackrel{m}{\nabla}_p \cdots \nabla_{i_n} \phi.
\end{split}
\end{equation*}
Hence, we obtain
\begin{equation}\label{eq-kg-N-l-1+3-general}
\begin{split}
& \lie_{\dtau} \nabla_{I_n} \Tbar \phi - ( \tr k N +1 ) \nabla_{I_n} \Tbar \phi - N\nabla_{I_n} \nabla^p \nabla_p \phi + \bar m^2(t) N \nabla_{I_n} \phi  \\
&+ \sum_{m=2}^n \lie_{\dtau} \Gamma^p_{i_1 i_m} \nabla_{i_2} \cdots  \stackrel{m}{\nabla}_p \cdots \nabla_{i_n} \phi - \tr k \nabla_{i_1} N \nabla_{i_2} \cdots \nabla_{i_n} \Tbar \phi \\
 & - \nabla_{i_1} N  \nabla_{i_2} \cdots \nabla_{i_n} \nabla^p \nabla_p \phi + \bar m^2(t) \nabla_{i_1} N  \nabla_{i_2} \cdots \nabla_{i_n} \phi\\
 &+ \nabla_{i_1} \left(  \bar m^2(t) \mathcal{N}_{I_{n-1}}(\phi)  - \mathcal{K}\mathcal{N}^\prime_{I_{n-1}}(\Tbar \phi) + \mathcal{N}_{I_{n}} (\nabla \phi) \right) =0.
\end{split}
\end{equation}
In views of the definition for $\mathcal{K}\mathcal{N}^\prime_{I_l}(\Tbar \phi), \mathcal{N}_{l+1} (\nabla \phi), \mathcal{N}_{l} (\phi)$ and $\mathcal{K}\mathcal{N}_{I_l}(\Tbar \phi)$, we rearrange the $2^{\text{nd}}$-$4^{\text{th}}$ lines in \eqref{eq-kg-N-l-1+3-general}, thus \eqref{eq-kg-N-l-1+3-general-simply}-\eqref{def-eq-kg-N-l-1+3-error}  holds for $l =n.$
\end{proof}

The high order KG equation in $1+3$ form  \eqref{eq-kg-N-l-1+3-general-simply}-\eqref{def-eq-kg-N-l-1+3-error}  leads to the partial energy identities of high order derivatives.
\begin{corollary}\label{coro-energy-id-KG-l}
Let $\phi$ be a solution to the KG equation \eqref{eq-kg}, and $l \geq 1,$ then the $l^{\text{th}}$-order tilde energy $\tilde{E}_l(\phi, t) $ (See definition in \eqref{def-tilde-energy-norm-l-kg}) admits
\begin{equation}\label{energy-id-l-kg-rescale}
\begin{split}
  &\dtau  \tilde E_l (\phi, t)   + \tilde E_l (\phi, t)  + \int_{\Sigma_t} 2 N |\nabla_{I_{l+1}} \phi|^2 \di \mu_g \\
  = & \int_{\Sigma_t} \left( BL_l + {}_lLK_1 + {}_lLK_2 + {}_lLK_3 \right) \di \mu_g,
\end{split}
\end{equation}
where the lower order terms ${}_lLK_1, {}_lLK_2, {}_lLK_3$ are given below
\begin{equation}\label{energy-estimate-kg-1-id-l-o-t-1}
\begin{split}
  {}_lLK_1 =& \left( \hat N + N\hat k \right)  | \nabla_{I_l} \D \phi|^2 + \left( 3N-3\hat N + N \hat k \right) \nabla^{I_l}  \phi \nabla_{I_l}  \Tbar \phi \\
  & - 2N \nabla^{I_l} \Tbar \phi  \mathcal{R}_{I_l} (\phi)  -  N \nabla^{I_l} \phi   \mathcal{R}_{I_l}(\phi),
\end{split}
\end{equation}
where   $\D \phi \in \{ \Tbar \phi, \nabla \phi, \bar m (t) \phi\}$, and
\begin{equation}\label{energy-estimate-kg-1-id-l-o-t-2}
\begin{split}
  {}_lLK_2 = &\nabla^{I_{l}} \D \phi \left( \mathcal{N}_{I_{l+1}}(\D \phi)  + \mathcal{N}_{I_{l}}(\D \phi)  + \mathcal{K}\mathcal{N}_{I_l}(\D \phi) + \mathcal{K}\mathcal{N}^\prime_{I_l}(\Tbar \phi) \right),
\end{split}
\end{equation}
\begin{equation}\label{energy-estimate-kg-1-id-l-o-t-3}
\begin{split}
  {}_lLK_3 = & \nabla^{I_l} \Tbar \phi   \mathcal{K}\mathcal{N}_{I_l} (\phi) + \nabla^{I_l}\phi  \mathcal{K}\mathcal{N}^\prime_{I_l}(\Tbar \phi)\\
  & + \nabla^{I_l}\phi \left( \mathcal{N}_{I_{l}}(\Tbar \phi) + \mathcal{N}_{I_{l+1}} (\nabla \phi) \right),
\end{split}
\end{equation}
and the borderline terms $BL_l$ are given by
\begin{equation}\label{def-energy-estimate-kg-id-bl} 
BL_l= 2\bar m^2 (t) \left( \nabla^{I_l} \Tbar \phi \mathcal{N}_{I_l}(\phi) - \nabla^{I_l} \phi \mathcal{N}_{I_l}(\Tbar  \phi) \right).
\end{equation}
We refer to \eqref{def-commuting-KN-l}-\eqref{def-commuting-K-l}, \eqref{def-N-l}, \eqref{def-R-l-commute-nabla-laplacian}, \eqref{def-KN-prime-l} and \eqref{def-eq-kg-N-l-1+3-error}-\eqref{def-hat-k-N-A-Bianchi} for definitions of $\mathcal{K}\mathcal{N}_{I_l}(\D\phi) \cdots$
\end{corollary}

\begin{remark}\label{rk-CMCSH gauge-error-shift}
If working within the CMCSH gauge, $\dt=NT+X_t, \, \dtau = N \Tbar +X$, where $X=t  X_t$ denotes the rescaled shift. And extra terms as $-\nabla_{I_l} \left( \lie_{X} \Tbar \phi \right) + \nabla_{I_l} \left(X^2 \cdot \Tbar\phi \right)$ must be added in the KG equation \eqref{eq-rescale-kg-1+3-0}. In the energy argument, these terms will yield the leading terms $\nabla_i \left( X^i |\nabla_{I_l} \Tbar \phi|^2 \right) \pm {}^X\Pi * |\nabla_{I_l} \Tbar \phi|^2 \pm \nabla_{I_{l-2}} \left(  {}^X\Pi  * \nabla \Tbar \phi \right) * \nabla_{I_l} \Tbar \phi$, where ${}^X\Pi_{ij}=\nabla_iX_j + \nabla_jX_i$ denotes the deformation tensor of $X$. The first term vanishes after integrating on $\Sigma_t$, while the other two terms can never be borderline terms, as we expect that $X$ shares the same (or better) estimates with $\hat N$, c.f. Remark \ref{rk-CMCSH gauge-shift}. 
\end{remark}

\begin{proof}
The proof is in spirit of that of Theorem \ref{Thm-energy-id-KG-0}.
We multiply $2 \nabla_{I_l} \Tbar \phi$ on the $l^{\text{th}}$-order KG equation in $1+3$ form \eqref{eq-kg-N-l-1+3-general-simply}-\eqref{def-eq-kg-N-l-1+3-error}. At first, for $2 \nabla^{I_l} \Tbar \phi \lie_{\dtau} \nabla_{I_l} \Tbar \phi,$ 
\begin{equation}\label{eq-1+3-KG-multiplying-1}
\begin{split}
 &2 \nabla^{I_l} \Tbar \phi \lie_{\dtau} \nabla_{I_l} \Tbar \phi  \\
 =& \dtau \left(  | \nabla_{I_l} \Tbar \phi|^2 \right) + 2\hat N |\nabla_{I_l} \Tbar \phi|^2 - 2 N \hat k *  \nabla_{I_l} \Tbar \phi * \nabla^{I_l} \Tbar \phi.
 \end{split}
\end{equation}
Secondly, for $-2  \nabla^{I_l} \Tbar \phi  N \nabla_{I_l} \nabla^p \nabla_p \phi$, by the commuting identity \eqref{commuting-nabla-laplacian-l-simplify}, 
\begin{equation}\label{eq-1+3-KG-multiplying-2-0}
\begin{split}
 &-2  \nabla^{I_l} \Tbar \phi  N \nabla_{I_l} \nabla^p \nabla_p \phi=-2  \nabla^{I_l} \Tbar \phi N \left(  \Delta \nabla_{I_l} \phi + \mathcal{R}_l(\phi) \right)\\
 =&- \nabla_i \left( 2  N \nabla^i \nabla_{I_l} \phi \nabla^{I_l} \Tbar \phi  \right) + 2  N \nabla^{I_{l+1}} \phi \nabla_{I_{l+1}} \Tbar \phi \\
 & +  \nabla^{I_l} \Tbar \phi * \nabla N * \nabla_{I_{l+1}} \phi  -  N \nabla^{I_l} \Tbar \phi \mathcal{R}_{I_l} (\phi).
 \end{split}
\end{equation}
Applying the commuting identity \eqref{id-commuting-nabla-T-l-simplify} to $\nabla_{I_{l+1}} \Tbar \phi$, $2 N \nabla^{I_{l+1}} \phi \nabla_{I_{l+1}} \Tbar \phi$ in \eqref{eq-1+3-KG-multiplying-2-0} can be further reduced in an analogous way to \eqref{eq-1+3-KG-multiplying-1}. Therefore,
\begin{equation}\label{eq-1+3-KG-multiplying-2}
\begin{split}
 &-2 \nabla^{I_l} \Tbar \phi N \nabla_{I_l} \nabla^p \nabla_p \phi =- \nabla_i \left( 2   N \nabla^i \nabla_{I_l} \phi \nabla^{I_l} \Tbar \phi  \right) \\
 &\quad  +  \dtau \left( |  \nabla_{I_{l+1}} \phi|^2 \right)  + 2 \hat N | \nabla^{l+1} \phi |^2  - 2  N \hat k * \nabla^{I_{l+1}} \phi * \nabla_{I_{l+1}} \phi   \\
& \quad -2 N  \nabla^{I_{l+1}} \phi *  \mathcal{K}\mathcal{N}\mathcal{T}_{I_{l+1}}(\phi)+ \nabla^{I_l} \Tbar \phi \left( \nabla N * \nabla_{I_{l+1}} \phi   -  N \mathcal{R}_{I_l} (\phi) \right).
 \end{split}
\end{equation}
Thirdly, $2 \bar m^2(t) \nabla^{I_l} \Tbar \phi  N \nabla_{I_l}  \phi$ can be treated in the same way, thus
\begin{equation}\label{eq-1+3-KG-multiplying-3}
\begin{split}
 &2\bar m^2 (t) \nabla^{I_l}  \Tbar \phi   N \nabla_{I_l}  \phi \\
 =& \dtau \left( \bar m^2(t) | \nabla_{I_l} \phi|^2 \right)  -2 \bar m^2(t) N \nabla^{I_{l}} \phi  * \mathcal{K}\mathcal{N}\mathcal{T}_{I_l}(\phi)\\
 & - 2 \bar m^2(t) |\nabla_{I_l} \phi|^2+2 \bar m^2(t) \hat N |\nabla_{I_l} \phi|^2   - 2 \bar m^2(t) N \hat k *  \nabla_{I_l} \phi *  \nabla^{I_l} \phi .
 \end{split}
\end{equation}
Therefore, putting \eqref{eq-1+3-KG-multiplying-1}-\eqref{eq-1+3-KG-multiplying-3} together,
\begin{align*}
 &\dtau  \left( \rho_{l}(\phi) \right) +4 N | \nabla_{I_l} \Tbar \phi|^2  - 2 \bar m^2(t) |\nabla_{I_l} \phi|^2 \\
 & - \nabla_i \left( 2 N \nabla^i \nabla_{I_l} \phi \nabla^{I_l} \Tbar \phi  \right) + \text{error}=0,
\end{align*}
where the error terms are given by
\begin{equation}\label{def-error-kg-id-l}
\begin{split}
 \text{error} =& 2 \hat N |\nabla_{I_l} \D \phi|^2  - 2  N \hat k * |\nabla^{I_l} \D \phi|^2 \\
 & + \nabla^{I_l} \Tbar \phi \nabla N * \nabla_{I_{l+1}} \phi   -  N \nabla^{I_l} \Tbar \phi \mathcal{R}_{I_l} (\phi)  \\
 & -2 N \left( \nabla^{I_{l+1}} \phi * \mathcal{K}\mathcal{N}\mathcal{T}_{I_{l+1}}(\phi) + \bar m^2(t) \nabla^{I_{l}} \phi   * \mathcal{K}\mathcal{N}\mathcal{T}_{I_l}(\phi) \right)\\
 &+ 2  \nabla^{I_l} \Tbar \phi \left( \bar m^2(t) \mathcal{N}_{I_l}(\phi)  - \mathcal{K}\mathcal{N}^\prime_{I_l}(\Tbar \phi) - \mathcal{N}_{I_{l+1}} (\nabla \phi) \right).
  \end{split}
\end{equation}
$ \mathcal{K}\mathcal{N}\mathcal{T}_{I_l} (\phi)$ (c.f. \eqref{def-commuting-K-l}) will now split into two parts:
\begin{equation*}
\begin{split}
 \mathcal{K}\mathcal{N}\mathcal{T}_{I_l} (\phi) &= N^{-1} \mathcal{N}_{I_l} (\Tbar \phi ) + N^{-1} \mathcal{K}\mathcal{N}_{I_l} (\phi), \,\, l \geq 1,
\end{split}
\end{equation*}
where the definitions of $\mathcal{N}_{I_l} (\Tbar \phi ), \mathcal{K}\mathcal{N}_{I_l} (\phi)$ are shown in \eqref{def-eq-kg-N-l-1+3-error} and \eqref{def-commuting-KN-l}.
Hence, the $3^{\text{rd}}$ line of \eqref{def-error-kg-id-l} contains $$-2 \nabla^{I_{l+1}} \phi * \mathcal{N}_{l+1}(\Tbar\phi) - 2 \bar m^2(t)  \nabla^{I_{l}} \phi * \mathcal{N}_{l}(\Tbar \phi). $$ We note that $\bar m^2 (t) = m^2 t^2$. The extra $t^2$ makes $\bar m^2(t) \nabla^{I_l} \phi * \mathcal{N}_{l}(\Tbar \phi)$ borderline term, while $\nabla^{I_{l+1}} \phi * \mathcal{N}_{l+1}(\Tbar \phi)$ is a lower order term. The other borderline term is $2\bar m^2 (t) \nabla^{I_l} \Tbar \phi * \mathcal{N}_{I_l}(\phi)$ which is originated from $\bar m^2(t) [\nabla_{I_l}, N] \phi$.
We further apply the identity \eqref{id-div-T} to achieve
\begin{equation}\label{energy-estimate-kg-1-step1}
\begin{split}
 &\dtau \rho_l(\phi) + 4 N | \nabla_{I_l} \Tbar \phi|^2  - 2 \bar m^2(t) |\nabla_{I_l} \phi|^2 \\
 & - \nabla_i \left( 2 N \nabla^i \nabla_{I_l} \phi \nabla^{I_l} \Tbar \phi  \right) -BL_l + \text{l.o.t.}_1=0.
\end{split}
\end{equation}
where we had separated the error terms into two parts: the borderline terms 
\begin{equation}\label{def-b-l-K}
BL_l = - 2 \bar m^2(t) \nabla^{I_l} \phi * \mathcal{N}_{I_l}(\Tbar \phi) + 2 \bar m^2(t)   \nabla^{I_l} \Tbar \phi * \mathcal{N}_{I_l}(\phi) 
\end{equation}
and the lower order terms
\begin{equation}\label{def-lot-1-K}
\begin{split}
 \text{l.o.t.}_1 =&  2 \hat N |\nabla_{I_l} \D \phi|^2  - 2N \hat k * |\nabla_{I_l} \D \phi|^2 - 2 \nabla^{I_{l+1}} \phi * \mathcal{N}_{I_{l+1}}(\Tbar \phi)  \\
  & - 2 \left( \nabla^{I_{l+1}} \phi * \mathcal{K}\mathcal{N}_{I_{l+1}}(\phi) +\bar m^2(t) \nabla^{I_{l}} \phi * \mathcal{K}\mathcal{N}_{I_l}(\phi)\right)\\
 &-   \nabla^{I_l} \Tbar \phi \left(  \mathcal{K}\mathcal{N}^\prime_{I_l}(\Tbar \phi) + \mathcal{N}_{I_{l+1}} (\nabla \phi) +N\mathcal{R}_{I_l} (\phi) \right).
\end{split}
\end{equation}
Here we note that $ \nabla^{I_l} T\phi \nabla N * \nabla_{I_{l+1}} \phi$ in \eqref{def-error-kg-id-l} has been covered by $\nabla^{I_{l+1}} \phi * \mathcal{N}_{I_{l+1}}(\Tbar \phi) $ in \eqref{def-lot-1-K}.

In the second step, we will concentrate on the quadratic terms in \eqref{energy-estimate-kg-1-step1},
\begin{equation}\label{trace-part-energy-id-kg-l}
 4 N | \nabla_{I_l} \Tbar \phi|^2  - 2 \bar m^2(t) |\nabla_{I_l} \phi|^2.
\end{equation}
We again apply the commuting identity \eqref{id-commuting-nabla-T-l-simplify} to $\nabla_{I_{l}} \Tbar \phi$,
\begin{equation*}
\begin{split}
&N |\nabla_{I_l} \Tbar \phi|^2  = \dtau \left(  \nabla^{I_l}  \Tbar \phi \nabla_{I_l}  \phi \right) +  \bar m^2(t) N |\nabla_{I_l} \phi|^2+ 3 \hat N \nabla^{I_l} \Tbar \phi \nabla_{I_l} \phi  \\
&- N \hat k \nabla^{I_l} \Tbar \phi * \nabla_{I_l} \phi  - N \nabla^{I_l} \Tbar \phi   * \mathcal{K}\mathcal{N}\mathcal{T}_{I_l} (\phi) + 2 N   \nabla^{I_l}\phi \nabla_{I_l} \Tbar \phi\\
&+  \nabla^{I_l}\phi \left( -  N \nabla_{I_l} \Delta \phi  + \bar m^2(t) \mathcal{N}_{I_l}(\phi)  - \mathcal{K}\mathcal{N}^\prime_{I_l}(\Tbar \phi) - \mathcal{N}_{I_{l+1}} (\nabla \phi) \right).
\end{split}
\end{equation*}
Using the high order KG equation \eqref{eq-kg-N-l-1+3-general-simply}, we have
\begin{equation}\label{eq-step-2-0-kg}
\begin{split}
&4 N |\nabla_{I_l} \Tbar \phi|^2 - 2 \bar m^2(t) |\nabla_{I_l} \phi|^2 = N |\nabla_{I_l} \Tbar \phi|^2 + \dtau \left( 3 \nabla^{I_l}  \Tbar \phi \nabla_{I_l}  \phi \right)  \\
& +  \bar m^2(t) N |\nabla_{I_l} \phi|^2 +  3N \nabla^{I_l}\phi \nabla_{I_l} \Tbar \phi  - 3N \nabla^{I_l}\phi  \nabla_{I_l} \Delta \phi + \text{l.o.t.}_2,
\end{split}
\end{equation}
where the lower order terms are given by
\begin{equation}\label{energy-estimate-kg-1-step2-1-lot-2}
\begin{split}
& \text{l.o.t.}_2=  (3N+9\hat N)   \nabla^{I_l}\phi \nabla_{I_l} \Tbar \phi -  N \hat k \nabla^{I_l} \Tbar \phi * \nabla_{I_l} \phi \\
&  -N \nabla^{I_l} \Tbar \phi * \mathcal{K}\mathcal{N}\mathcal{T}_{I_l}(\phi) + 2 \bar m^2(t) \hat N |\nabla_{I_l} \phi|^2 \\
&+ \nabla^{I_l}\phi \left( \bar m^2(t)  \mathcal{N}_{I_l}(\phi) - \mathcal{K}\mathcal{N}^\prime_{I_l}(\Tbar \phi) - \mathcal{N}_{I_{l+1}} (\nabla \phi) \right).
\end{split}
\end{equation}
Applying the commuting identity between $\Delta$ and $\nabla_{I_l}$ \eqref{commuting-nabla-laplacian-l-simplify} to $\Delta \nabla_{I_l} \phi$ above,
we have
\begin{equation}\label{energy-estimate-kg-1-step2-2}
\begin{split}
 & - 3 N\nabla^{I_l}\phi  \nabla_{I_l} \Delta \phi  = - 3N \nabla^{I_l}\phi \left( \Delta \nabla_{I_l} \phi - \mathcal{R}_{I_l}(\phi) \right)  \\
=& -3 \nabla^{i} \left( N \nabla^{I_l}\phi  \nabla_i \nabla_{I_l}\phi  \right) + 3 N \nabla^{I_{l+1}}\phi \nabla_{I_{l+1}} \phi + \text{l.o.t.}_3,
\end{split}
\end{equation}
where the lower order terms are given by
\begin{equation}\label{lot-kg-1-step2-2}
\text{l.o.t.}_3 =3 \nabla^{I_l}\phi\nabla^{i}  N  \nabla_i \nabla_{I_l}\phi + 3\nabla^{I_l}\phi N \mathcal{R}_{I_l}(\phi).
\end{equation}
Putting \eqref{eq-step-2-0-kg}-\eqref{energy-estimate-kg-1-step2-1-lot-2} and \eqref{energy-estimate-kg-1-step2-2}-\eqref{lot-kg-1-step2-2} together, we obtain
\begin{equation}\label{energy-estimate-kg-1-step2-4}
\begin{split}
&4 N  |\nabla_{I_l} \Tbar \phi|^2 - 2 \bar m^2(t) |\nabla_{I_l} \phi|^2 \\
=& N  |\nabla_{I_l} \Tbar \phi|^2 + \bar m^2(t) N |\nabla_{I_l} \phi|^2 + 3 N |\nabla_{I_{l+1}} \phi|^2 + 3N \nabla^{I_l}\phi \nabla_{I_l} \Tbar \phi \\
&+ \dtau \left( 3 \nabla^{I_l}  \Tbar \phi \nabla_{I_l}  \phi \right) -3 \nabla^{i} \left( N \nabla^{I_l}\phi  \nabla_i \nabla_{I_l}\phi  \right) + \text{l.o.t.}_2  +   \text{l.o.t.}_3.
\end{split}
\end{equation}

Combining \eqref{energy-estimate-kg-1-step2-4} with \eqref{energy-estimate-kg-1-step1}-\eqref{def-lot-1-K}, we arrive at
\begin{equation}\label{energy-estimate-kg-l-id}
\begin{split}
  &\dtau  \left( \rho_l (\phi) +  3 \nabla^{I_l}  \Tbar \phi \nabla_{I_l}  \phi  \right) + N  \left( \rho_l (\phi) +  3 \nabla^{I_l}  \Tbar \phi \nabla_{I_l}  \phi  \right) \\
  & + 2 N |\nabla_{I_{l+1}} \phi|^2  - \nabla^{i} \left( 3N \nabla^{I_l}\phi  \nabla_i \nabla_{I_l}\phi + 2 N \nabla^i \nabla_{I_l} \phi \nabla^{I_l} \Tbar \phi   \right) \\
  & -BL_l +  \text{l.o.t.}_1 + \text{l.o.t.}_2 + \text{l.o.t.}_3,
\end{split}
\end{equation}
where $BL_l$, $\text{l.o.t.}_1, \cdots, \text{l.o.t.}_3$ are defined in \eqref{def-b-l-K}, \eqref{def-lot-1-K}, \eqref{energy-estimate-kg-1-step2-1-lot-2}, \eqref{lot-kg-1-step2-2}.
Integrate \eqref{energy-estimate-kg-l-id} on $\Sigma_t$,
\begin{equation*}
\begin{split}
  &\dtau \int_{\Sigma_t} \left(  \rho_l (\phi) +  3 \nabla^{I_l}  \Tbar \phi \nabla_{I_l}  \phi  \right)  + \int_{\Sigma_t} \left(   \rho_l (\phi, t) +  3 \nabla^{I_l}  \Tbar \phi \nabla_{I_l}  \phi  \right)  \\
 & + \int_{\Sigma_t} 2 N |\nabla_{I_{l+1}} \phi|^2 -BL_l - {}_lLK_1 - {}_lLK_2- {}_lLK_3 =0,
\end{split}
\end{equation*}
where ${}_lLK_i$ are defined as in \eqref{energy-estimate-kg-1-id-l-o-t-1}-\eqref{energy-estimate-kg-1-id-l-o-t-3}.
  And we denote $\D \phi \in \{ \Tbar \phi, \nabla \phi, \bar m \phi\}$
As we can see, the $2 N | \nabla_{I_{l+1}} \phi |^2$ above has a good sign and it will contribute a positive double integral in the energy estimate.
We then achieve \eqref{energy-id-l-kg-rescale}-\eqref{energy-estimate-kg-1-id-l-o-t-3}.
\end{proof}

\subsection{Energy identities for the $1+3$ Bianchi equations}\label{sec-ee-Bianchi-1+3}
 \subsubsection{The Weyl fields and Bianchi equations}
The Weyl tensor $W_{\alpha \beta \gamma \delta}$ being the traceless part of the curvature tensor is 
\begin{equation}\label{def-Weyl}
\begin{split}
W_{\alpha \beta \gamma \delta} = &\bar R_{\alpha \beta \gamma \delta} - \frac{1}{2} ( \bar g \odot Ric(\bar g) )_{\alpha \beta \gamma \delta} +\frac{1}{12} \bar R ( \bar g \odot \bar g )_{\alpha \beta \gamma \delta},
\end{split}
\end{equation}
where $Ric(\bar g)$ is the Ricci tensor of $\bar g$, and $\odot$ is defined as \eqref{def-odot}.
We define the left and right Hodge duals of $W$ by
\begin{equation}\label{eq-Bianchi-J}
{}^{*}W_{\alpha \beta \gamma \delta} = \frac{1}{2} \epsilon_{\alpha \beta \mu \nu} W^{\mu \nu}{}_{\! \gamma \delta},
\quad  W^*_{\alpha \beta \gamma \delta} = \frac{1}{2} W_{\alpha \beta}{}^{\! \mu \nu} \epsilon_{\mu \nu \gamma \delta}.
\end{equation}
However, by virtue of the algebraic properties of $W$ the two duals coincide: ${}^{*}W = W^*$.
The Bianchi identities entail the divergence equations
\begin{subequations}
\begin{equation}
D^\alpha W_{\alpha \beta \gamma \delta} = J_{\beta \gamma \delta},  \label{eq-Bianchi-J}
\end{equation}
\begin{equation}
D^\alpha {}^{*}W_{\alpha \beta \gamma \delta} = J^*_{\beta \gamma \delta},  \label{eq-Bianchi-J-*}
\end{equation}
\end{subequations}
 where
 \begin{equation}\label{def-J}
 J_{\beta \gamma \delta} = \frac{1}{2} \left( D_\gamma \bar R_{\delta \beta} - D_{\delta} \bar R_{\gamma \beta}\right) - 
\frac{1}{12} \left(\bar g_{\beta \delta} D_\gamma \bar R - \bar g_{\beta \gamma} D_\delta \bar R \right),
\end{equation}
and then
\begin{equation}\label{def-J-*}
 J^*_{\beta \gamma \delta} = \frac{1}{2} J_\beta{}^{\! \mu \nu} \epsilon_{\mu \nu \gamma \delta}.
\end{equation} 
The source terms $J_{\beta \gamma \delta}, J^*_{\beta \gamma \delta}$ are related to the matter field, since the Ricci tensor $\bar R_{\alpha \beta}$ is given by the KG field, see \eqref{eq-ricci-kg-phi}.

The electric and magnetic parts $E(W), H(W)$ of the Weyl fields $W$, with respect to the foliation $\Sigma_t$ are defined by
\begin{equation}\label{def-electric-magnetic}
E(W)_{\alpha \beta} = W_{\alpha \mu \beta \nu} \Tbar^\mu \Tbar^\nu, \quad H(W)_{\alpha \beta} = {}^{*}W_{\alpha \mu \beta \nu} \Tbar^\mu \Tbar^\nu.
\end{equation} 
The tensors $E$ and $H$ are $T$-tangent, i.e. $E_{\alpha \beta}T^\alpha = H_{\alpha \beta}T^\alpha=0$ and trace free $\bar g^{\alpha \beta} E_{\alpha \beta} = \bar g^{\alpha \beta} H_{\alpha \beta}=0$. It follows that $g^{ij} E_{ij} =g^{ij} H_{ij}=0$.
The following identities relate $W, {}^{*}W, E=E(W), H=H(W),$ (c.f. Page 144 in \cite{Christodoulou-K-93}),
\begin{align*}
W_{ijp \Tbar} &= -\epsilon_{ij}{}^{\! m} H_{mp}, \quad \quad \quad {}^{*}W_{ijp \Tbar} = \epsilon_{ij}{}^{\! m} E_{mp},\\
W_{ijpq} &= -\epsilon_{ijm} \epsilon_{pqn} E^{mn}, \quad \,\, {}^{*}W_{ijpq} = \epsilon_{ijm} \epsilon_{pqn} H^{mn}.
\end{align*} 

We note that, the spacetime hodge dual $\ast$, and $J_{\alpha \beta \gamma}$, $E_{ij}, H_{ij}$ are all scale-free.
The Gauss and Codazzi equations can be written in terms of $E$ and $H$ 
\begin{subequations}
\begin{equation}
R_{ij} + 2 g_{ij} - \hat k_{im} \hat k^m{}_{\! j} - \hat k_{ij} =E_{ij} +\frac{1}{2} \bar R_{ij} + \frac{1}{6}(3\bar R_{TT} + \bar R) g_{ij}.  \label{eq-Gauss-E}
\end{equation}
\begin{equation}
\text{curl} k_{ij} = -H_{ij}. \quad \quad \quad \quad \quad \quad \quad \label{eq-Codazzi-H}
\end{equation}
\end{subequations}
For the energy estimate scheme of the Weyl field, we will focus on the following (rescaled) $1+3$ splitting of Bianchi equations \eqref{eq-Bianchi-J}-\eqref{eq-Bianchi-J-*} (c.f. Proposition 7.2.1 in \cite{Christodoulou-K-93} or Corollary 3.2 in \cite{A-M-04})
\begin{subequations}
\begin{equation}
\begin{split}
 \lie_{\dtau} E_{ij} =&N\curl H_{ij}- \left( \nabla N \wedge H \right)_{ij} \\
 &- \frac{5N}{2} \left( E \times k \right)_{ij} - \frac{2N}{3} \left( E \cdot k \right) g_{ij}-  \frac{N}{2}\tr k E_{ij} - NJ_{i \Tbar j},  \label{eq-1+3-bianchi-T-E}
 \end{split}
\end{equation}
\begin{equation}
\begin{split}
 \lie_{\dtau} H_{ij}  =&-N \curl E_{ij} +\left( \nabla N \wedge E \right)_{ij} \\
 &- \frac{5N}{2} \left( H \times k \right)_{ij} - \frac{2N}{3} \left( H \cdot k \right) g_{ij} -  \frac{N}{2}\tr k H_{ij} - NJ^*_{i \Tbar j}. \label{eq-1+3-bianchi-T-H}
 \end{split}
\end{equation}
\end{subequations}
See Appendix for the definition of $\wedge, \times,$ div and curl. 

We shall allow ourselves to use $W$ to represent the electric part $E$ or the magnetic part $H$.
Define the $l^{\text{th}}$-order energy norm for $W$
\begin{equation}\label{def-energy-l-homo-Weyl}
E_i (W, \tau) = \int_{\Sigma_t}  \left( |\nabla_{I_i} E|_g^2 + |\nabla_{I_i} H |_g^2  \right) \di \mu_g.
\end{equation}
Let $E^{\tilde g}_i (W, t)$ be the energy norm associated to $\tilde g$, i.e. replacing $g$ by $\tilde g$ in \eqref{def-energy-l-homo-Weyl}. Then, $t^2 E_i (W, t)= t^{3+2i} E^{\tilde g}_i (W, t)$. And define
\begin{equation}\label{def-energy-density-l-homo-Weyl}
\rho_i (W) = |\nabla_{I_i} E|_g^2 + |\nabla_{I_i} H |_g^2.
\end{equation}

\subsubsection{The zeroth order energy identity}\label{sec-0-einstin-1+3}

\begin{theorem}\label{lemma-energy-id-1-Bianchi}
For solutions to the Bianchi equation in $1+3$ form \eqref{eq-1+3-bianchi-T-E}-\eqref{eq-1+3-bianchi-T-H}, there is the energy identity
\begin{equation}\label{energy-id-0-bianchi}
\begin{split}
&\dtau E_0 (W, \bar \tau) + 2E_0 (W, \bar \tau) = E_0(W, {\bar \tau}_0)  + \int_{\Sigma_{t}} \left( {}_0LW + S_0 \right) \di \mu_g,
\end{split}
\end{equation}
where the error terms $LW_0$ and the source terms $S_0$ are given by
\begin{equation}\label{def-e-s-1-bianchi}
\begin{split}
{}_0LW =&3 \hat N \left( |E|^2+|H|^2 \right) + \nabla N * E*H \\
&  \pm  N \hat k *  \left(E*E + H*H \right), \\
S_0 =& - 2 N\left( J_{i \Tbar j} E^{ij} + J^*_{i \Tbar j} H^{ij} \right).
\end{split}
\end{equation}
\end{theorem}

\begin{proof}
We multiply $2 E_{ij}$ on \eqref{eq-1+3-bianchi-T-E} and $2 H_{ij}$ on \eqref{eq-1+3-bianchi-T-H} to derive 
\begin{equation} \label{eq-1+3-bianchi-T-EH-1}
\begin{split}
& \dtau \left( |E|^2+|H|^2 \right) -\frac{2N}{3} \tr k \left( |E|^2+|H|^2 \right)  \\
=&\dive \left( 2N \left( E \wedge H \right) \right)  + 6 \nabla N \cdot \left( E \wedge H \right) \\
 &+ 4N \hat k * \left( E*E+ H*H\right) -4\hat N \left( E^2 + H^2 \right) \\
&- 3N \left( E \times \hat k \right)\cdot E -3N \left( H \times \hat k \right)\cdot H \\
&- 2N J_{i \Tbar j} E^{ij}  - 2N J^*_{i \Tbar j} H^{ij}.
\end{split}
\end{equation}
Here we have used the identity \eqref{wedge-curl-div} relating curl, $\cdot$,  div, $\wedge$ and the splitting,
\begin{equation}\label{eq-E-time-K}
\left( E \times k \right)_{ij}  =\left( E \times \hat k \right)_{ij} - \frac{\tr k}{3}E_{ij}.
\end{equation} 
Integrating on $\Sigma_{t}$, we have
\begin{equation}\label{eq-1+3-bianchi-T-EH-3}
\begin{split}
& \dtau \int_{\Sigma_{t}} \rho_0(W) \di \mu_g  + \int_{\Sigma_{t}} 2N \rho_0(W) -3\hat N \rho_0 (W) \di \mu_g \\
= & \int_{\Sigma_{t}} 4N \hat k * \left( E*E+ H*H\right) -4\hat N \left( E^2 + H^2 \right) \\
& \quad + 6 \nabla N \cdot \left( E \wedge H \right) - 2N J_{i \Tbar j} E^{ij}  - 2N J^*_{i \Tbar j} H^{ij}.
\end{split}
\end{equation}
This gives \eqref{energy-id-0-bianchi}.
\end{proof}

\subsubsection{Higher order energy identities}\label{sec-h-einstin-1+3}
In this section, we always require $l \leq 3$.
To proceed to the high order case, we need a high order version of the identity \eqref{wedge-curl-div}, which relates curl, $\cdot$,  div, $\wedge$.
\begin{lemma}\label{lemma-div-curl}
 Let  $l \geq 0$, and $H_{ij}, E_{ij}$ be any two $(0, 2)$ symmetric and trace free tensors on $\Sigma_t$. Then there is the identity
\begin{equation}\label{wedge-curl-div-general}
\begin{split}
& \nabla_{I_l} \left( \curl H \right)_{ij}  \nabla^{I_l} E^{ij} - \nabla_{I_l} \left( \curl E \right)_{ij} \nabla^{I_l} H^{ij} \\
=& \sum_{a \leq l}  \nabla_q    \left( \nabla_{I_a} H  * \nabla_{I_a} E  \right)^q \pm \left( \hat{\mathcal{R}}_{I_{l-1}}(H) * \nabla_{I_l} E + \hat{\mathcal{R}}_{I_{l-1}}(E) * \nabla_{I_l} H \right),
\end{split}
\end{equation}
where $\hat{\mathcal{R}}_{I_{l-1}}(A)$ is defined as 
\begin{equation}\label{Def-hat-R-div-curl}
\begin{split}
 \hat{\mathcal{R}}_{I_{l-1}}(A) =& \sum_{a+b = l-1} \nabla_{I_a}  \hat R_{imjn}  * \nabla_{I_b} A, 
\end{split}
\end{equation}
with $\hat R_{imjn} = R_{imjn} +( g \odot g )_{imjn} $ being the error term in the Riemanian curvature. By virtue of \eqref{Gauss-Riem-hat-k},
\begin{equation}\label{def-hat-R}
\begin{split}
\hat R_{imjn} =& ( g \odot \hat k )_{imjn}  - \frac{1}{2} ( \hat k \odot \hat k )_{imjn}   + \bar{R}_{imjn}.
\end{split}
\end{equation}
\end{lemma}
The proof for this Lemma is presented in the Appendix \ref{sec-app-id}.

Theorem \ref{lemma-energy-id-1-Bianchi} can be generalized to the high order case as below.
\begin{corollary}\label{coro-energy-id-l-Bianchi}
For solutions to the Bianchi equations in $1+3$ form \eqref{eq-1+3-bianchi-T-E}-\eqref{eq-1+3-bianchi-T-H} with $l \geq 1$, there is the energy identity
\begin{equation}\label{energy-id-l-bianchi}
\begin{split}
 & \dtau E_l (W, \bar \tau) + 2 E_l(W, \bar \tau) =  \int_{\Sigma_{t}} \left(  {}_lLW_1 + {}_lLW_2  + {}_lLW_3  + S_l  \right) \di \mu_g,
\end{split}
\end{equation}
where
\begin{equation}\label{def-LB-1-2-S}
\begin{split}
{}_lLW_1 = & \left( \mathcal{K}\mathcal{N}_{I_l}(W) - \hat{\mathcal{K}}\mathcal{N}_{I_l} (W) \right) \nabla^{I_l} W, \\
{}_lLW_2 = &   - \hat N |\nabla_{I_l} W|^2 \pm \sum_{a \leq l} \mathcal{N}_{I_{a}}(W) \nabla^{I_a} W, \\
{}_lLW_3 = &  \pm \hat{\mathcal{R}}_{I_{l-1}} (W)  \nabla^{I_l} W, \\
S_l=& - \nabla_{I_l} \left( N J_{i \Tbar j}\right)  \nabla^{I_l} E^{ij} - \nabla_{I_l} \left( N J^*_{i \Tbar j}\right) \nabla^{I_l} H^{ij}.
\end{split} 
\end{equation}
We refer to \eqref{def-commuting-K-R-T-l-tensor}, \eqref{def-eq-kg-N-l-1+3-error}-\eqref{def-hat-k-N-A-Bianchi}  and \eqref{Def-hat-R-div-curl}  for the definition of $\mathcal{K}\mathcal{N}_{I_l}(E_{ij}),$ $ \hat{\mathcal{R}}_{I_{l-1}} (E_{ij}),$ $\mathcal{N}_{I_l} (E_{ij}), \hat{\mathcal{K}}\mathcal{N}_{I_l} (E_{ij})\cdots$
\end{corollary}

\begin{proof}
Perform $\nabla_{I_l}$ on \eqref{eq-1+3-bianchi-T-E}-\eqref{eq-1+3-bianchi-T-H}, and apply the commuting identity \eqref{id-commuting-nabla-N-T-l-simplify-tensor} to $\lie_{\dtau} \nabla_{I_l} E_{ij}$ and $\lie_{\dtau} \nabla_{I_l} H_{ij}$. Then,
\begin{align*}
& \lie_{\dtau} \nabla_{I_l}  E_{ij} =  \nabla_{I_l} \left( N \curl H_{ij}  + \frac{N}{3} \tr k E_{ij}  \right) + \mathcal{K}\mathcal{N}_{I_l}(E_{ij})\\
& \quad -  \nabla_{I_l} \left( \left( \nabla N \wedge H \right)_{ij} + \frac{5N}{2} \left( E \times \hat k \right)_{ij} + \frac{2N}{3} \left( E \cdot \hat k \right) g_{ij}+ N J_{i \Tbar j} \right).
\end{align*}
Here the splitting \eqref{eq-E-time-K} is used. We rearrange and collect similar terms
\begin{equation}\label{eq-evolu-E-1}
\begin{split}
 \lie_{\dtau} \nabla_{I_l}  E_{ij} =&  N \nabla_{I_l} \left( \curl H_{ij}  \right) + \frac{N}{3} \tr k \nabla_{I_l} E_{ij}  - \nabla_{I_l} \left( N J_{i \Tbar j}\right)\\
&+ \mathcal{K}\mathcal{N}_{I_l}(E_{ij}) \pm \mathcal{N}_{I_{l}}(H_{ij})  - \hat{\mathcal{K}}\mathcal{N}_{I_l} (E_{ij}).
\end{split} 
\end{equation}
Multiplying $2 \nabla^{I_l} E^{ij}$ on \eqref{eq-evolu-E-1},  we derive the transported equation for $|\nabla_{I_l} E|^2,$
\begin{align*}
&\dtau \left( |\nabla_{I_l} E|^2 \right) -\frac{2N}{3} \tr k |\nabla_{I_l} E|^2\\
=& 2N \nabla_{I_l} \curl H_{ij} \nabla^{I_l} E^{ij}  - \nabla_{I_l} \left( N J_{i \Tbar j}\right)  \nabla^{I_l} E^{ij} -\hat N \nabla_{I_l} E * \nabla^{I_l} E \\
&+\left(  \mathcal{K}\mathcal{N}_{I_l}(E_{ij}) \pm \mathcal{N}_{I_{l}}(H_{ij})  - \hat{\mathcal{K}}\mathcal{N}_{I_l} (E_{ij}) \right) \nabla^{I_l} E^{ij}.
\end{align*}
We had ignore irrelevant constant.
Similar argument also applies to \eqref{eq-1+3-bianchi-T-H}.

Putting the transported equations for $ |\nabla_{I_l} E|^2$ and $ |\nabla_{I_l} H|^2$ together,  we have
\begin{equation*}
\begin{split}
&\dtau \left( |\nabla_{I_l} E|^2 + |\nabla_{I_l} H|^2 \right) - \frac{2N}{3} \tr k \left( |\nabla_{I_l} E|^2 + |\nabla_{I_l} H|^2 \right)  \\
=&  2N \left( \nabla_{I_l} \curl H_{ij} \nabla^{I_l} E^{ij}  - \nabla_{I_l} \curl E_{ij} \nabla^{I_l} H^{ij} \right) \\
&+\left(  \mathcal{K}\mathcal{N}_{I_l}(W) \pm \mathcal{N}_{I_{l}}(W)  - \hat{\mathcal{K}}\mathcal{N}_{I_l}(W) \right) \nabla^{I_l} W - \hat N \nabla_{I_l} W * \nabla^{I_l} W   \\
& - \nabla_{I_l} \left( N J_{i \Tbar j}\right)  \nabla^{I_l} E^{ij} - \nabla_{I_l} \left( N J^*_{i \Tbar j}\right)  \nabla^{I_l} H^{ij}.
\end{split} 
\end{equation*}
In views of the identity \eqref{wedge-curl-div-general} relating curl, $\cdot$ and div, $\wedge$,  the second line above becomes
\begin{align*}
& N \left( \nabla_{I_l} \curl H_{ij} \nabla^{I_l} E^{ij} - \nabla_{I_l} \curl E_{ij} \nabla^{I_l} H^{ij} \right)\\
=& div \pm \hat{\mathcal{R}}_{I_{l-1}} (W) * \nabla^{I_l} W \pm \sum_{a \leq l}  \nabla N * \nabla_{I_a} E *  \nabla_{I_a} H.
\end{align*}
Here $div$ denotes divergence forms. Integrating on $\Sigma_{t}$, we have
\begin{equation}\label{eq-1+3-bianchi-T-EH-3}
\begin{split}
& \dtau \int_{\Sigma_{t}} \rho_l(W) \di \mu_g  + \int_{\Sigma_{t}} 2N \rho_l(W) \\
= & \int_{\Sigma_{t}} \left( {}_lLW_1 + {}_lLW_2  + {}_lLW_3 + S_l \right) \di \mu_g.
\end{split}
\end{equation}
where the lower order terms ${}_lLW_1, \cdots, {}_lLW_3$ and source terms $S_l$ are defined in \eqref{def-LB-1-2-S}. 
\end{proof}

\section{Preliminary estimate}\label{sec-preliminary-estimate}

Let us remind the partial energy norm for the KG field $E_i(\phi, t)$ \eqref{def-energy-l-homo-kg-g} and for the Weyl tensor $E_i(W, t)$ \eqref{def-energy-l-homo-Weyl}.
We also define the energy norm up to $l$ order as follows,
\begin{equation}\label{def-energy-sum-inhomo}
\mathcal{E}_l (\phi, t) = \sum_{i\leq l} E_i(\phi, t), \quad \mathcal{E}_l (W, t) = \sum_{i \leq l} E_i(W, t).
\end{equation}
Besides, the energies for $\hat k$ are defined by
\begin{equation}\label{def-energy-hat-k}
E_i (\hat k, t )= \int_{\Sigma_t}  |\nabla_{I_i} \hat k|^2 \di \mu_g, \quad \mathcal{E}_l(\hat k, t )=\sum_{i \leq l} E_i(\hat k, t ).
\end{equation}

\subsection{The continuity argument}\label{sec-Boot-assum}
Fix a constant $0<\delta < \frac{1}{6}$ and let $M$ be a large constant to be determined.
We start with the following weak assumptions:
The bootstrap assumptions for the KG field $\phi$,
\begin{equation}\label{BT-KG}
 t E_i(\phi, t) \leq \varepsilon^2M^2 t^{2\delta}, \quad i =1,\cdots, 4;
\end{equation}
The bootstrap assumptions for the Weyl tensor $E(W), H(W)$, 
\begin{equation}\label{BT-E-H}
t^2 E_i(W, t) \leq \varepsilon^2M^2 t^{2\delta}, \quad i=1, \cdots, 3;
\end{equation}
The bootstrap assumptions for the $\hat k$ and $\hat N = N-1$,
\begin{equation}\label{BT-k-N-L-infty}
t^2 E_i (\hat k, t ) \leq  \varepsilon^2  M^2 t^{2\delta}, \, 0 \leq i \leq 4, \quad t^2 \|\hat N\|^2_{L^\infty} \leq \varepsilon^2 M^2 t^{4\delta}.
\end{equation}

The smallness condition for KG data \eqref{intro-initial-data} given in Theorem \ref{main-thm-into} entails that \eqref{BT-KG} holds for $t=t_0$.
Conducting the elliptic estimates for $\hat N$ (c.f. the second (elliptic) equation in \eqref{eq-hat-lapse}) on the initial slice, we further show that the requirement for data \eqref{intro-initial-data} implies that \eqref{BT-k-N-L-infty} holds on $\{ t = t_0\}$. Moreover, in views of the Gauss equations \eqref{Gauss-Riem-hat-k}-\eqref{Gauss-Ricci-hat-k} and the Weyl tensor \eqref{def-Weyl}, \eqref{intro-initial-data} implies that \eqref{BT-E-H} holds for $t=t_0.$ Let $[t_0, t']$ be the largest time interval on which \eqref{BT-KG}-\eqref{BT-k-N-L-infty} still holds. We shall show that if $\varepsilon$ is sufficiently small (depending on initial data and $\delta$), then on $[t_0, t']$,  \eqref{BT-KG}-\eqref{BT-k-N-L-infty} implies the same inequality with the constant $\varepsilon^2 M^2$ being replaced by $\frac{\varepsilon^2 M^2}{2}$, c.f. Section \ref{sec-close-b-t}. It will then follow that the solution and the energy estimate \eqref{BT-KG}-\eqref{BT-k-N-L-infty} can be extended to a larger time interval $[t_0,T^\prime]$, thus contradicting the maximality of $t'$. This will imply that $t'=\infty$ and the solution is global (see the formulation of Theorem \eqref{main-thm-global}). We will in fact prove that for a sufficiently small $\varepsilon$, the stronger estimate \eqref{stronger-estimate-closing-bt} (Theorem \eqref{main-thm-global}) holds true on the interval $[0,t']$. From now on, we always assume $t \in [t_0, t']$.

We use the following convention for initial data for $\psi$, 
\begin{equation}\label{def-initial-convention}
\varepsilon^2 I^2_l(\psi): = t_0 \mathcal{E}_l(\psi, t_0).
\end{equation} 
In application, $\psi$ will be taken as $\phi, \hat k$ or $W$.
At the first stage, we choose $M$ large enough so that $I_4(\phi)$, $I_3(W)$, $I_4(\hat k) < 2M,$
and $\varepsilon$ is chosen small enough such that $\varepsilon M <1$.

\subsection{Sobolev inequality}\label{sec-Sobolev}
Recall that $g_{ij}=t^{-2} \tilde g_{ij}$ is the scale-free metric, with $\nabla$ the corresponding connection. 
Throughout the paper, we simplify the notation $L^p(\Sigma ,g)$ by $L^p(\Sigma_t), \, p>0$, and $\int_{\Sigma} \di \mu_g$ by $\int_{\Sigma_t} $. 
By the Sobolev embedding theorem on $(\Sigma, g)$, we derive the following preliminary estimates: for $\phi,$ (recalling that $\D \phi \in \{\Tbar \phi, \nabla \phi, \bar m(t) \phi\}$)
\begin{equation}\label{BT-KG-L-infty-4-6}
\begin{split}
 \sum_{j \leq 2} \|\nabla^j \D \phi \|_{L^\infty} + \sum_{i \leq 3} \|\nabla^i \D \phi \|_{L^4(\Sigma_t)}  + \|\nabla^i \D \phi \|_{L^6(\Sigma_t)} \lesssim \mathcal{E}_4 (\phi, t).
\end{split}
\end{equation}
For $W \in \{E, H\}$ \eqref{BT-E-H}, 
\begin{equation}\label{BT-Weyl-L-infty-4-6}
\begin{split}
\sum_{j \leq 1} \|\nabla^j W\|_{L^\infty}  +  \sum_{i \leq 2}  \|\nabla^i W\|_{L^4(\Sigma_t)} + \|\nabla^i W\|_{L^6(\Sigma_t)} & \lesssim \mathcal{E}_3 (W, t). 
\end{split}
\end{equation}
The bootstrap assumption for $\hat k$ \eqref{BT-k-N-L-infty} leads to
\begin{equation}\label{BT-k-L-infty-4-6}
\begin{split}
\sum_{j \leq 2} \|\nabla^j \hat k\|_{L^\infty}  +  \sum_{i \leq 3} \|\nabla^i \hat k \|_{L^4(\Sigma_t)} +  \|\nabla^i \hat k \|_{L^6(\Sigma_t)}  \lesssim  \mathcal{E}_4 (\hat k, t).
\end{split}
\end{equation}

\begin{remark}
Let $g_0 = t^{-2}_0 \tilde g_0$ be the initially normalized metric. 
Using the estimate for $\| \hat k\|_{L^\infty},$ $\| \hat N\|_{L^\infty}$ \eqref{BT-k-N-L-infty}, \eqref{BT-k-L-infty-4-6} and the evolution equation for $g_{ij}$ \eqref{eq-evolution-1}, we are able to prove that  $ g$ is close to $ g_0$: $| g_{ij} -  g_{0ij}| \lesssim \varepsilon M.$
And the volumes of $g_{ij}$ and $g_{0ij}$ are equivalent: $\exp(-\varepsilon M) \lesssim \frac{\di \mu_{g}}{ \di \mu_{g_0}} \lesssim \exp(\varepsilon M)$. Thus, there exists a uniform constant bounding all the Sobolev constants for $g_{ij}$.
\end{remark}

\subsection{The lapse}\label{sec-N}
In this section, we present the elliptic estimates for the lapse. In particular, we highlight the subtle relation between the lapse and the KG field, which shall server as a key point in establishing the hierarchy of energy estimates for the KG field, c.f. Section \ref{sec-0-ee-kg} .

The equation for lapse $N$ can be derived by taking the trace of the transport equation for $\hat k_{ij}$ \eqref{eq-evolution-2}. 
$\hat N$ defined in \eqref{def-hat-N} should be regarded as small quantity in the small data situation we are considering.
\begin{equation}\label{eq-hat-lapse}
\Delta \hat N - 3 \hat N = N\left( \bar R_{\Tbar \Tbar} + | \hat k |^2 \right).
\end{equation}

Let us turn to the elliptic estimates for $\hat N$.
Define the energy for $\hat N$:
\begin{equation}\label{def-energy-hat-N}
E_i (\hat N, t )= \int_{\Sigma_t}  |\nabla^i \hat N|_g^2 \di \mu_g, \quad \mathcal{E}_l(\hat N, t )=\sum_{i \leq l} E_i(\hat N, t ).
\end{equation}
Again we have $E_i (\hat N, t ) =t^{-3 + 2i} E^{\tilde g}_i (\hat N, t ),$ where $E^{\tilde g}_i (\hat N, t )$ is defined by replacing the $g$ in \eqref{def-energy-hat-N} by $\tilde g$.

\begin{lemma}\label{lem-N-ineq-source}
With the bootstrap assumptions \eqref{BT-KG}-\eqref{BT-k-N-L-infty}, we have 
 \begin{equation}\label{N-inte-ineq-123-lemma}
 t^2 E_i(\hat N, t ) \lesssim  \varepsilon^2 M^2 t^{4\delta}, \quad 0<\delta<\frac{1}{6}, \quad i \leq 5.
\end{equation}
Moreover, $\hat N$ depends on the KG field in the exact manner as follows: letting 
 \begin{equation}\label{lemma-N-inte-ineq-nabla3-source}
\|\Tbar \phi\|_{L^\infty}^2 + \|\bar m (t) \phi\|_{L^\infty}^2:= r^2(\phi, t),
\end{equation}
we have
\begin{align}
& \mathcal{E}_2(\hat N, t ) \lesssim   \mathcal{E}^2_{1}(\hat k,t) + r^2(\phi, t) \mathcal{E}_0(\phi, t), \label{lemma-N-inte-ineq-2order-source} \\
&E_3(\hat N, t )  \lesssim  \mathcal{E}^2_{2}(\hat k,t) +  r^2(\phi, t) \mathcal{E}_1(\phi, t), \label{lemma-N-inte-ineq-nabla3-source}\\
 & E_4(\hat N, t )  \lesssim  \mathcal{E}^2_{3}(\hat k,t) + \mathcal{E}^2_2(\phi, t) +  r^2(\phi, t) \mathcal{E}_2(\phi, t), \label{lemma-N-inte-ineq-nabla4-source}\\
&E_5(\hat N, t )  \lesssim  \mathcal{E}^2_{4}(\hat k,t) + \mathcal{E}_2(\phi,t)\mathcal{E}_3(\phi,t)  + r^2(\phi, t)) \mathcal{E}_3(\phi,t).\label{lemma-N-inte-ineq-nabla5-source}
\end{align}
\end{lemma}
\begin{remark}\label{rk-rough-lapse}
If we do not take advantage of the special structure of the source term, the ellipticity of the equation for the lapse only gives \eqref{lemma-N-inte-ineq-2order-source} and
\begin{equation*}
\begin{split}
\mathcal{E}_{l+2}(\hat N, t )  \lesssim &\mathcal{E}^2_4(\phi,t) + \mathcal{E}^2_{l+1}(\hat k,t) +  \mathcal{E}_{l+1}(\hat N, t ), \quad l \leq 3.
\end{split}
\end{equation*}
Hence, we merely derive $\mathcal{E}_{5}(\hat N, t )  \lesssim \mathcal{E}^2_4(\phi,t)$, ignoring lower order terms.
\end{remark}
\begin{remark}\label{rk-CMCSH gauge-shift}
If one works within the CMCSH gauge, the elliptic equation for the shift $X$ interacts with the matter current $\mathcal{T}_{\Tbar}^{i}(\phi) = \nabla^i \phi \Tbar \phi$:
\begin{equation}\label{eq-shift}
\begin{split}
\Delta X^i + R^i_j X^j =& -2 \nabla_j N \hat k^{ij} - \frac{\tr k}{3} \nabla^i N + 2 N  \nabla^i \phi \Tbar \phi \\
& - \left( 2N \hat k ^{pq}-\nabla^p X^q  \right) \left(\Gamma_{pq}^i(g) - \Gamma_{pq}^i(\gamma) \right).
\end{split}
\end{equation}
 Moreover, $\nabla^i \phi \Tbar \phi \sim t^{-1} \nabla^i \left(\bar m (t) \phi\right) \Tbar \phi.$ Therefore, the leading term on the right hand side of \eqref{eq-shift} is $\nabla^i N$, and we can expect the following estimate
 \begin{equation}\label{estimate-shift}
\begin{split}
\mathcal{E}_5\left(X, t\right) \lesssim \mathcal{E}_4(\hat N, t) + t^{-2}\mathcal{E}^2_4(\phi, t) + l.o.t.,
\end{split}
\end{equation}
where $\mathcal{E}_5\left(X, t\right)$ is defined in the same way as \eqref{def-energy-hat-N}.
\end{remark}

\begin{proof}
We multiply $\hat N$ on \eqref{eq-hat-lapse} to derive
 \begin{equation}\label{N-inte-1}
\begin{split}
& \int_{\Sigma_t} |\nabla \hat N|^2 + 3 |\hat N|^2 + |\hat N|^2 |\hat k|^2 \\
=& - \int_{\Sigma_t} \hat N \left( |\hat k|^2 + \bar R_{\Tbar \Tbar} \right) + |\hat N|^2 \bar R_{\Tbar \Tbar}.
\end{split}
\end{equation}
Making use of the $L^\infty$ estimate for $\phi$ \eqref{BT-KG-L-infty-4-6}, the last term, $\bar R_{\Tbar \Tbar} = (\Tbar \phi)^2-\frac{\bar m^2(t)}{2} \phi^2$, admits
 \begin{equation*}
  \int_{\Sigma_t}  |\hat N|^2 |\bar R_{\Tbar \Tbar}| \lesssim  \int_{\Sigma_t} |\hat N|^2 \left( |\Tbar \phi|^2 + |\bar m (t) \phi|^2 \right) \lesssim  \int_{\Sigma_t} \varepsilon^2 M^2 t^{-1 + 2\delta} |\hat N|^2.
\end{equation*}
And the Cauchy-Schwarz inequality shows that for some constant $c,$ 
 \begin{equation*}
 \int_{\Sigma_t} \hat N \left( |\hat k|^2 + |\bar R_{\Tbar \Tbar}| \right) 
\leq \int_{\Sigma_t} c |\hat N|^2 + \int_{\Sigma_t} c^{-1} \left( |\hat k|^2 + |\bar R_{\Tbar \Tbar}| \right)^2.
\end{equation*}
We choose the constant $c<2$ so that $\int_{\Sigma_t} \left( c + \varepsilon^2 M^2 t^{-1 + 2\delta} \right) |\hat N|^2$ is absorbed by the left hand side of \eqref{N-inte-1}. 
 \begin{equation}\label{estimate-nabla-N-1}
 \begin{split}
 &\int_{\Sigma_t} |\nabla \hat N|^2 +  |\hat N|^2 
\lesssim \int_{\Sigma_t}  \left( |\hat k|^2 + |\bar R_{\Tbar \Tbar}| \right)^2 \\
& \lesssim  \mathcal{E}^2_1(\hat k, t) + \int_{\Sigma_t} \left(|\Tbar \phi|_{L^\infty}^2 + |\bar m (t) \phi|_{L^\infty}^2 \right) \left( |\Tbar \phi|^2 + |\bar m (t) \phi|^2 \right).
\end{split}
\end{equation}
Substituting the bootstrap assumption for $\hat k$, and the assumed $L^2$ estimate for $\Tbar \phi$ and $\bar m (t) \phi$, we have proved \eqref{N-inte-ineq-123-lemma} up to 1st order.

Concerning the higher order cases, we aim to prove by induction that
\begin{equation}\label{estimate-nabla2-hat-N-i-induction}
\int_{\Sigma_t}  t^2  |\nabla_{I_{i}} \hat N|^2 \lesssim \varepsilon^2 M^2 t^{4\delta}, \quad 0\leq i \leq 5.
\end{equation}
It has been proved that \eqref{estimate-nabla2-hat-N-i-induction} holds for $i=0, 1.$ Suppose \eqref{estimate-nabla2-hat-N-i-induction} holds for $i\leq l+1, \, l \leq 3$, we shall now prove that \eqref{estimate-nabla2-hat-N-i-induction}  holds for $i=l+2$ as well. 
 
Applying $\nabla_{I_l} (0 \leq l \leq 3)$ to the equation for lapse \eqref{eq-hat-lapse}, 
\begin{equation}\label{eq-lapse-nabla-l}
\begin{split}
\nabla_{I_l} \Delta \hat N  =&  3 \nabla_{I_l} N + \sum_{a+b=l} \nabla_{I_a} N * \nabla_{I_b} \left( \bar R_{\Tbar \Tbar} + | \hat k |^2 \right).
 \end{split}
\end{equation}
Noting the commuting identity between $\nabla_{I_l}$ and $\Delta$ \eqref{commuting-nabla-laplacian-l-simplify}-\eqref{def-R-l-commute-nabla-laplacian},
we integrate by parts to derive that for $0 \leq l \leq 3,$
\begin{equation*}
\begin{split}
 \int_{\Sigma_t}   |\nabla_{I_{l+2}} \hat N|^2  = &  \int_{\Sigma_t} | \nabla_{I_l} \Delta \hat{N}|^2  +  \mathcal{R}_{I_l} ( \hat N ) \nabla_{I_l} \Delta \hat N -  \mathcal{R}_{I_{l+1}} ( \hat N)  \nabla_{I_{l+1}} \hat N.
\end{split}
\end{equation*}
See \eqref{def-R-l-commute-nabla-laplacian} for the definition of $\mathcal{R}_{I_l} (\hat N )$.
Thus, 
\begin{equation*}
\begin{split}
 & \int_{\Sigma_t} |\nabla_{I_{l+2}} \hat N|^2  \lesssim   \int_{\Sigma_t} |\nabla_{I_l} \Delta \hat{N}|^2+| \mathcal{R}_{I_l} (\hat N ) |^2 +\big| \mathcal{R}_{I_{l+1}} (\hat N ) * \nabla_{I_{l+1}} \hat N\big|.
\end{split}
\end{equation*}
By \eqref{eq-lapse-nabla-l} and Cauchy-Schwarz inequality,
\begin{equation*}
\begin{split}
  \int_{\Sigma_t} |\nabla_{I_{l+2}} \hat N|^2  \lesssim&   \int_{\Sigma_t}  |\nabla_{I_l} \hat{N}|^2 +  |\mathcal{R}_{I_l} (\hat N ) |^2  + | \mathcal{R}_{I_{l+1}} (\hat N ) * \nabla^{I_{l+1}} \hat N |  \\
 &+ \int_{\Sigma_t} \sum_{a+b =l} |\nabla_{I_a} N * \nabla_{I_b} \bar R_{\Tbar \Tbar} |^2+ |\nabla_{I_a} N * \nabla_{I_b} (\hat k)^2|^2.
\end{split}
\end{equation*}
We derive
\begin{equation}\label{ineq-E-l+2-source}
\begin{split}
\int_{\Sigma_t}  |\nabla_{I_{l+2}} \hat N|^2  \lesssim&   \int_{\Sigma_t}  |\nabla_{I_{l}} \hat{N}|^2 \\
 & + SN_1 +SN_2+SN_3 + RN_1 + RN_2,
\end{split}
\end{equation}
where
\begin{equation}
\begin{split}
 RN_1 = & \sum_{a+b =l \atop a \leq l-1} | \nabla_{I_a} R_{imjn} * \nabla_{I_b} \hat N |^2, \\
RN_2= & \sum_{a+b =l+1 \atop a \leq l}  |\nabla_{I_a} R_{imjn} * \nabla_{I_b} \hat N * \nabla_{I_{l+1}} \hat N |, \\
 SN_1 = &\sum_{a+b+c =l}   \big| \nabla_{I_a} N * \nabla_{I_{b}} \Tbar \phi * \nabla_{I_c} \Tbar\phi \big|^2,  \\
SN_2 = & \sum_{a+b+c =l}  \big| \nabla_{I_a} N * \nabla_{I_{b}} (\bar m (t) \phi) * \nabla_{I_c} (\bar m (t) \phi) \big|^2,  \\
SN_3 = & \sum_{a+b+c =l}   \big| \nabla_{I_a} N * \nabla_{I_{b}} \hat k* \nabla_{I_c}\hat k  \big|^2.
\end{split}
\end{equation}
Viewing \eqref{Gauss-Riem-hat-k}-\eqref{Gauss-Ricci-hat-k} and \eqref{def-Weyl}, the expansion for $R_{imjn}$ is rewritten as
$$R_{imjn} =  g*g \pm  \hat k *g \pm \hat k* \hat k +W_{imjn} \pm \bar R_{pq} *g \pm \bar R *g*g.$$
In the estimates, we only need the formula $\|A*B\| \lesssim \|A\| \|B\|$, while the detailed product structure and constants are irrelevant. As a $(0, 4)$-tensor on $\Sigma_t$, $R_{imjn}$ can be simplified as
\begin{equation}\label{R-imjn-E}
R_{imjn} \sim 1 \pm \hat k \pm |\hat k|^2 + E \pm |\D \phi|^2.
\end{equation}

By the assumption, \eqref{estimate-nabla2-hat-N-i-induction} holds for $i\leq l +1, l \leq 3$,
namely,
\begin{equation}\label{ineq-SN-0}
t^2 \mathcal{E}_{l+1} (\hat N, t) \lesssim \varepsilon^2 M^2 t^{4\delta}.
 \end{equation}

Since $SN_1$ and $SN_2$ can be treated in the same way, we only take $SN_2$ for example., 
\begin{itemize}
\item {\bf Case $SN_{21}$}: $a=0$. We always apply $L^4$ to all the four factors in $SN_2,$ for $l \leq 3.$ That is, in this case, $SN_2$ becomes ($b+c = l$)
\begin{equation*}
\begin{split}
 & \int_{\Sigma_t} \big|N \nabla_{I_{b}} (\bar m (t) \phi) * \nabla_{I_c} (\bar m (t) \phi) \big|^2
  \lesssim  \mathcal{E}^2_{l+1}(\phi,t). 
\end{split}
\end{equation*}
This estimate will be treated in a more subtle way later.
\item {\bf Case $SN_{22}$}: $a \geq 1$. Then  $0 \leq b,c \leq l-1 \leq 2$. We can apply $L^6$ to all the six factors in $SN_2$ ($a+b+c =l, a \geq 1$)
\begin{equation*}
\begin{split}
 & \int_{\Sigma_t}  \big|( t\nabla)_{I_a } \hat N * \nabla_{I_{b}} (\bar m (t) \phi) * \nabla_{I_c} (\bar m (t) \phi) \big|^2 
  \lesssim \mathcal{E}_{l+1}(\hat N,t) \mathcal{E}^2_{l}(\phi,t).  
\end{split}
\end{equation*}
\end{itemize}
Similar procedures apply to $SN_3$. Roughly speaking, $\hat k$ enjoys better estimates than $\phi$, and we derive 
\begin{equation}\label{esti-SN3}
\begin{split}
\int_{\Sigma_t} SN_3 \lesssim & \mathcal{E}^2_{l+1}(\hat k,t) + \mathcal{E}_{l+1}(\hat N,t) \mathcal{E}^2_{l}(\hat k,t).
\end{split}
\end{equation}

We summarize these estimates as
\begin{equation}\label{esti-SN123}
\begin{split}
\int_{\Sigma_t} \sum_{i=1}^3 |SN_i|  \lesssim &  \mathcal{E}^2_{l+1}(\phi,t) +  \mathcal{E}^2_{l+1}(\hat k,t) \\
& + \mathcal{E}_{l+1}(\hat N,t) \left(\mathcal{E}^2_{l}(\hat k,t) + \mathcal{E}^2_{l}(\phi,t) \right)  \lesssim  \varepsilon^2 M^2 t^{-2+ 4\delta}.
\end{split}
\end{equation}

For the terms involving curvature $RN_1$, 
\begin{itemize}
\item {\bf Case $RN_{11}$}: $a=0, b=l$. In views of the expansion for $R_{imjn}$ \eqref{R-imjn-E}, and the $L^\infty$ estimates for $\hat k, \phi$, $RN_{11}$ is estimated as
 \begin{equation*}
 \begin{split}
&\int_{\Sigma_t}  \left( 1 + |\hat k|^2 + |\hat k|^4 + |\D \phi|^4 + |E|^2 \right) |\nabla_{I_l} \hat N|^2\\
\lesssim & \mathcal{E}_{l}(\hat N,t) \left( \mathcal{E}^2_{2}(\hat k,t)  + \mathcal{E}_{2}(\hat k,t)  + \mathcal{E}_{2}(W,t) + \mathcal{E}^2_{2}(\phi,t) +1 \right).  
\end{split}
\end{equation*}
\item {\bf Case $RN_{12}$}: $1 \leq a \leq l-1.$ Then $1 \leq b \leq l-1,$ and  $RN_{12}$ becomes ($a+b =l, 1 \leq a ,b \leq l-1, l \leq 3$)
 \begin{equation*}
 \begin{split}
 & \big| \nabla_{I_a}  \left( |\hat k| + |\hat k|^2  + |\D \phi|^2  + E \right)\big|^2  \big|\nabla_{I_b} \hat N |^2.
\end{split}
\end{equation*}
We can always apply $L^4$ to all the four factors or $L^6$ to all the six factors of each term in $RN_{12}$, so that
 \begin{equation*}
 \begin{split}
 \int_{\Sigma_t} |RN_{12}| \lesssim & \mathcal{E}_l(\hat N,t)  \left( \mathcal{E}^2_{l}(\phi,t) +\mathcal{E}^2_{l}(\hat k,t) +  \mathcal{E}_{l}(W,t) + \mathcal{E}_{l}(\hat k,t)  \right).
\end{split}
\end{equation*}
\end{itemize}

For the term involving curvature $RN_2$, 
\begin{itemize}
\item {\bf Case $RN_{21}$}: $a=0, b=l+1$. Analogous to $RN_{11}$, $RN_{21}$ shares the estimate
 \begin{equation*}
 \begin{split}
&\int_{\Sigma_t} |RN_{21}| \lesssim  \mathcal{E}_{l+1}(\hat N,t). 
\end{split}
\end{equation*}
\item {\bf Case $RN_{22}$}:  $1 \leq a \leq l.$ Then $1 \leq b \leq l,$ and $RN_{22}$ is ($a+b =l+1, 1 \leq a,b \leq l, l \leq 3$) 
 \begin{equation*}
 \begin{split}
& \nabla_{I_a}  \left(|\hat k| + |\hat k|^2 + |\D \phi|^2 + |E| \right) * \nabla_{I_b} \hat N *\nabla_{I_{l+1}} \hat N.
\end{split}
\end{equation*}
\begin{enumerate}
\item If $a \leq 2$, then similar to $RN_{12}$, we apply $L^4, L^4, L^2$ to the three factors or $L^6, L^6, L^6, L^2$ to the four factors in each term of $RN_{22}$. 
so that
 \begin{equation*}
 \begin{split}
 \int_{\Sigma_t} |RN_{22}| \lesssim & \mathcal{E}_{l+1}(\hat N,t)  \left( \mathcal{E}_{a+1}(\phi,t) +\mathcal{E}_{a+1}(\hat k,t) \right) \\
 & + \mathcal{E}_{l+1}(\hat N,t) \left( \mathcal{E}^{\frac{1}{2}}_{a+1}(W,t) + \mathcal{E}^{\frac{1}{2}}_{a+1}(\hat k,t)  \right).
\end{split}
\end{equation*}
\item If $a=3$, then $l=3$ and $b=1$. We apply $L^2, L^\infty, L^2$ to the three factors or $L^4, L^4, L^\infty, L^2$ to the four factors in each term of $RN_{22}$ to derive (noting that here $l=3$)
 \begin{equation*}
 \begin{split}
 \int_{\Sigma_t} |RN_{22}| \lesssim & \mathcal{E}^{\frac{1}{2}}_{l+1}(\hat N,t)  \mathcal{E}^{\frac{1}{2}}_{l}(\hat N,t)  \left( \mathcal{E}_{4}(\phi,t) +\mathcal{E}_{4}(\hat k,t) \right) \\
 & + \mathcal{E}^{\frac{1}{2}}_{l+1}(\hat N,t)  \mathcal{E}^{\frac{1}{2}}_{l}(\hat N,t)  \left( \mathcal{E}^{\frac{1}{2}}_{3}(W,t) + \mathcal{E}^{\frac{1}{2}}_{3}(\hat k,t)  \right).
\end{split}
\end{equation*}
\end{enumerate}
\end{itemize}
 In summary, $RN_1, RN_2$ enjoy the following estimate
\begin{equation}\label{esti-RN1-RN2}
\begin{split}
\int_{\Sigma_t} |RN_1| + |RN_2| \lesssim & \mathcal{E}_{l+1}(\hat N,t) \left( 1+  \varepsilon M t^{-1+ 2\delta} \right) \lesssim \varepsilon^2 M^2 t^{-2+4\delta}.
\end{split}
\end{equation}

Putting all those estimates \eqref{esti-SN123}-\eqref{esti-RN1-RN2} together, we can close the inductive argument  and complete the proof for  \eqref{N-inte-ineq-123-lemma}. In particular, we have proved that for $0 \leq l \leq 3,$
\begin{equation}\label{estimate-lapse-nabla-l-3-4-5-1}
\begin{split}
E_{l+2}(\hat N,t)  \lesssim &  \int_{\Sigma_t} \sum_{b+c =l} \big|N \nabla_{I_{b}}\Tbar \phi * \nabla_{I_c} \Tbar \phi \big|^2\\
 &+\int_{\Sigma_t} \sum_{b+c =l} \big|N \nabla_{I_{b}} (\bar m (t) \phi) * \nabla_{I_c}  (\bar m (t) \phi) \big|^2\\
 & + \mathcal{E}_{l+1}(\hat N,t) \left(1+  \varepsilon M t^{-1+ 2\delta}\right) +  \mathcal{E}^2_{l+1}(\hat k,t).
 \end{split}
\end{equation}

In the above proof, we find out that when $l=0$, both $RN_1$ and $RN_{22}$ are absent. More precisely, we have, in views of \eqref{estimate-nabla-N-1},
 \begin{equation}\label{ineq-hat-N-2order-source}
\begin{split}
& E_2(\hat N,t)  \lesssim \mathcal{E}_{1}(\hat N,t) +  \mathcal{E}^2_{1}(\hat k,t) +  r^2(\phi, t) \mathcal{E}_0(\phi, t) \\
\lesssim  &  \mathcal{E}^2_{1}(\hat k,t) +  r^2(\phi, t) \mathcal{E}_0(\phi, t).
\end{split}
\end{equation}
When $l=1,$ \eqref{estimate-lapse-nabla-l-3-4-5-1} and the estimates for $\mathcal{E}_2(\hat N,t)$, \eqref{estimate-nabla-N-1} and \eqref{ineq-hat-N-2order-source}, show that
 \begin{equation}\label{N-inte-ineq-3order-source}
\begin{split}
& E_3(\hat N,t)  \lesssim   \mathcal{E}^2_{2}(\hat k,t) +  r^2(\phi, t) \mathcal{E}_0(\phi, t) \\
& \quad + \int_{\Sigma_t}  r^2(\phi, t) \left( |\nabla \Tbar \phi|^2 + |\nabla (\bar m (t) \phi)|^2 \right),
\end{split}
\end{equation}
which derives \eqref{lemma-N-inte-ineq-nabla3-source}.
When $l=2,$ we combine \eqref{estimate-lapse-nabla-l-3-4-5-1} and the estimates for $\hat N$ up to three derivatives \eqref{estimate-nabla-N-1}, \eqref{ineq-hat-N-2order-source}-\eqref{N-inte-ineq-3order-source} to deduce that
 \begin{equation}\label{N-inte-ineq-4order-source}
\begin{split}
E_4(\hat N,t)   \lesssim &   \mathcal{E}^2_{3}(\hat k,t) +  r^2(\phi, t) \mathcal{E}_1(\phi, t)+ \int_{\Sigma_t}  \big| \nabla \Tbar \phi \big|^4 + \big| \nabla ( \bar m (t) \phi) \big|^4 \\
& + \int_{\Sigma_t} r^2(\phi, t) \left( |\nabla^2 \Tbar \phi|^2 + |\nabla^2 (\bar m (t) \phi)|^2 \right).
\end{split}
\end{equation}
By the Sobolev inequality, we obtain \eqref{lemma-N-inte-ineq-nabla4-source}.
When $l=3,$ by \eqref{estimate-lapse-nabla-l-3-4-5-1} and the estimates for $\hat N$ up to four derivatives \eqref{estimate-nabla-N-1}, \eqref{ineq-hat-N-2order-source}-\eqref{N-inte-ineq-4order-source}, there is
 \begin{equation}\label{N-inte-ineq-5order-source}
\begin{split}
&E_5(\hat N,t)   \lesssim   \mathcal{E}^2_{4}(\hat k,t) +  r^2(\phi, t) \mathcal{E}_2(\phi, t)+\mathcal{E}^2_2(\phi, t) \\
& + \int_{\Sigma_t} \left( | \nabla \Tbar \phi * \nabla^2 \Tbar \phi | +  | \nabla ( \bar m (t) \phi) * \nabla^2 ( \bar m (t) \phi) | \right)^2 \\
& \quad + \int_{\Sigma_t} r^2(\phi, t) \left( |\nabla^3 \Tbar \phi|^2 + |\nabla^3 (\bar m (t) \phi)|^2 \right).
\end{split}
\end{equation}
By the Sobolev inequality, we prove \eqref{lemma-N-inte-ineq-nabla5-source}. Note that, we only need the refined estimate up to four order derivatives for $\hat N$ ($\mathcal{E}_4(\hat N,t)$) in the hierarchy of energy estimate for the KG field.
\end{proof}

\begin{remark}\label{rk-TT-phi}
In views of the definition of partial energies \eqref{def-energy-density-l-homo-kg-g}-\eqref{def-energy-l-homo-kg-g}, there is no energy norm involving two normal derivatives, such as $\|\nabla^i \Tbar ^2\phi\|_{L^2(\Sigma_t)}$. And hence we lose the control  for $\|\Tbar ^2\phi\|_{L^\infty}$ as well. However, by the KG equation in $1+3$ form \eqref{eq-rescale-kg-1+3-0}, there is
\begin{equation}\label{TT-KG-L-infty}
\begin{split}
\|\Tbar ^2\phi\|_{L^\infty} &\lesssim  \|\Tbar\phi\|_{L^\infty} + \|\bar m^2 (t) \phi\|_{L^\infty} + \|\nabla^2\phi\|_{L^\infty}   \\
&\quad \quad + \|N^{-1} \nabla \hat N\|_{L^\infty} \| \nabla \phi\|_{L^\infty}.
 \end{split}
\end{equation}
Thus, by the bootstrap assumption \eqref{BT-KG} and the $L^\infty$ estimate \eqref{BT-KG-L-infty-4-6} followed, we achieve 
\begin{equation}\label{TT-KG-L-2-infty-0}
\begin{split}
\|\Tbar^2\phi\|_{L^2(\Sigma_t)} \lesssim &\bar m (t) \mathcal{E}^{\frac{1}{2}}_{1} (\phi, t) + \mathcal{E}^{\frac{1}{2}}_{1} (\phi, t) \mathcal{E}^{\frac{1}{2}}_{2} (\hat N, t) \lesssim \varepsilon M t^{\frac{1}{2} + \delta},\\
\|\Tbar^2\phi\|_{L^\infty(\Sigma_t)} \lesssim & \bar m (t) \mathcal{E}^{\frac{1}{2}}_{3} (\phi, t) + \mathcal{E}^{\frac{1}{2}}_{2} (\phi, t) \mathcal{E}^{\frac{1}{2}}_{3} (\hat N, t) \lesssim \varepsilon M t^{\frac{1}{2} + \delta}
\end{split}
\end{equation}
This will be used to derive the preliminary estimates for $N^\prime$, see Lemma \ref{lem-TN-ineq-source} below.
\end{remark}

From now on, we let
\begin{equation}\label{def-energy-mertic-rescale}
\mathcal{E}_{i} (g, t) =  \mathcal{E}_{i} (W, t) + \mathcal{E}_{1+i} (\hat k, t) + \mathcal{E}_{2+i} (\hat N, t).
\end{equation}

The hierarchy of energy estimates for the KG filed, namely the idea of linearization, also relies deeply on an $L^\infty - L^\infty$ estimate for the KG field, which indeed retrieves the $\delta$ loss for $\|\phi\|_{L^\infty}, \|T\phi\|_{L^\infty}$. This will be explained in Section \ref{sec-L-infty-L-infty-kg}. To continue these procedures, we first need $\|N^\prime\|_{L^\infty}$.
Recall the definition
 \begin{equation}\label{def-N--prime}
N^\prime = \dtau N,
\end{equation}
and the elliptic equation for $N^\prime$
\begin{equation*}
\begin{split}
& \Delta N^\prime - 3 N^\prime
 = N^\prime \left(\bar R_{\Tbar \Tbar}(\phi) + |\hat k|^2 \right) + N \dtau \left(\bar R_{\Tbar \Tbar}(\phi) + |\hat k|^2 \right) \\
 &  - |\nabla \hat N |^2  -2 \hat k^{ij} \left( N\nabla_i\nabla_j \hat N+ \nabla_i \hat N \nabla_j \hat N \right) +2 \hat N \Delta \hat N + 2N \nabla^i \hat N \bar R_{\Tbar i} (\phi).
 \end{split}
\end{equation*}

\begin{lemma}\label{lem-TN-ineq-source}
Fixing $0<\delta<\frac{1}{6},$ and assuming \eqref{BT-KG}-\eqref{BT-k-N-L-infty}, we have for $N^\prime$,
 \begin{equation}\label{N-prime-inte-ineq-123-lemma}
 \int_{\Sigma_t} |\nabla N^\prime|^2 +  |N^\prime|^2 + |\nabla^2 N^\prime|^2
\lesssim \varepsilon^2 M^2 t^{4\delta}.
\end{equation}
And hence by Sobolev inequality,
\begin{equation}\label{nabla-N-L-infty}
\|N^\prime \|_{L^\infty} \lesssim  \varepsilon M t^{ 2\delta}.
\end{equation}
\end{lemma}
\begin{proof}
We multiply $N^\prime$ on both hand sides of the elliptic equation for $N^\prime$, and integrate by parts,
 \begin{equation}\label{N-prime-inte-1}
\begin{split}
& \int_{\Sigma_t} |\nabla N^\prime|^2 + 3 |N^\prime|^2=  \int_{\Sigma_t} \left(  |\nabla N |^2  - 2 \hat N \Delta N \right)  N^\prime  \\
&\quad \quad - \int_{\Sigma_t} |N^\prime|^2 \left( |\hat k|^2 + \bar R_{\Tbar \Tbar} \right) + N^\prime  N \dtau \left( |\hat k|^2 + \bar R_{\Tbar \Tbar} \right)\\
&\quad \quad + \int_{\Sigma_t} 2 N^\prime \hat k^{ij} \left( N \nabla_i\nabla_j \hat N+ \nabla_i \hat N \nabla_j \hat N \right) -2 N^\prime  N \nabla^i \hat N \bar R_{\Tbar i}.
\end{split}
\end{equation}
Then by Cauchy-Schwarz inequality, and the $L^\infty$ estimate for $\hat k$, $\D \phi$,
 \begin{equation*}
 \begin{split}
& \int_{\Sigma_t} |\nabla N^\prime|^2 +  |N^\prime|^2 \lesssim
\int_{\Sigma_t}   \left( \dtau |\hat k|^2 + \dtau \bar R_{\Tbar \Tbar} \right)^2  +  |\nabla \hat N |^2 |\bar R_{\Tbar i}|^2 \\
& \quad \quad \quad \quad + |\nabla \hat N|^4(1+|\hat k|^2) + |\nabla^2 \hat{N}|^2 (|\hat N|^2 + |\hat k|^2).
 \end{split}
\end{equation*}
Using the evolution equation for $\hat k$ (c.f. \eqref{eq-evolution-2}), and substituting the estimates for $\hat N$ (c.f. Lemma \ref{lem-N-ineq-source}) into the above formula, we arrive at  
 \begin{equation}\label{N-prime-inte-ineq-1}
 \begin{split}
 \int_{\Sigma_t} |\nabla N^\prime|^2 +  |N^\prime|^2 & \lesssim \mathcal{E}^2_1 (g,t) \left(1+ \mathcal{E}_1 (g,t) \right) + \mathcal{E}_1 (g,t)\mathcal{E}^2_1 (\phi,t) \\
& +   \int_{\Sigma_t} N |\Tbar \phi|^2 |\Tbar^2\phi|^2  +N \bar m^2 (t) |\bar m (t) \phi|^2 |\Tbar \phi|^2.
 \end{split}
\end{equation}
Referring to Remark \ref{rk-TT-phi} for $\|\Tbar^2 \phi \|_{L^2(\Sigma_t)}$, there is
 \begin{equation}\label{N-prime-inte-ineq-2}
  \int_{\Sigma_t} |\nabla N^\prime|^2 + |N^\prime|^2 \lesssim  t^2 \mathcal{E}^2_2 (\phi,t)  + \text{l.o.t.} \lesssim  \varepsilon^2 M^2 t^{4\delta}.
\end{equation}
$\text{l.o.t.}$ means lower order terms.

To proceed to the second order derivative, we make use of the B\"{o}chner identity, then
 \begin{equation*}
\begin{split}
& \int_{\Sigma_t} |\nabla^2 N^\prime|^2 =  \int_{\Sigma_t} | \Delta N^\prime|^2 -  R_{ij} \nabla^i N^\prime \nabla^j N^\prime.
\end{split}
\end{equation*}
Noting the expansion for $R_{imjn}$ \eqref{R-imjn-E}, $R_{ij}$ can be expanded in the same way.
Therefore, with the help of Cauchy-Schwarz inequality, we obtain,
 \begin{equation*}
\begin{split}
 \int_{\Sigma_t} |\nabla^2 N^\prime|^2 & \lesssim    \int_{\Sigma_t}    \left(\dtau \bar R_{\Tbar \Tbar} + \dtau |\hat k|^2 \right)^2+ |N^\prime|^2 \left( |\Tbar \phi|^4 + |\bar m (t) \phi|^4 + |\hat k|^4 \right) \\
 &+ |\nabla \hat N|^2 \left( |\nabla \hat N|^2 \left( 1+ |\hat k|^2 \right)  +|\Tbar \phi \nabla \phi|^2 \right) + |\nabla^2 N|^2 \left( |\hat N|^2 + |\hat k|^2\right) \\
 &+  |\nabla N^\prime|^2 \left(1+ |\hat k| + |\hat k|^2 + |\D \phi|^2 + |E| \right).
\end{split}
\end{equation*}
Substituting the estimates for $N$ (see Lemma \ref{lem-N-ineq-source}) and the known results \eqref{N-prime-inte-ineq-2}, we get
 \begin{equation*}
\begin{split}
 \int_{\Sigma_t} |\nabla^2 N^\prime|^2 \lesssim & t^2 \mathcal{E}^2_2 (\phi,t)  + \text{l.o.t.} \lesssim \varepsilon^2 M^2 t^{ 4\delta}.
\end{split}
\end{equation*}
Finally, the $L^\infty$ estimate \eqref{nabla-N-L-infty} is due to the Sobolev inequality.
\end{proof}

\section{Energy estimate for KG field}\label{sec-energy-kg}

\subsection{The $L^\infty - L^\infty$ estimate for KG field}\label{sec-L-infty-L-infty-kg}
We begin with the $L^\infty$ estimates below which follows from the weak bootstrap assumption for the KG field:
\begin{align*}
& \sum_{i=0}^{2} \|\nabla^i \D \phi\|_{L^\infty} \lesssim \varepsilon M  t^{- \frac{1}{2} + \delta}.
\end{align*}
Obviously, there is $\delta$ loss in the above estimate.
In this section, we will apply the technical ODE estimate (c.f. Lemma 3.5 in \cite{Ma-Lefoch-16}, of which a specific case for application is exhibited in Lemma \ref{lemma-ode} of appendix) to retrieve the $\delta$ loss.
\begin{corollary}[Improved $L^\infty$ estimates of KG field]\label{coro-Improved L-infty of Klein Gordon field}
With the bootstrap assumption \eqref{BT-KG}-\eqref{BT-k-N-L-infty}, we have the improved estimate for $\phi$,
\begin{equation}\label{improved-L-infty-kg-1}
|t^{\frac{1}{2}} \bar m (t) \phi| + |t^{\frac{1}{2}} \Tbar \phi| \lesssim \varepsilon I_4 (\phi) + \varepsilon M.
\end{equation}
\end{corollary}

\begin{proof}
Set
\begin{equation}\label{def-Phi}
\Phi = t^{\frac{3}{2}} \phi.
\end{equation}
Then a direct calculation with respect to the KG equation in $1+3$ form \eqref{eq-rescale-kg-1+3-0} leads to the second order ODE
\begin{equation}\label{ode-kg}
\begin{split}
T^2 \Phi + m \Phi =&    t^{-\frac{1}{2}} \left( \Delta \phi +N^{-1} \nabla_p N \nabla^p \phi  \right) - \frac{3}{2}  t^{-\frac{1}{2}} N^{-3} N^\prime t^{-1} \bar m (t) \phi  \\
&-3t^{-\frac{1}{2}} N^{-1}  \hat N \Tbar \phi  +\frac{3}{4} t^{-\frac{1}{2}} N^{-2} t^{-1} \bar m (t) \phi.
\end{split}
\end{equation}
We use the $L^\infty$ estimates for $\Delta \phi$,
\begin{align*}
&\|\Delta \phi\|_{L^\infty} \lesssim t^{-1} \|\Delta (\bar m (t) \phi) \|_{L^\infty} \lesssim t^{-1} \mathcal{E}^{\frac{1}{2}}_4(\phi, t).
\end{align*}
Then, the right hand side of \eqref{ode-kg} can be bounded by 
\begin{align*}
|T^2 \Phi + m \Phi | \lesssim& t^{-\frac{3}{2}} \left( \mathcal{E}^{\frac{1}{2}}_4(\phi, t) + \mathcal{E}^{\frac{1}{2}}_2(N^\prime, t) \mathcal{E}^{\frac{1}{2}}_2(\phi, t) \right) \\
&+ t^{-\frac{1}{2}}  \mathcal{E}^{\frac{1}{2}}_3(\hat N, t) \mathcal{E}^{\frac{1}{2}}_2(\phi, t) \lesssim \varepsilon M t^{-2 + 3\delta}.
\end{align*}
We integrate along the integral curves of $T$, which are the same as that of $\p_t$ but parametrized by arc
length $s$, i.e. $T=\frac{\p}{\p s}$. Along this curve $\di s = N\di t$. It follows from the technical ODE estimate \cite{Ma-Lefoch-16} (or Lemma \ref{lemma-ode}) that
\begin{equation}\label{ineq-ode-kg}
|\Phi(t)| + |T\Phi(t)| \lesssim \varepsilon I_4(\phi) + \int_{t_0}^{t} \varepsilon M t^{-2 + 3\delta} N \di t.
\end{equation}
This further leads to the improved $L^\infty$ estimates of $\phi$ \eqref{improved-L-infty-kg-1}.
\end{proof}

\subsection{The lower order terms}\label{sec-e-lot-kg} 
In this subsection, we will prove the following estimates for the lower order terms.
\begin{corollary}\label{coro-LK-1-LK_3}
With the bootstrap assumption  \eqref{BT-KG}-\eqref{BT-k-N-L-infty}, the lower order terms ${}_0LK$ \eqref{def-0-LK} and ${}_lLK_1-{}_lLK_3$ ($l \leq 4$) \eqref{energy-estimate-kg-1-id-l-o-t-1}-\eqref{energy-estimate-kg-1-id-l-o-t-3} in the energy identities for KG field 
 share the estimates
\begin{equation}\label{estimate-0LK}
 \int_{\Sigma_t} |{}_0LK|  \lesssim t^{-1} E_0 (\phi, t) +  \mathcal{E}^{\frac{1}{2}}_1 (g, t) E_0 (\phi, t),
\end{equation}
\begin{equation}\label{estimate-L-1-2-3}
\begin{split}
\int_{\Sigma_t}  \sum_{i=1}^3 |{}_lLK_i| \lesssim & t^{-1}E_l (\phi, t) + \mathcal{E}_4 (\phi, t) \left( \mathcal{E}^{\frac{1}{2}}_3 (g, t) +  \mathcal{E}_4 (\phi, t) \right).
\end{split}
\end{equation}
\end{corollary}

\subsubsection{The estimate for ${}_0LK$}\label{sec-e-lot-kg-LK0} 
The formula for ${}_0LK$ is given in \eqref{def-0-LK}.  We should note that $\phi =m^{-1} t^{-1} \bar m (t) \phi$, hence ${}_0LK$ can be rewritten as 
\begin{equation*}
\begin{split}
  {}_0LK = & \left( t^{-1} N + t^{-1} \hat N +  N \hat k + \nabla \hat N  \right) |\D \phi|^2.
\end{split}
\end{equation*}
Therefore, by $\|\hat N\|_{L^\infty}$ \eqref{BT-k-N-L-infty} and $\|\hat k\|_{L^\infty}$ \eqref{BT-k-L-infty-4-6}, we derive \eqref{estimate-0LK}.

\subsubsection{The estimate for ${}_lLK_1$}\label{sec-e-lot-kg-L1} 
Since $$\nabla^{I_l} \phi = m^{-1}t^{-1} \nabla^{I_l} (\bar m (t) \phi),$$ the first line in ${}_lLK_1$ \eqref{energy-estimate-kg-1-id-l-o-t-1} can be rewritten as 
\begin{equation*}
\begin{split}
  {}_lLK^1_1 =& \left( \hat N + N\hat k + t^{-1} N + t^{-1} \hat N + t^{-1} N \hat k  \right) | \nabla_{I_l} \D \phi|^2,
\end{split}
\end{equation*}
while second line in  ${}_lLK_1$ \eqref{energy-estimate-kg-1-id-l-o-t-1} is
\begin{equation*}
\begin{split}
  {}_lLK^2_1 = & -  N \nabla^{I_l} \Tbar \phi  \mathcal{R}_{I_l} (\phi)  - t^{-1} N \nabla^{I_l} (\bar m (t) \phi)   \mathcal{R}_{I_l}(\phi).
\end{split}
\end{equation*}

Making use of $\|\hat N\|_{L^\infty}$ \eqref{BT-k-N-L-infty} and $\|\hat k\|_{L^\infty}$ \eqref{BT-k-L-infty-4-6}, we have 
\begin{equation}\label{estimate-LK-11}
 \int_{\Sigma_t} |{}_lLK^1_1|  \lesssim t^{-1} E_l (\phi, t) +  \mathcal{E}^{\frac{1}{2}}_1 (g, t) E_l (\phi, t).  
\end{equation}

For ${}_lLK^2_1$, comparing to the first term, the second term in ${}_lLK^2_1$ taking the form of $t^{-1} N \nabla^{I_l} \D \phi  \mathcal{R}_{I_l} (\phi)$ is of lower order.  We now consider the general form $N \nabla^{I_l} \D \phi  \mathcal{R}_{I_l} (\phi)$, which can be further rewritten as
\begin{equation}\label{def-L-23}
\begin{split}
t^{-1} \nabla^{I_l} \D \phi \sum_{a +b= l, \,\, a \leq l-1} N \nabla_{I_a} R_{imjn} *  \nabla_{I_b} (\bar m (t) \phi).
\end{split}
\end{equation}

\begin{itemize}
\item {\bf Case ${}_lLK^2_{1}$-I}: $a=0$. Then $b=l.$ In views of the expansion for $R_{imjn}$ \eqref{R-imjn-E}, we have
\begin{equation}\label{estimate-L231-1}
\begin{split}
&\int_{\Sigma_t} t^{-1}\big| \nabla^{I_l} \D \phi * N R_{imjn} * \nabla_{I_l} \D \phi \big| \\
 \lesssim & t^{-1} E_l(\phi, t)  +t^{-1} E_l(\phi, t) \left( \mathcal{E}^{\frac{1}{2}}_2 (g, t) +  \mathcal{E}_2 (\phi, t). \right)
\end{split}
\end{equation}
\item {\bf Case ${}_lLK^2_{1}$-II:} $a \geq 1$. Then $1 \leq a, b \leq l-1 \leq 3$.
\begin{enumerate}
\item If $a \leq 2$, we can apply $L^4, L^4$ on $\nabla_{I_a}  R_{imjn} * \nabla_{I_b}  \D \phi$. To estimate $\| \nabla_{I_a} R_{imjn}\|_{L^4(\Sigma_t)}$,
we note that $\nabla_{I_a}  R_{imjn} = \nabla_{I_a}  ( \hat k * \hat k \pm  \hat k \pm E \pm \mathcal{D} \phi * \mathcal{D} \phi ).$  Furthermore, to estimate 
$ \|\nabla_{I_a} \left(\D \phi * \D \phi \right)\|_{L^4(\Sigma_t)},$ or $ \|\nabla_{I_a} \left(\hat k* \hat k \right)\|_{L^4(\Sigma_t)}$, noticing that $[\frac{a}{2}] + 2 \leq [\frac{l-1}{2}] + 2 \leq 3 <4$ for $(l \leq 4)$, we can apply $L^4, L^\infty$ or $ L^\infty, L^4$ to the two factors. Finally, we derive for $1 \leq a \leq 2$,
\begin{equation*}
\begin{split}
\| \nabla_{I_a} R_{imjn} \|_{L^4(\Sigma_t)}  \|  \nabla_{I_b} \D \phi \|_{L^4(\Sigma_t)} \lesssim & \left( \mathcal{E}^{\frac{1}{2}}_3 (g, t) +  \mathcal{E}_4 (\phi, t) \right) \mathcal{E}^{\frac{1}{2}}_l (\phi, t).
\end{split}
\end{equation*}

\item If $a=3$, which occurs only when $l=4,$ and therefore $b=1$. In this case, we apply $L^2, L^\infty$ on $\nabla^3  R_{imjn} * \nabla \D \phi$, and
\begin{equation*}
\begin{split}
\| \nabla^3 R_{imjn} \|_{L^2(\Sigma_t)} \|\nabla \D \phi\|_{L^\infty} \lesssim &  \left( \mathcal{E}^{\frac{1}{2}}_3 (g, t) +  \mathcal{E}_3 (\phi, t) \right) \mathcal{E}^{\frac{1}{2}}_3 (\phi, t).
\end{split}
\end{equation*}
\end{enumerate}
Thus, it follows that for ${}_lLK^2_{1}$-II 
\begin{equation}\label{estimate-L231-2}
\begin{split}
&\int_{\Sigma_t} t^{-1} \big| \nabla^{I_l} \D \phi * \nabla_{I_a}  R_{imjn} * \nabla^{I_b} \D \phi \big| \\
\lesssim & t^{-1} \mathcal{E}_4 (\phi, t) \left( \mathcal{E}^{\frac{1}{2}}_3 (g, t) +  \mathcal{E}_3 (\phi, t) \right).
\end{split}
\end{equation}
\end{itemize}
We deduced for ${}_lLK^2_{1}$
\begin{equation*}
\begin{split}
&\int_{\Sigma_t} | {}_lLK^2_{1} |
\lesssim t^{-1}  E_l(\phi, t)  +  t^{-1} \mathcal{E}_4 (\phi, t) \left( \mathcal{E}^{\frac{1}{2}}_3 (g, t) +  \mathcal{E}_4 (\phi, t) \right).
\end{split}
\end{equation*}
Combine with \eqref{estimate-LK-11},
\begin{equation}\label{estimate-LK1-2}
\begin{split}
&\int_{\Sigma_t} | {}_lLK_{1} |
\lesssim t^{-1}  E_l(\phi, t)  +  t^{-1} \mathcal{E}_4 (\phi, t) \left( \mathcal{E}^{\frac{1}{2}}_3 (g, t) +  \mathcal{E}_4 (\phi, t) \right).
\end{split}
\end{equation}

\subsubsection{The estimate for ${}_lLK_2$}\label{sec-e-lot-kg-L2}
${}_lLK_2$ \eqref{energy-estimate-kg-1-id-l-o-t-2} can be further written as ${}_lLK_2 = {}_lLK_{21} + {}_lLK_{22}$, where (we have always ignored irrelevant constants) ${}_lLK_{2i}, \, i=1,2$ are given as follows: ${}_lLK_{21} = {}_lLK^1_{21} + {}_lLK^2_{21}$
\begin{equation}\label{def-LK-21}
\begin{split}
{}_lLK^1_{21} := &  \nabla^{I_l} \D \phi \sum_{a +b+1 = l+1} \nabla_{I_{a}} \nabla \hat N  * \nabla_{I_{b}} \D \phi, \\
 {}_lLK^2_{21}:=& \nabla^{I_l} \D \phi \sum_{a +1+b = l}  \nabla_{I_{a}} \nabla \hat N  * \nabla_{I_b} \D \phi,
 \end{split}
\end{equation}
 which comes from $\nabla^{I_l} \D \phi \mathcal{N}_{l+1}(\D \phi),$ and $\nabla^{I_l} \D \phi \mathcal{N}_{l}(\D \phi).$ And
 \begin{equation}\label{def-LK-22}
\begin{split}
{}_lLK_{22} :=&  \nabla^{I_{l}} \D \phi \sum_{a +b + 1 = l } \nabla_{I_{a+1}} \left( N \hat k + \hat N \right) * \nabla_{I_{b}} \D \phi,
\end{split}
\end{equation}
which is $\nabla^{I_l} \D \phi  \mathcal{K}\mathcal{N}_{I_l}(\D \phi)$.
We remark that $1 \leq l \leq 4$.
 
For ${}_lLK^1_{21}, {}_lLK^2_{21}$, $a+ b \leq l,$ and $[\frac{l}{2}] + 2 \leq 4 (l \leq 4),$  we can always apply $L^2, L^\infty$  or $L^\infty, L^2$ to the two factors in $ \nabla_{I_{a}} \nabla \hat N * \nabla_{I_b} \D \phi$ to derive 
\begin{equation}\label{estimate-LK211}
\begin{split}
\int_{\Sigma_t} |{}_lLK_{21}| 
 \lesssim &E^\frac{1}{2}_l(\phi, t) \mathcal{E}^{\frac{1}{2}}_5(\hat N, t)  \mathcal{E}^{\frac{1}{2}}_4(\phi, t),  \lesssim  \varepsilon^3 M^3  t^{-2+4 \delta}.
\end{split}
\end{equation}

${}_lLK_{22}$ can be further split into the following two terms:
\begin{align*}
 {}_lLK^{1}_{22} &= \nabla^{I_{l}} \mathcal{D} \phi \sum_{p+1+q +b =l} \nabla_{I_{p+1}} \hat N *\nabla_{I_{q}} \hat k * \nabla_{I_{b}} \mathcal{D} \phi, \\
{}_lLK^{2}_{22} &= \nabla^{I_{l}} \mathcal{D} \phi \sum_{p+b=l, p\geq 1} \left( N \nabla_{I_{p}} \hat k + \nabla_{I_{p}} \hat N \right) * \nabla_{I_{b}} \mathcal{D}\phi.
\end{align*} 
For ${}_lLK^{1}_{22}$, since $ b,q \leq l-1$, and $p+1 \leq l$ ($l \leq 4$). Then we can apply $L^2, L^6, L^6, L^6$ to the four factors, so that
\begin{equation}\label{estimate-LK221}
\begin{split}
\int_{\Sigma_t} |{}_lLK^{1}_{22}| 
 \lesssim &\mathcal{E}_l(\phi, t) \mathcal{E}_{l-1}(g, t) \lesssim  \varepsilon^4 M^4 t^{-3+5 \delta}.
\end{split}
\end{equation}
As ${}_lLK^{2}_{22}$, $[\frac{l}{2}] + 2 \leq 4 (l \leq 4),$ we can always apply $L^2, L^2, L^\infty$  or $L^2, L^\infty, L^2$ to the three factors in  ${}_lLK_{22}^{2}$
\begin{equation}\label{estimate-LK222}
\begin{split}
\int_{\Sigma_t} |{}_lLK^{2}_{22}| 
 \lesssim & E^\frac{1}{2}_l(\phi, t) \mathcal{E}^{\frac{1}{2}}_3(g, t)  \mathcal{E}^{\frac{1}{2}}_4(\phi, t) \lesssim  \varepsilon^3 M^3 t^{-2+4 \delta}.
\end{split}
\end{equation}
 In summary, we have
\begin{equation}\label{estimate-LK2-2}
\begin{split}
& \int_{\Sigma_t}  |{}_lLK_{22}| \lesssim   \mathcal{E}_4(\phi, t)\mathcal{E}^{\frac{1}{2}}_3(g, t)
\end{split}
\end{equation}

Combining \eqref{estimate-LK211} with \eqref{estimate-LK2-2}, we conclude that for ${}_lLK_{2}$
\begin{equation}\label{estimate-LK2}
\begin{split}
&\int_{\Sigma_t}|  {}_lLK_{2}|
\lesssim \mathcal{E}_4 (\phi, t) \left( \mathcal{E}^{\frac{1}{2}}_3 (g, t) +  \mathcal{E}_4 (\phi, t) \right).
\end{split}
\end{equation}

\subsubsection{The estimate for ${}_lLK_3$}\label{sec-e-lot-kg-L3}
Generally speaking, comparing to ${}_lLK_2$, ${}_lLK_3$  \eqref{energy-estimate-kg-1-id-l-o-t-3} has additionally factor $t^{-1}$, for $\nabla_{I_l} \phi \sim t^{-1} \nabla_{I_l} \left( \bar m(t) \phi \right)$.  To be more specific, the following terms in ${}_lLK_3$ admit the quantitatively equivalence \footnote{Here $\psi_1$ and $\psi_2$ are quantitatively equivalent, denoted by $\psi_1 \dot\sim \psi_2$, in the sense that $\psi_1$ and $\psi_2$ share the same structures and they can be estimate in the same way.}:
\begin{subequations}
\begin{equation}
 \nabla^{I_l} \Tbar \phi  \mathcal{K}\mathcal{N}_{I_l}(\phi) +   \nabla^{I_l} \phi \mathcal{K}\mathcal{N}^\prime_{I_l}(\Tbar \phi) \dot\sim t^{-1} \nabla^{I_l} \D \phi  \mathcal{K}\mathcal{N}^\prime_{I_l}( \D\phi) = t^{-1} {}_lLK_{22},\label{eq-LK31-1}
\end{equation}
\begin{equation}
  \nabla^{I_l} \phi  \mathcal{N}_{I_l}(\Tbar \phi) \dot\sim t^{-1} \nabla^{I_l} \D \phi \mathcal{N}_l(\D \phi)  = t^{-1} {}_lLK^2_{21},\label{eq-LK31-2}
\end{equation}
\begin{equation}
 \nabla^{I_l} \phi \mathcal{N}_{I_{l+1}}(\nabla\phi) \dot\sim t^{-1} \nabla^{I_l} \D \phi \mathcal{N}_{I_{l+1}}(\D \phi)  = t^{-1} {}_lLK^1_{21}. \label{eq-LK31-6}
\end{equation}
\end{subequations} 
Due to the estimates for ${}_lLK_{2}$ \eqref{estimate-LK2}, ${}_lLK_{3}$ can be bounded as 
\begin{equation}\label{estimate-LK3}
\begin{split}
&\int_{\Sigma_t}  | {}_lLK_{3}|
\lesssim t^{-1} \mathcal{E}_4 (\phi, t) \left( \mathcal{E}^{\frac{1}{2}}_3 (g, t) +  \mathcal{E}_4 (\phi, t) \right).
\end{split}
\end{equation}

\subsection{The borderline terms and hierarchy of energy}\label{sec-0-ee-kg} 
In this section, we focus on the borderline terms $BL_l$ \eqref{def-energy-estimate-kg-id-bl}. We will take advantage of the relation between lapse and KG field (c.f. Lemma \ref{lem-N-ineq-source}) and the enhanced $L^\infty$ estimates \eqref{improved-L-infty-kg-1} to conduct an hierarchy of energy estimate. Then each order of KG energy and lapse will be refined step by step, base on which the nonlinear terms $BL_l$ are reduced  to linear ones. The idea of linearization and hierarchy of energy work together to close the energy argument for KG field.  

\subsubsection{The zeroth order energy estimate}\label{sec-0-ee-kg-1} 
\begin{lemma}\label{lemma-improved-0-energy-KG}
We have the improved estimate for the zeroth order energy 
\begin{equation}\label{eq-improved-energy-0-order}
t E_0(\phi,t) \lesssim \varepsilon^2I^2_0(\phi).
\end{equation}
\end{lemma}
\begin{proof}
According to the zeroth order energy identity \eqref{energy-id-0-kg} (see Theorem \ref{Thm-energy-id-KG-0}) and the estimate for ${}_0LK$ \eqref{estimate-0LK}, we obtain
\begin{equation*}
\dtau \tilde E_0(\phi, t) + \tilde E_0(\phi, t)  \lesssim t^{-1} E_0(\phi, t) +  E_0 (\phi, t) \mathcal{E}^{\frac{1}{2}}_1 (g, t).
\end{equation*}
We then change to the $t$ time function to derive that
\begin{equation}\label{ineq-energy-diff-0-order}
\dt ( t\tilde E_0(\phi, t)) \lesssim t^{-2} (t E_0(\phi, t)) + t^{-1} \mathcal{E}^{\frac{1}{2}}_1 (g, t) (tE_0 (\phi, t)).
\end{equation}
We recall the definition for $\mathcal{E}_1 (g, t)$.
Substitute the bound for $\mathcal{E}_1 (g, t)$ \eqref{BT-E-H}-\eqref{BT-k-N-L-infty}, \eqref{N-inte-ineq-123-lemma} and integrate along $t$ interval, noting that $ t\tilde E_0(\phi, t) \sim t E_0(\phi, t)$
\begin{equation}\label{ineq-energy-0-order}
t^\prime E_0(\phi, t^\prime) \lesssim \varepsilon^2I^2_0(\phi) + \int^{t^\prime}_{t_0} \left( \frac{1}{t^2} + \frac{\varepsilon M}{t^{2-2\delta}} \right) (t E_0(\phi,t)) \di t, \quad 0 < \delta < \frac{1}{6}.
\end{equation}
An application of the Gr\"{o}nwall inequality yields \eqref{eq-improved-energy-0-order}.
\end{proof}

With the result of Lemma \ref{lemma-improved-0-energy-KG}, we can further sharpen the estimates for lapse up to two derivatives.
\begin{corollary}\label{lem-N-ineq-source-improved-2}
The estimate for $\mathcal{E}_2(\hat N,t)$ is improved as
 \begin{equation}\label{N-improved-estimate-source}
t^2 \mathcal{E}_2(\hat N,t) \lesssim  \varepsilon^2  I^2_4(\phi).
\end{equation}
\end{corollary}

\begin{proof}
By virtue of \eqref{lemma-N-inte-ineq-2order-source} in Lemma \ref{lem-N-ineq-source}, we substitute the enhanced $|\bar m (t) \phi |, |\Tbar \phi|$ \eqref{improved-L-infty-kg-1}, $E_0(\phi,t)$ \eqref{eq-improved-energy-0-order} and the bootstrap assumption for $\mathcal{E}_1 (\hat k, t)$ into \eqref{lemma-N-inte-ineq-2order-source}, 
 \begin{align*}
  t^2 \mathcal{E}_2(\hat N,t) &\lesssim  \varepsilon^2 I^2_0(\phi) \left( \varepsilon^2 M^2 + \varepsilon^2 I^2_4(\phi) \right) + \varepsilon^4 M^4 t^{-2 + 4\delta}\lesssim \varepsilon^2  I^2_4(\phi).
\end{align*}
\end{proof}

\subsubsection{The high order energy estimates}\label{sec-2-ee-kg}
Now we make use of the previously improved $\phi, \hat N$ to further refine the estimate for $E_l(\phi, t)$ step by step. 
\begin{lemma}\label{lemma-improved-l-energy-KG}
There is some constant $C_M$ depending linearly on $M$, such that
\begin{equation}\label{improved-energy-l-order-1}
t E_l(\phi, t) \lesssim \left( \varepsilon^2 I^2_4(\phi) + \varepsilon^3 M^3 \right) t^{C_M \varepsilon}, \, 1 \leq l \leq 4.
\end{equation}
Meanwhile, the lapse admits the enhanced estimates
 \begin{equation}\label{N-improved-estimate-source-to-4}
t^2 \mathcal{E}_{i+2}(\hat N, t) \lesssim \left( \varepsilon^2 I^2_{i}(\phi) + \varepsilon^3 M^3 \right) t^{C_M \varepsilon}, 1 \leq i \leq 2.
\end{equation}
\end{lemma}

\begin{proof}
After substituting the estimates for the lower order terms ${}_lLK_1-{}_lLK_3$ (Corollary \ref{coro-LK-1-LK_3}) into the energy identity for $\phi$ \eqref{energy-id-l-kg-rescale} (Corollary \ref{coro-energy-id-KG-l}), we arrive at for $ l \geq 1$ (noting that $\tilde E_l(\phi, t) \sim E_l(\phi, t)$)
\begin{equation}\label{ineq-energy-diff-l-order-1}
\begin{split}
\dtau \tilde E_l(\phi, t) + \tilde E_l(\phi, t)  \lesssim & t^{-1} E_l (\phi, t) + \int_{\Sigma_t} | BL_l |\\
&+ \mathcal{E}_4 (\phi, t) \left( \mathcal{E}^{\frac{1}{2}}_3 (g, t) +  \mathcal{E}_4 (\phi, t) \right),
\end{split}
\end{equation}
where the borderline terms $BL_l$ \eqref{def-energy-estimate-kg-id-bl}  take the form:
$BL_l=BL_l^1 + \cdots+BL_l^l$, and $BL_l^i, \, i=1,\cdots, l$ read
\begin{equation}\label{def-BLl-1234}
\begin{split}
BL_l^i =&\bar m (t) \nabla^{l-i+1} N *\left( \nabla^l (\bar m (t) \phi) * \nabla^{i-1} \bar T \phi \pm \nabla^l \Tbar \phi * \nabla^{i-1}( \bar m (t) \phi )\right).
\end{split}
\end{equation}
The special structure of $BL_l^i$  allows us to refine the estimates step by step.

{\bf Step I: the improvement for $E_1(\phi, t)$}. 
 With the help of Corollary \ref{lem-N-ineq-source-improved-2} and the enhanced $\|\bar m (t) \phi\|_{L^\infty}, \|\Tbar \phi\|_{L^\infty}$ \eqref{improved-L-infty-kg-1}, the borderline term $BL_1= BL_1^1$, which reads $$\bar m (t) \nabla N * \nabla (\bar m (t) \phi) * \Tbar \phi - \bar m (t) \nabla N * \nabla \Tbar \phi * (\bar m (t) \phi)$$ is in fact a linear term.
 
Applying $L^2, L^2, L^\infty$ to the three factors in $BL_1,$ we have, by Cauchy-Schwarz inequality,
\begin{equation}\label{ineq-energy-1-order-bl}
\begin{split}
 \int_{\Sigma_t} | BL_1|  \lesssim & t \|\bar T\phi\|_{L^\infty} \| \nabla N \|_{L^2(\Sigma_t)}  \|\nabla (\bar m (t) \phi) \|_{L^2(\Sigma_t)}   \\
 &\quad + t \|\bar m (t) \phi\|_{L^\infty} \| \nabla N \|_{L^2(\Sigma_t)}  \| \nabla \Tbar \phi \|_{L^2(\Sigma_t)} \\
  \lesssim & t^{\frac{1}{2}} \varepsilon M \mathcal{E}^{\frac{1}{2}}_1 (\hat N, t)  E_1^{\frac{1}{2}}(\phi,t)  
  \lesssim t^{-\frac{1}{2}} \varepsilon^2 M I_4(\phi)  E_1^{\frac{1}{2}}(\phi,t).
  \end{split}
\end{equation}
Hence adding \eqref{ineq-energy-1-order-bl} to \eqref{ineq-energy-diff-l-order-1}  with $l=1$,
\begin{equation*}
\begin{split}
\dtau \tilde E_l(\phi, t) + \tilde E_l(\phi, t)  \lesssim & \left( t^{-1} + \varepsilon^2 M^2 \right) E_l (\phi, t) +  \varepsilon^2 I_4^2 (\phi) t^{-1} \\
&+ \mathcal{E}_4 (\phi, t) \left( \mathcal{E}^{\frac{1}{2}}_3 (g, t) +  \mathcal{E}_4 (\phi, t) \right).
\end{split}
\end{equation*}
We then change to $t$ function,
\begin{equation}\label{ineq-energy-diff-1-order-1}
\begin{split}
\dt (t \tilde E_l(\phi, t)) \lesssim & \left( t^{-2} + \varepsilon M t^{-1} \right) (tE_l (\phi, t) ) +  \varepsilon^2 I_4^2 (\phi) t^{-1} \\
&+ \varepsilon^3 M^3 t^{-2+3\delta},
\end{split}
\end{equation}
Integrate along $\dt$, $\tilde E_l(\phi, t) \sim E_l(\phi, t)$,
\begin{equation}\label{ineq-energy-1-order-2}
\begin{split}
t^\prime E_1(\phi, t') \lesssim&  \varepsilon^2 I^2_4(\phi) t'^{\varepsilon M} + \varepsilon^3 M^3.
  \end{split}
\end{equation}
By Gr\"{o}nwall inequality, there is some constant $C_M$ depending linearly on $M$, so that \eqref{improved-energy-l-order-1} holds with $l=1$.

{\bf Step II: the improvement for $E_2(\phi, t)$}. The borderline term $BL_2=BL_2^1 + BL_2^2$, where $BL_2^1, \, BL_2^2$ read as below
\begin{equation*}
\begin{split}
BL_2^1 =& \bar m (t) \nabla^2 (\bar m (t) \phi) * \nabla^2 N * \Tbar \phi + \bar m (t) \nabla^2 \Tbar \phi * \nabla^2 N * \bar m (t) \phi, \\
BL_2^2 =& \bar m (t) \nabla^2 (\bar m (t) \phi) * \nabla N *  \nabla \Tbar \phi + \bar m (t) \nabla^2 \Tbar \phi * \nabla N *  \nabla (\bar m(t) \phi).
\end{split}
\end{equation*}
Now, in additional to the enhanced $E_0(\phi, t)$ \eqref {ineq-energy-0-order}, $\mathcal{E}_2(\hat N, t)$ \eqref{N-improved-estimate-source}, and $\|\bar m (t) \phi\|_{L^\infty}$, $\|\Tbar \phi\|_{L^\infty}$ \eqref{improved-L-infty-kg-1}, we should also utilize the newly improved $E_1(\phi, t)$.

For $BL_2^1$,  we follow the estimate for of $BL_1$ \eqref{ineq-energy-1-order-bl},
\begin{equation}\label{ineq-energy-2-order-bl-1}
\begin{split}
\int_{\Sigma_t} | BL_2^1| 
  \lesssim & t^{-\frac{1}{2}} \varepsilon^2 M I_4(\phi)  E_2^{\frac{1}{2}}(\phi,t) \lesssim   \varepsilon^2 M^2 E_2 (\phi, t) +  \varepsilon^2 I_4^2 (\phi) t^{-1}.
  \end{split}
\end{equation}

For $BL_2^2$, we apply $L^2, L^4, L^4$ to the three factors. By Sobolev inequalities, $\|\nabla \D \phi \|_{L^4(\Sigma_t)} \lesssim \mathcal{E}_2^\frac{1}{2} (\phi, t)$. Note that $\mathcal{E}^\frac{1}{2}_2 (\phi, t) \lesssim \mathcal{E}^\frac{1}{2}_1(\phi, t) + E^\frac{1}{2}_2(\phi, t)$, and $\mathcal{E}^\frac{1}{2}_1(\phi, t)$ has been improved via Step I. Then, $BL_2^2$ is indeed a linear term,
\begin{equation}\label{ineq-energy-2-order-bl-2}
\begin{split}
\int_{\Sigma_t}  | BL_2^2|  \lesssim & t  \| \nabla N \|_{L^4(\Sigma_t)}  \|\nabla \Tbar  \phi \|_{L^4(\Sigma_t)} \|\nabla^2 \left( \bar m (t) \phi \right) \|_{L^2(\Sigma_t)}  \\
&+ t  \| \nabla N \|_{L^4(\Sigma_t)}  \| \nabla \left( \bar m (t)  \phi \right) \|_{L^4(\Sigma_t)} \| \nabla^2 \Tbar \phi\|_{L^2(\Sigma_t)} \\
 &\lesssim  t  \mathcal{E}^{\frac{1}{2}}_2 (\hat N, t) \left( \mathcal{E}^\frac{1}{2}_1(\phi, t) + E^\frac{1}{2}_2(\phi, t) \right)  E_2^{\frac{1}{2}}(\phi,t)  \\
 & \lesssim  \varepsilon I_4(\phi) \left(  \left( \varepsilon I_4(\phi) +\varepsilon^\frac{3}{2} M^\frac{3}{2} \right)  t^{-\frac{1}{2} + C_M \varepsilon} + E^\frac{1}{2}_2(\phi, t)\right) E^\frac{1}{2}_2 (\phi,t)  \\
 & \lesssim \varepsilon^2 I_4^2(\phi) t^{-1+C_M \varepsilon} +   \left( \varepsilon^3 M^3  +  \varepsilon I_4(\phi)  \right) E_2(\phi,t).
\end{split}
\end{equation}

With \eqref{ineq-energy-2-order-bl-1} and \eqref{ineq-energy-2-order-bl-2}, the energy inequality \eqref{ineq-energy-diff-l-order-1} with $l=2$ becomes
\begin{equation}\label{ineq-energy-2-order-2}
\begin{split}
\dtau E_2(\phi, t) + E_2(\phi, t) \lesssim &  \left( t^{-1} + \varepsilon  M \right)   E_2(\phi,t)  \\
&+\varepsilon^2 I_4^2(\phi) t^{-1+C_M \varepsilon} + \varepsilon^3 M^3  t^{-2+3\delta}.
\end{split}
\end{equation}
Change to $t$,
\begin{equation*}
\begin{split}
\dt (t E_2(\phi, t))  \lesssim &  \left( t^{-2} + \varepsilon  Mt^{-1} \right)  (t E_2(\phi,t) ) \\
&+\varepsilon^2 I_4^2(\phi) t^{-1+C_M \varepsilon} + \varepsilon^3 M^3  t^{-2+3\delta}.
\end{split}
\end{equation*}
And the Gr\"{o}nwall inequality yields \eqref{improved-energy-l-order-1} with $l=2.$

{\bf Setp III: the improvement for $\hat N$}. In general, the refinement for $\mathcal{E}_{l}(\phi, t)$ leads to improvement for $\mathcal{E}_{l+2}(\hat N, t).$
We substitute the previously improved $\mathcal{E}_2(\phi, t)$ and $\|\bar m (t) \phi\|_{L^\infty}$, $\|\Tbar \phi\|_{L^\infty}$ \eqref{improved-L-infty-kg-1} into the corresponding inequality for $\hat N$ \eqref{lemma-N-inte-ineq-nabla3-source}-\eqref{lemma-N-inte-ineq-nabla5-source} ($l=1$ corresponds to \eqref{lemma-N-inte-ineq-nabla3-source}, $l=2$ corresponds to \eqref{lemma-N-inte-ineq-nabla4-source}). Then,
 \begin{equation*}
t^2 E_{l+2}(\hat N, t) \lesssim  \varepsilon^3 M^3 t^{-2 + 4\delta} + \varepsilon^2 M^2 \left( \varepsilon^2 I^2_4(\phi) + \varepsilon^3 M^3 \right) t^{C_M \varepsilon}, \, l \leq 2,
\end{equation*}
where $0<\delta < \frac{1}{6}.$ Hence \eqref{N-improved-estimate-source-to-4} follows.

{\bf Step IV: the improvement for $E_3(\phi, t)$}. We make use of the refined $E_l(\phi, t), l \leq 2$, $\mathcal{E}_3(\hat N, t)$ \eqref{N-improved-estimate-source}, \eqref{N-improved-estimate-source-to-4}, and $\|\bar m (t) \phi\|_{L^\infty}$, $\|\Tbar \phi\|_{L^\infty}$ \eqref{improved-L-infty-kg-1} to linearize the borderline term $BL_3$. Note that $BL_3=BL_3^1 + BL_3^2+BL_3^3$, where $BL_3^i,\, i=1,2,3$ are shown in \eqref{def-BLl-1234}.

For $BL_3^1$, we follow the argument for $BL_1$ \eqref{ineq-energy-1-order-bl} and $BL_2^1$ \eqref{ineq-energy-2-order-bl-1} to derive,
\begin{equation}\label{ineq-energy-3-order-bl-1}
\begin{split}
\int_{\Sigma_t} |BL_3^1| 
  \lesssim &  t^{-\frac{1}{2}} \varepsilon M \left( \varepsilon I_4(\phi) +\varepsilon^\frac{3}{2} M^\frac{3}{2}  \right)  t^{C_M \varepsilon} E_3^{\frac{1}{2}}(\phi,t)\\ \lesssim  &  \varepsilon^2 M^2 E_3 (\phi, t) + \left( \varepsilon^2 I^2_4(\phi) +\varepsilon^3 M^3 \right) t^{-1+C_M \varepsilon}.
  \end{split}
\end{equation}
For $BL_3^2$, we apply $L^2, L^4, L^4$ on the three factors, then
\begin{equation}\label{ineq-energy-3-order-bl-2}
\begin{split}
 \int_{\Sigma_t}  | BL_3^2| &\lesssim  t  \| \nabla^2 N \|_{L^4(\Sigma_t)}  \| \nabla \Tbar \phi \|_{L^4(\Sigma_t)} \| \nabla^3 \left( \bar m (t) \phi \right)) \|_{L^2(\Sigma_t)}  \\
&+ t \| \nabla^2 N \|_{L^4(\Sigma_t)}  \|\nabla \left( \bar m (t)  \phi \right) \|_{L^4(\Sigma_t)} \|\nabla^3 \Tbar \phi\|_{L^2(\Sigma_t)}  \\
 &\lesssim  t  \mathcal{E}^{\frac{1}{2}}_3 (\hat N, t)  \mathcal{E}^\frac{1}{2}_2(\phi, t)   E_3^{\frac{1}{2}}(\phi,t)  \\
 & \lesssim  t^{-\frac{1}{2} }  \varepsilon M  \left( \varepsilon I_4(\phi) +\varepsilon^\frac{3}{2} M^\frac{3}{2} \right)^2  t^{C_M \varepsilon}   E^\frac{1}{2}_3 (\phi,t)  \\
 & \lesssim \left( \varepsilon^2 I^2_4(\phi) +\varepsilon^3 M^3 \right)  t^{-1+C_M \varepsilon} +   \left( \varepsilon^3 M^3  +  \varepsilon^2 I^2_4(\phi)  \right) E_3(\phi,t).
\end{split}
\end{equation}
For $BL_3^3$, an analogous argument for $BL_2^2$ \eqref{ineq-energy-2-order-bl-2} can be adapted here, noting that $\mathcal{E}^\frac{1}{2}_3 (\phi, t) \lesssim \mathcal{E}^\frac{1}{2}_2(\phi, t) + E^\frac{1}{2}_3(\phi, t)$, and it leads to
\begin{equation}\label{ineq-energy-3-order-bl-3}
\begin{split}
 \int_{\Sigma_t}   |BL_3^3| 
\lesssim &  \varepsilon^2 I_4^2(\phi) t^{-1+C_M \varepsilon} +   \left( \varepsilon^3 M^3  +  \varepsilon^2 I^2_4(\phi)  \right) E_3(\phi,t).
\end{split}
\end{equation}

Combining the above estimates for $BL_3$ and  the energy inequality \eqref{ineq-energy-diff-l-order-1} with $l=3$ yields
\begin{equation}\label{ineq-energy-3-order}
\begin{split}
 \dtau E_3(\phi, t) + E_3(\phi, t)  \lesssim&  \left( t^{-1} + \varepsilon  M \right)   E_3(\phi,t)  \\
&+  \left( \varepsilon^2 I^2_4(\phi) +\varepsilon^3 M^3 \right) t^{-1+C_M \varepsilon}.
  \end{split}
\end{equation}
Then \eqref{improved-energy-l-order-1} with $l=3$ follows by Gr\"{o}nwall inequality.

{\bf Step V: the improvement for $E_4(\phi, t)$}.  Comparing to Step III, we now additionally take advantage of the previously improved $E_3(\phi, t)$ and $E_4(\hat N, t)$ to estimate the borderline term $BL_4$. It follows from \eqref{def-BLl-1234} that $BL_4=BL_4^1 + BL_4^2+BL_4^3+BL_4^4$.

$BL_4^1$ can be argued in the same way as $BL_1$ \eqref{ineq-energy-1-order-bl},  $BL_2^1$ \eqref{ineq-energy-2-order-bl-1}, and $BL_3^1$ \eqref{ineq-energy-3-order-bl-1},
\begin{equation}\label{ineq-energy-4-order-bl-1}
\begin{split}
 \int_{\Sigma_t}  |BL_4^1|
  \lesssim & \varepsilon^2 M^2 E_4 (\phi, t) + \left( \varepsilon^2 I^2_4(\phi) +\varepsilon^3 M^3 \right) t^{-1+C_M \varepsilon}.
  \end{split}
\end{equation}
Analogous to $BL_3^2$ \eqref{ineq-energy-3-order-bl-2}, $BL_4^2$ and $BL_4^3$ share the following estimates:
\begin{equation}\label{ineq-energy-4-order-bl-2-3}
\begin{split}
 \int_{\Sigma_t} |BL_4^2| + |BL_4^2| 
  \lesssim&  \left( \varepsilon^2 I^2_4(\phi) +\varepsilon^3 M^3 \right)  t^{-1+C_M \varepsilon} \\
  & +   \left( \varepsilon^3 M^3  +  \varepsilon^2 I^2_4(\phi)  \right) E_4(\phi,t).
\end{split}
\end{equation}
Finally,  following the argument for $BL_2^2$ \eqref{ineq-energy-2-order-bl-2}, $BL_3^3$ \eqref{ineq-energy-3-order-bl-3}, and noting that $\mathcal{E}^\frac{1}{2}_4 (\phi, t) \lesssim \mathcal{E}^\frac{1}{2}_3(\phi, t) + E^\frac{1}{2}_4(\phi, t)$, we obtain for $BL_4^4$,
\begin{equation}\label{ineq-energy-4-order-bl-4}
\begin{split}
 \int_{\Sigma_t}  |BL_4^4|  
\lesssim & \varepsilon^2 I_4^2(\phi) t^{-1+C_M \varepsilon} +   \left( \varepsilon^3 M^3  +  \varepsilon^2 I^2_4(\phi)  \right) E_4(\phi,t).
\end{split}
\end{equation}
Summarizing the above estimates, we derive
\begin{equation}\label{ineq-energy-diff-4-order}
\begin{split}
 \dtau E_4(\phi, t) + E_4(\phi, t)  \lesssim&  \left( t^{-1} + \varepsilon  M \right)   E_4(\phi,t)  \\
&+  \left( \varepsilon^2 I^2_4(\phi) +\varepsilon^3 M^3 \right) t^{-1+C_M \varepsilon}.
  \end{split}
\end{equation} 
Making use of the Gr\"{o}nwall inequality, we accomplish the proof for \eqref{improved-energy-l-order-1} with $l=4.$
\end{proof}

\section{Energy estimate for Weyl field}\label{sec-energy-Bianchi}
With improved energy for KG field \eqref{eq-improved-energy-0-order}, \eqref{improved-energy-l-order-1}, we are ready to improve energy estimates for Weyl tensor.
\begin{lemma}\label{prop-energy-estimate-Bianchi}
For some constant $C_M$ depending linearly on $M$, there is,
\begin{equation}\label{eq-energy-estimate-Bianchi-l-leq-2}
t^2 E_l(W,t) \lesssim  \varepsilon^2 I^2_{l}(W)+  \varepsilon^3 M^3, \quad 0 \leq l \leq 2.
\end{equation}
\begin{equation}\label{eq-energy-estimate-Bianchi-l-3}
t^2 E_3(W,t) \lesssim \left( \varepsilon^2 I^2_{4}(\phi) + \varepsilon^2 I^2_{3}(W)+  \varepsilon^3 M^3 \right)t^{C_M \varepsilon}.
\end{equation}
\end{lemma} 

We start with the following Corollary which will be proved later in subsections \ref{sec-estimate-error-bianchi} and subsection \ref{sec-source-bianchi}.
\begin{corollary}\label{coro-estimate-nonlinear-error-source-Bianchi}
For the nonlinear terms ${}_lLW_1, \, {}_lLW_2$ and source term $S_l$ in the Bianchi identities \eqref{energy-id-0-bianchi} (Lemma \ref{lemma-energy-id-1-Bianchi}) and  \eqref{energy-id-l-bianchi} with $l \leq 3$ (Corollary \ref{coro-energy-id-l-Bianchi}), we have
\begin{equation}\label{estimate-LW-1-2-Sl-leq-2}
\begin{split}
& \int_{\Sigma_t}  \sum_{l=0}^3 | {}_lLW_1| + |{}_lLW_2|  +  \sum_{l=0}^2 |S_l| \\
\lesssim & \mathcal{E}^{\frac{3}{2}}_3(g, t) +\left(  t^{-1} \mathcal{E}^{\frac{1}{2}}_3 (g, t) + \mathcal{E}_3 (g, t)\right) \mathcal{E}_4 (\phi, t).
\end{split}
\end{equation}
And the top order source term $S_3$ admits
\begin{equation}\label{estimate-S3}
\begin{split}
&\int_{\Sigma_t}   |S_{3}| \lesssim \varepsilon^2 M^2  E_3(W, t) + t^{-1} \mathcal{E}_4 (\phi, t).
\end{split}
\end{equation}
\end{corollary}

With the help of Corollary \ref{coro-estimate-nonlinear-error-source-Bianchi}, we now prove Lemma \ref{prop-energy-estimate-Bianchi}.
\begin{proof}[Proof of Lemma \ref{prop-energy-estimate-Bianchi}]
Substituting the results for $S_l, {}_lLW_1, {}_lLW_2 (0 \leq l \leq 2)$ \eqref{estimate-LW-1-2-Sl-leq-2} and the bootstrap assumption \eqref{BT-KG}-\eqref{BT-k-N-L-infty}, and estimates for $\mathcal{E}_i(\hat N, t)$ Lemma \ref{lem-N-ineq-source}, into the $l^{\text{th}}$ order energy identities \eqref{energy-id-l-bianchi}, we derive for $0 \leq l \leq 2,$
\begin{equation}\label{eq-energy-inequality-Bianchi-leq-2}
\begin{split}
 \dtau E_l(W,t) + 2 E_l(W, t)  \lesssim & \varepsilon^3 M^3 t^{-3+4\delta}, \quad 0<\delta < \frac{1}{6}.
 \end{split}
\end{equation}
We then change to the time function $t$  and  multiplying $t$ on both hand side of \eqref{eq-energy-inequality-Bianchi-leq-2} to obtain
\begin{equation*}
\begin{split}
 \dt (t^2 E_l(W,t))  \lesssim & \varepsilon^3 M^3 t^{-2+4\delta}, \quad 0<\delta < \frac{1}{6}.
 \end{split}
\end{equation*}
Then \eqref{eq-energy-estimate-Bianchi-l-leq-2} with $0 \leq l \leq 2$ follows from the Gr\"{o}nwall inequality.
Substituting the results for $S_3, {}_3LW_1, {}_3LW_2$ \eqref{estimate-LW-1-2-Sl-leq-2}, \eqref{estimate-S3} and the improved estimates for $\mathcal{E}_4 (\phi, t)$ (Lemma \ref{lemma-improved-l-energy-KG}) into the energy identity \eqref{energy-id-l-bianchi} with $l= 3$, we have
\begin{equation}\label{eq-energy-inequality-Bianchi}
\begin{split}
 \dtau E_3(W,t)  + 2 E_3(W, t) \lesssim &  \varepsilon M E_3(W, t)  \\
 &  +   \left( \varepsilon^2 I^2_{4}(\phi) +  \varepsilon^3 M^3 \right)t^{-2+C_M \varepsilon}.
 \end{split}
\end{equation}
Changing to $t$ and multiplying $t$ on both hand side of \eqref{eq-energy-inequality-Bianchi}, we deduce 
\begin{equation*}
\begin{split}
 \dt (t^2 E_3(W,t)) \lesssim &  \varepsilon M t^{-1} (t^2 E_3(W, t))  \\
 &  +   \left( \varepsilon^2 I^2_{4}(\phi) +  \varepsilon^3 M^3 \right)t^{-1+C_M \varepsilon}.
 \end{split}
\end{equation*}
An application of the Gr\"{o}nwall inequality then yields \eqref{eq-energy-estimate-Bianchi-l-3}.
\end{proof}

We remark that the energy estimates for the Weyl field are carried out only up to three derivatives, then $l \leq 3$ throughout this section.
\subsection{The nonlinear error terms}\label{sec-estimate-error-bianchi}
In this section we present the estimates for the nonlinear error terms ${}_lLW_1, {}_lLW_2 (l \leq 3)$.
\begin{corollary}\label{coro-estimate-nonlinear-error-Bianchi}
For the nonlinear error terms ${}_lLW_1, {}_lLW_2$ ($l \leq 3$), we have
\begin{equation}\label{estimate-LW-1-2}
\int_{\Sigma_t} |{}_0LW| + |{}_lLW_1| + |{}_lLW_2|  \lesssim \mathcal{E}^{\frac{3}{2}}_3(g, t) +   \mathcal{E}_3(g, t)  \mathcal{E}_3(\phi, t). 
\end{equation}
\end{corollary}

\begin{proof}
The error terms ${}_lLW_1, {}_lLW_2$ can be rearranged as ${}_lLW_1=\sum_{i=1}^3 {}_lLW_{1i},$
\begin{equation}\label{def-LB-1-re}
\begin{split}
{}_lLW_{11} &= \sum_{a+b +c=l } \nabla_{I_{a}} N *\nabla_{I_{b}} \hat k *  \nabla_{I_{c}} W * \nabla^{I_{l}} W, \\
{}_lLW_{12} & = \sum_{a+1+b =l+1 }  \nabla_{I_{a+1}} \hat N * \nabla_{I_{b}} W * \nabla^{I_{l}} W, \\
{}_lLW_{13 } &= \sum_{a+b=l } \nabla_{I_{a}} \hat N * \nabla_{I_{b}} W * \nabla^{I_{l}} W, 
\end{split}
\end{equation}
(note that, when $l=0$, the above ${}_lLW_1$ covers ${}_0LW$ \eqref{def-e-s-1-bianchi}),
and ${}_lLW_2$
\begin{equation}\label{def-LB-2-re}
\begin{split}
{}_lLW_{2} & =  \sum_{a+b =l-1} N \nabla_{I_{a}} \hat R_{imjn} * \nabla_{I_{b}} W * \nabla^{I_{l}} W.
\end{split}
\end{equation}

For ${}_lLW_{11}$, if $a=0$, noting that $[\frac{l}{2}] + 2 \leq 3 (l \leq 3)$, then we can perform $L^\infty, L^2, L^2$ or $L^2, L^\infty, L^2$ on the three factors; If $a \geq 1$, then $b,c \leq l-1$ and $a \leq l \leq 3$, we can perform $L^6, L^6, L^6, L^2$ on the four factors. Putting these two cases together, we have
\begin{equation}\label{estimate-LW11}
\int_{\Sigma_t} |{}_lLW_{11}|  \lesssim  \mathcal{E}^{\frac{3}{2}}_3(g, t)  + \mathcal{E}^{2}_3(g, t). 
\end{equation}

For ${}_lLW_{12}$ and ${}_lLW_{13}$, noting that $[\frac{l}{2}] + 2 \leq 3 (l \leq 3)$ and $[\frac{l}{2}] +1+ 2 \leq 4$ (to perform $L^\infty$ on $\nabla_{I_{a+1}} \hat N$ in ${}_lLW_{12}$), we can always perform $L^2, L^\infty, L^2$ or $L^2, L^\infty, L^2$ on the three factors in each of ${}_lLW_{12}$, ${}_lLW_{13}$ so that
\begin{equation}\label{estimate-LW1213}
\int_{\Sigma_t} |{}_lLW_{12}|  + |{}_lLW_{12}|  \lesssim  \mathcal{E}^{\frac{3}{2}}_3(g, t). 
\end{equation}

Finally, for ${}_lLW_{2}$, we recall that $\hat R_{imjn} = \hat k \pm |\hat k|^2 \pm E \pm |\D \phi|^2$. Since, $a, b \leq l-1$, we can always perform $L^4, L^4, L^2$ on the three factors (such as $\nabla_{I_a} |\hat k| * \nabla_{I_b} W * \nabla^{I_l} W$ ) or $L^6, L^6, L^6, L^2$ on the four factors ($\nabla_{I_a} |\hat k|^2 * \nabla_{I_b} W * \nabla^{I_l} W$) in ${}_lLW_{2}$. As a consequence, we derive 
\begin{equation}\label{estimate-LW1213}
\int_{\Sigma_t} |{}_lLW_{2}|  \lesssim  \mathcal{E}^{\frac{3}{2}}_3(g, t) +   \mathcal{E}_3(g, t)  \mathcal{E}_3(\phi, t).
\end{equation}

Summarize all these estimates, we obtain \eqref{estimate-LW-1-2}
\end{proof}

\subsection{The nonlinear coupling terms}\label{sec-source-bianchi}
In this section, we will prove the following estimate for the nonlinear couplings in the Bianchi equations.
\begin{corollary}\label{coro-estimate-source-Bianchi}
For the coupling terms $S_l, l \leq 3$, (see \eqref{energy-id-0-bianchi} and \eqref{energy-id-l-bianchi}), we have
\begin{equation}\label{estimate-Sl-leq-2}
 \int_{\Sigma_t}  |S_{l}|  \lesssim \left(  t^{-1} \mathcal{E}^{\frac{1}{2}}_l (g, t) + \mathcal{E}_l (g, t)\right) \mathcal{E}_4 (\phi, t), \quad 0 \leq l \leq 2,
\end{equation}
and the top order case $S_3$ admits  
\begin{equation}\label{estimate-Source-S3}
 \int_{\Sigma_t}   |S_{3}| \lesssim \varepsilon^2 M^2  E_3(W, t) + t^{-1} \mathcal{E}_4 (\phi, t).  
\end{equation}
\end{corollary} 

We postpone the proof to subsection \ref{sec-estimate-source-bianchi} and explore the structures for the couplings first.
\subsubsection{The structures of the couplings}\label{sec-stru-source-bianchi}
The source terms $S_l, \, l \leq 3$ can be split into electric and magnetic parts: $S_l=S_{l1} + S_{l2}$, where
\begin{equation}\label{def-S-re}
\begin{split}
S_{l1} = & \nabla_{I_l} \left( N J_{i \Tbar j} \right) \nabla^{I_l} E^{ij},  \\
S_{l2} = & \nabla_{I_l} \left( N J^*_{i \Tbar j} \right)  \nabla^{I_l} H^{ij}.
\end{split}
\end{equation}
 Due to the fact that $E_{ij}$ is trace free and $T$-tangent, we have
\begin{equation}\label{stru-source-J-bianchi-0}
\begin{split}
  J_{i \Tbar j} E^{ij} =& \frac{1}{2} \left( D_{\Tbar} \bar R_{ij} - D_j \bar R_{\Tbar i} \right) E^{ij} \\
  =&  \frac{1}{2} \left(D_{\Tbar} D_i \phi D_j \phi - D_i D_j \phi  \Tbar \phi \right) E^{ij} \\
  = & \left( \nabla \Tbar \phi \pm k_{i}^j * \nabla \phi \right) * \nabla \phi * E \pm \nabla^2 \phi * \Tbar \phi * E \pm (\Tbar \phi)^2 \hat k * E.
\end{split}
\end{equation}
and 
\begin{equation}\label{stru-source-AJ-bianchi-0}
\begin{split}
 J^*_{i \Tbar j} H^{ij} 
  =&   \frac{1}{4} \left( D_\mu \bar R_{\nu i} - D_\nu \bar R_{\mu i} \right) \cdot \epsilon^{\mu \nu} {}_{\! \Tbar j}  H^{ij} \\
  & \quad \quad  - \frac{1}{24} \left( g_{i\nu} D_\mu \bar R - g_{i\mu} D_\nu \bar R \right) \cdot \epsilon^{\mu \nu} {}_{\! \Tbar j}  H^{ij}\\
  =&   \frac{1}{4} \left( D_p D_i \phi  D_q \phi - D_q D_i \phi  D_p \phi \right) \epsilon^{pq} {}_{\! j}  H^{ij} +Tr_1 +Tr_2, 
  \end{split}
\end{equation}
where
  \begin{equation*}
\begin{split}
Tr_1 &= \frac{1}{4} \bar m^2(t) \phi D_p \phi \epsilon^{p} {}_{\! ij}  H^{ij} -  \frac{1}{4} \bar m^2(t) \phi D_q \phi  \epsilon_i{}^{ q} {}_{\! j}  H^{ij}, \\
Tr_2 &=   - \frac{1}{24} D_p \bar R  \epsilon^{p} {}_{\! ij}  H^{ij} +  \frac{1}{24} D_q \bar R  \epsilon_i{}^{q} {}_{\! j}  H^{ij}.
\end{split}
\end{equation*}
Notice that, $\epsilon^{p} {}_{\! ij}, \epsilon_i{}^{\! q} {}_{\! j}$ are both antisymmetric in $i$ and $j$, while $H_{ij}$ is symmetric in $i,j.$ Hence, $Tr_1=Tr_2=0$. With the same reason, $ J^*_{i \Tbar j} H^{ij}$ can be further reduced as
\begin{equation}\label{stru-source-AJ-bianchi-0-1}
\begin{split}
  J^*_{i \Tbar j} H^{ij} =& \frac{1}{4} \left( D_p D_i \phi  D_q \phi - D_q D_i \phi  D_p \phi \right)  \epsilon^{pq} {}_{\! j}  H^{ij} \\
= & \left( \nabla^2 \phi \pm \hat k * \Tbar \phi \right) * \nabla \phi * H.
\end{split}
\end{equation}

The fact that $E_{ij}$ is trace free and $T$-tangent can be generalized to the case of high order derivatives. Obviously, $\nabla_{I_l} E_{i j}$, which is projected to $\Sigma_t$, is $T-$tangent.
In addition to this, we can derive that $ g^{ij}\nabla_{I_l} E_{ij} =0$. Hence, the results of \eqref{stru-source-J-bianchi-0} and \eqref{stru-source-AJ-bianchi-0-1} can be generalized as
\begin{equation}\label{stru-source-J-bianchi-1}
\begin{split}
 & \nabla_{I_l} \left( NJ_{i \Tbar j} \right) \nabla^{I_l} E^{ij} \\
=&\pm \nabla_{I_l} \left(N \nabla \Tbar \phi * \nabla \phi \pm  N \nabla^2 \phi * \Tbar \phi \right) * \nabla^{I_l} E\\
&\pm  \nabla_{I_l} \left( N\left(k_i^j * \nabla \phi * \nabla \phi \pm \hat k* (\Tbar \phi)^2 \right) \right) * \nabla^{I_l} E,
\end{split}
\end{equation}
and
\begin{equation}\label{stru-source-AJ-bianchi-1}
\begin{split}
 &  \nabla_{I_l} \left( N J^*_{i \Tbar j} \right) \nabla^{I_l} H^{ij} \\
= &   \nabla_{I_l} \left( N \left( \nabla^2 \phi \pm \hat k* \Tbar \phi \right) * \nabla \phi \right) * \nabla^{I_l} H.
\end{split}
\end{equation}

\subsubsection{The estimate for the coupling terms}\label{sec-estimate-source-bianchi}
We rewrite \eqref{stru-source-J-bianchi-1} and \eqref{stru-source-AJ-bianchi-1} as
\begin{equation}\label{def-Sl1-re}
\begin{split}
S_{l1} =  \sum_{a+b =l }& \nabla_{I_a} N * \nabla_{I_b} (\nabla \Tbar \phi * \nabla \phi) \nabla^{I_l} E^{ij}  \\
&+  \nabla_{I_a} N * \nabla_{I_b} (\nabla^2 \phi * \Tbar \phi) \nabla^{I_l} E^{ij}  \\
&+ \nabla_{I_a} N * \nabla_{I_b} (\nabla \phi * \nabla \phi * k_i^j) \nabla^{I_l} E^{ij}  \\
&+  \nabla_{I_a} N * \nabla_{I_b} ((\Tbar \phi)^2 * \hat k) \nabla^{I_l} E^{ij},
\end{split}
\end{equation}
and 
\begin{equation}\label{def-Sl2-re}
\begin{split}
S_{l2} =  \sum_{a+b =l }&  \nabla_{I_a} N * \nabla_{I_b} (\nabla^2\phi * \nabla \phi ) \nabla^{I_l} H^{ij}  \\
&+  \nabla_{I_a} N * \nabla_{I_b} ( \hat k* \Tbar \phi * \nabla \phi ) \nabla^{I_l} H^{ij}.
\end{split}
\end{equation}
Notice that, then energy estimate argument for the Weyl tensor is done up to three order derivatives, hence in this section, we have $l \leq 3.$

For the electric part $S_{l1},$ 
\begin{itemize}
\item {\bf Case $S_{l1}$-I}: $a=0.$ Then $b=l.$ $S_{l1}$-I consists
\begin{equation}\label{def-S-l11-1234}
\begin{split}
 S_{l11}^1= &N  \nabla_{I_l} \left( \nabla \Tbar \phi * \nabla \phi \right) * \nabla^{I_l} E,\\
 S^2_{l11} =& N \nabla_{I_l} \left( \nabla^2 \phi * \Tbar \phi \right) * \nabla^{I_l} E,\\
S^3_{l11} =&  N \nabla_{I_l} \left( k_i^j * \nabla \phi * \nabla \phi \right) * \nabla^{I_l} E, \\
S^4_{l11} =&N \nabla_{I_l} \left( (\Tbar \phi)^2  \hat k \right) * \nabla^{I_l} E.
\end{split}
\end{equation}

For $S_{l11}^1,$ we further rewrite it as 
\begin{equation}\label{def-S-l11-1}
S_{l11}^1= \sum_{p+q=l} (mt)^{-1} N \nabla_{I_{p+1}} \Tbar \phi * \nabla_{I_{q+1}} (\bar m (t) \phi) * \nabla^{I_l} E.
\end{equation}
 Since $[\frac{l}{2}] + 1 + 2 \leq 4 (l \leq 3),$ we can apply $L^2, L^\infty$ or $L^\infty, L^2$ to the two factors in $\nabla_{I_{p+1}} \Tbar \phi *\nabla_{I_{q+1}} \phi$ and derive 
\begin{equation}\label{estimate-S11-1}
\begin{split}
\int_{\Sigma_t} |S_{l11}^1| \lesssim & t^{-1} \mathcal{E}^{\frac{1}{2}}_l (W, t) \mathcal{E}_4 (\phi, t)   \lesssim  \varepsilon^3 M^3  t^{-3+3\delta}.
\end{split}
\end{equation}

For $S_{l11}^2,$ we further rewrite it as 
\begin{align*}
S_{l11}^2= \sum_{p+q=l} N \nabla_{I_{p+1}}  \nabla \phi * \nabla_{I_{q}}  \Tbar \phi *  \nabla^{I_l} E.
\end{align*}
\begin{enumerate}
\item In the case of $0 \leq l \leq 2$.  We rearrange it as $$\sum_{p+q=l \leq 2} t^{-1} N \nabla_{I_{p+2}} (\bar m (t) \phi) * \nabla_{I_q} \Tbar \phi *  \nabla^{I_l} E.$$
We apply $L^4, L^4, L^2$ (if $p \leq 1$) or $L^2, L^\infty, L^2$ (if $p=2$) on the three factors, since $p+q \leq 2$. With the same argument for $S_{l11}^1$ \eqref{estimate-S11-1}, we derive
\begin{equation}\label{estimate-S011-leq-2}
\begin{split}
\int_{\Sigma_t} |S_{l11}^2| \lesssim & t^{-1} \mathcal{E}^{\frac{1}{2}}_l (W, t) \mathcal{E}_4 (\phi, t), \, l \leq 2.
\end{split}
\end{equation}

\item In the case of $l =3.$

\begin{itemize}
 
\item[i)] If $p \leq l-1\leq 2,$ then $\nabla_{I_{p+2}}\phi$ will never exceed the top ($4$) order derivative. We can perform the same argument as in the case of $0 \leq l \leq 2$ \eqref{estimate-S011-leq-2}. The same bound $ t^{-1} \mathcal{E}^{\frac{1}{2}}_l (\phi, t) \mathcal{E}_4 (\phi, t) $ can be obtained as well.

\item[ii)] If $p=l=3$, the argument for \eqref{estimate-S011-leq-2} are not valid here. Because of the regularity, $\nabla^5(\bar m (t) \phi)$ is not allowed, hence $\nabla^4 \nabla \phi$ can not absorb an addition $t$ (or $\bar m (t)$) as the previous case, which makes it the borderline case. However, as in the proof that leads to Lemma \ref{lemma-improved-l-energy-KG}, the improved $\|\Tbar\phi\|_{L^\infty}$ and $E_{4}(\phi,t)$ \eqref{improved-energy-l-order-1} will help to linearize this borderline term, that is
\begin{equation}\label{estimate-Sl11-2-1}
\begin{split}
\int_{\Sigma_t} |S_{311}^2| =& \int_{\Sigma_t} \big| N  \nabla^4 \nabla \phi * \Tbar \phi *  \nabla^3 E \big| \\
\lesssim &  \| \nabla^4 \nabla \phi  \|_{L^2(\Sigma_t)} \|\Tbar \phi\|_{L^\infty}  \|\nabla^{3} E\|_{L^2(\Sigma_t)} \\
 \lesssim & E^{\frac{1}{2}}_4 (\phi, t) \cdot  \varepsilon M t^{-\frac{1}{2}} \cdot E^\frac{1}{2}_3(W, t) \\
 \lesssim & \varepsilon^2 M^2  E_3(W, t) + t^{-1} E_4 (\phi, t).
\end{split}
\end{equation}
 \end{itemize}
\end{enumerate}

Relative to $S_{l11}^1$ and $S_{l11}^2$, $S_{l11}^3$ and $S_{l11}^4$ are in fact lower order terms.  

$S_{l11}^3$ can be rewritten as 
\begin{align*}
S_{l11}^3= \sum_{p+q+h=l} N  \nabla_{I_h} k_i^j * \nabla_{I_{p+1}} \phi * \nabla_{I_{q+1}} \phi * \nabla^{I_l} E.
\end{align*}
If $h=0,$ it becomes $\sum_{p+q=l} N k_i^j * \nabla_{I_{p+1}} \phi * \nabla_{I_{q+1}} \phi * \nabla^{I_l} E$, which is equivalent to (noting that $p+1, q+1 \leq l+1 \leq 4$)
\begin{equation}\label{S-l-11-3-h=0}
\sum_{p+q=l} t^{-2} N  \nabla_{I_{p+1}} (\bar m (t) \phi) * \nabla_{I_{q+1}}(\bar m (t) \phi) * \nabla^{I_l} E,
\end{equation}
 Quantitatively, \eqref{S-l-11-3-h=0} turns out to be equivalent to $t^{-1} S_{l11}^1$ \eqref{def-S-l11-1}. Thus, similar to \eqref{estimate-S11-1}, \eqref{S-l-11-3-h=0} can be bounded by $t^{-2} \mathcal{E}^{\frac{1}{2}}_l (W, t) \mathcal{E}_4 (\phi, t)$.
If $h \geq 1,$ we get
\begin{equation}\label{S-l-11-3-h-geq-1}
\sum_{p+q+h=l \atop h \geq 1} N  t^{-2} \nabla_{I_h} \hat k* \nabla_{I_{p+1}}(\bar m (t) \phi ) * \nabla_{I_{q+1}} (\bar m (t) \phi) * \nabla^{I_l} E.
\end{equation}
Knowing that $1 \leq h \leq l \leq 3,$ and $0 \leq p, q \leq l-1 \leq 2,$ we can apply $L^6, L^6, L^6, L^2$ \eqref{S-l-11-3-h-geq-1} and derive  the bound $t^{-2} \mathcal{E}_l (g, t) \mathcal{E}_4 (\phi, t)$, which is of lower order than \eqref{S-l-11-3-h=0}. As a summary, we have
\begin{equation}\label{estimate-S11-3}
\begin{split}
 \int_{\Sigma_t} |S_{l11}^3|
\lesssim t^{-2} \mathcal{E}^{\frac{1}{2}}_l (g, t) \mathcal{E}_4 (\phi, t).
\end{split}
\end{equation}

$S_{l11}^4$ can be further rewritten as 
\begin{equation}\label{def-S-l11-4}
S_{l11}^4= \sum_{p+q + h=l} N \nabla_{I_p} \Tbar \phi * \nabla_{I_q} \Tbar \phi *\nabla_{I_h} \hat k * \nabla^{I_l} E.
\end{equation}
Noting that $p,q,h \leq l \leq 3,$  we can apply $L^6, L^6, L^6, L^2$ to the four factors to derive 
\begin{equation}\label{estimate-S11-4}
\begin{split}
 \int_{\Sigma_t}|S_{l11}^4| \lesssim \mathcal{E}_l (g, t) \mathcal{E}_4 (\phi, t).
\end{split}
\end{equation}

\item {\bf Case $S_{l1}$-II}: $a\geq 1.$  Letting $a =c+1, c\geq 0,$ $S_{l1}$-II composes of  
 \begin{equation}\label{def-S-l12-1234}
\begin{split}
 S_{l12}^1= & \sum_{c+1+b=l}\nabla_{I_{c+1}} \hat N * \nabla_{I_b} \left( \nabla \Tbar \phi * \nabla \phi \right) *  \nabla^{I_l} E\\
 S^2_{l12} =& \sum_{c+1+b=l} \nabla_{I_{c+1}} \hat N * \nabla_{I_b} \left( \nabla^2 \phi * \Tbar \phi \right) * \nabla^{I_l} E\\
S^3_{l12} =&  \sum_{c+1+b=l} \nabla_{I_{c+1}} \hat N * \nabla_{I_b} \left( k_i^j * \nabla \phi * \nabla \phi \right) * \nabla^{I_l} E \\
S^4_{l12} =& \sum_{c+1+b=l} \nabla_{I_{c+1}} \hat N* \nabla_{I_b} \left( (\Tbar \phi)^2 \hat k \right) * \nabla^{I_l} E.
\end{split}
\end{equation}
In fact, $S_{l1}$-II \eqref{def-S-l12-1234} enjoys better estimate than $S_{l1}$-I \eqref{def-S-l11-1234}. This can be roughly verified by comparing $N$ with $\hat N$.  
$S^1_{l12}$ can be rearranged as (for $b=b_1 + b_2 \leq l-1\leq 2$, then $b_1 +1, b_2 +1 \leq l \leq 3$), 
\begin{equation}\label{Sl12-1}
\begin{split}
&  \sum_{c+1+b_1 + b_2 =l } t^{-1}\nabla_{I_{c+1}} \hat N * \nabla_{I_{b_1+1}} \Tbar \phi * \nabla_{I_{b_2+1}} (\bar m (t) \phi) * \nabla^{I_l} E.
\end{split}
\end{equation}
Here $c+1 \leq 3$ and $b_1+1, b_2+1 \leq 3$, we can apply $L^6, L^6, L^6, L^2$ to the four factors in \eqref{Sl12-1} to derive
\begin{equation}\label{estimate-Sl12-1}
\begin{split}
& \int_{\Sigma_t} |S_{l12}^1| \lesssim  t^{-1} \mathcal{E}_3 (g, t) \mathcal{E}_4(\phi, t).
\end{split}
\end{equation}

$S^2_{l12}$ can be rearranged as (for $b=b_1 + b_2 \leq l-1\leq 2$, then $b_1 +2 \leq l+1 \leq 4$), 
\begin{equation}\label{Sl12-1}
\begin{split}
&  \sum_{c+1+b_1 + b_2 =l } t^{-1}\nabla_{I_{c+1}} \hat N * \nabla_{I_{b_1+2}} (\bar m (t) \phi) * \nabla_{I_{b_2}} \Tbar \phi * \nabla^{I_l} E.
\end{split}
\end{equation}
Here $c+1 \leq 3$ and $b_1, b_2 \leq 2$. If $b_1 \leq 1$, then $b_1 +2 \leq 3,$ and we can apply $L^6, L^6, L^6, L^2$ to the four facts in in \eqref{Sl12-1}; If $b_1 =2$, then $c=b_2=0$. We can use $L^\infty, L^2, L^\infty, L^2$ alternatively. Finally, $S_{l12}^2$ admits the estimate
\begin{equation}\label{estimate-Sl12-2}
\begin{split}
& \int_{\Sigma_t} |S_{l12}^2| \lesssim  t^{-1} \mathcal{E}_3 (g, t) \mathcal{E}_4(\phi, t).
\end{split}
\end{equation}

$S^3_{l12}$ can be rearranged as ($c+1+b_1 + b_2 + b_3 =l $ and $b_2 +1, b_3 +1 \leq l \leq3$),
 \begin{equation}\label{Sl12-3}
 t^{-2} \nabla_{I_{c+1}} \hat N * \nabla_{I_{b_1}} k_i^j * \nabla_{I_{b_2+1}} (\bar m (t) \phi ) * \nabla_{I_{b_3+1}} (\bar m (t) \phi) * \nabla^{I_l} E.
\end{equation}
Here $c+1, b_2+1, b_3+1 \leq l \leq 3$. If $b_1=0$, $k_i^j \sim 1$, we apply $L^6, L^\infty, L^6, L^6, L^2$ on the five factors in \eqref{Sl12-3} and derive the bound $t^{-2} \mathcal{E}_3 (g, t) \mathcal{E}_4 (\phi, t)$; If $1 \leq b_1,$ noting that $b_1 \leq l-1\leq 2,$  we apply $L^6, L^\infty, L^6, L^6, L^2$ on the five factors as well and the bound is replaced by $t^{-2}  \mathcal{E}^{\frac{3}{2}}_3 (g, t) \mathcal{E}_4 (\phi, t)$ which is of lower order. Finally, we have
\begin{equation}\label{estimate-Sl12-3}
\begin{split}
& \int_{\Sigma_t}  |S_{l12}^3| 
 \lesssim t^{-2} \mathcal{E}_3 (g, t) \mathcal{E}_4 (\phi, t).
\end{split}
\end{equation}

$S^4_{l12}$ can be rearranged as ($c+1+b_1 + b_2 + b_3 =l $), 
\begin{equation}\label{Sl12-4}
\begin{split}
& \nabla_{I_{c+1}} \hat N * \nabla_{I_{b_1}} \Tbar \phi * \nabla_{I_{b_2}} \Tbar \phi  * \nabla_{I_{b_3}} \hat k * \nabla^{I_l} E.
\end{split}
\end{equation}
Now that, $c+b_1 + b_2 + b_3 =l -1 \leq 2$, we can always apply $L^2, L^\infty, L^\infty, L^\infty, L^2$ to the five factors of \eqref{Sl12-4} and derive that 
\begin{equation}\label{estimate-Sl12-4}
\begin{split}
& \int_{\Sigma_t}  |S_{l12}^4| 
 \lesssim \mathcal{E}^{\frac{3}{2}}_3 (g, t) \mathcal{E}_4 (\phi, t).
\end{split}
\end{equation}
\end{itemize}

Putting  {\bf Case $S_{l1}$-I} and  {\bf Case $S_{l1}$-II} together, we have for $S_{l1},$
\begin{align}
 \int_{\Sigma_t}  |S_{l1}| &\lesssim \left(  t^{-1} \mathcal{E}^{\frac{1}{2}}_l (g, t) + \mathcal{E}_l (g, t)\right) \mathcal{E}_4 (\phi, t), \quad 0 \leq l \leq 2, \label{estimate-Sl1-l-leq-2}\\
 \int_{\Sigma_t}   |S_{31}| &\lesssim \varepsilon^2 M^2  E_3(W, t) + t^{-1} \mathcal{E}_4 (\phi, t). \label{estimate-Sl1-l=3}
\end{align}

For the magnetic part $S_{l2}$, 
\begin{itemize}
\item {\bf Case $S_{l2}$-I}: $a=0.$ Then $b=l.$ $S_{l2}$-I composes of 
\begin{equation}\label{def-S-l21-12}
\begin{split}
 S_{l21}^1= &N  \nabla_{I_{l}} \left( \nabla^2 \phi * \nabla \phi \right) * \nabla^{I_{l}} H, \\
S^2_{l21} =&  N \nabla_{I_{l}} \left(\hat k* \nabla \phi * \Tbar \phi \right) * \nabla^{I_{l}} H.
\end{split}
\end{equation}
Note that, $S^1_{l21}, S^2_{l21}$ possess the same structures as $S_{l11}^1, S_{l11}^4$ \eqref{def-S-l11-1234}  respectively if we replace $\nabla^2 \phi$ in $S^1_{l21}$ by $\nabla \Tbar \phi$, $\nabla \phi$ in $S^2_{l21}$ by $\Tbar \phi$. In fact, the structures in $S^1_{l21}, S^2_{l21}$ is better, since they both have one more order of spatial derivative than $S_{l11}^1, S_{l11}^4$ respectively.  Anyway, the analogous argument for $S_{l11}^1, S_{l11}^4$ can be applied to $S^1_{l21}, S^2_{l21}$ as well. Then 
\begin{equation}\label{estimate-Sl21-1-2}
\begin{split}
& \int_{\Sigma_t}  |S_{l21}^1| +  |S_{l21}^2|
 \lesssim  \left(  t^{-1} \mathcal{E}^{\frac{1}{2}}_l (g, t) + \mathcal{E}_l (g, t)\right) \mathcal{E}_4 (\phi, t).
\end{split}
\end{equation}

\item {\bf Case $S_{l2}$-II}: $a\geq 1.$ Letting $a =c+1, c \geq 0,$ $S_{l2}$-II consists
  \begin{equation}\label{def-S-l22-12}
\begin{split}
 S_{l22}^1= & \sum_{c+1+b=l}\nabla_{I_{c+1}} \hat N *  \nabla_{I_b} \left( \nabla^2 \phi * \nabla \phi \right) * \nabla^{I_l} H, \\
S^2_{l22} =& \sum_{c+1+b=l} \nabla_{I_{c+1}} \hat N* \nabla_{I_b} \left( \Tbar \phi * \nabla \phi *  \hat k \right) * \nabla^{I_l} H.
\end{split}
\end{equation}
Again, if replacing $\nabla^2 \phi$ in $S^1_{l22}$, $\nabla \phi$ in $S^2_{l22}$ by $\nabla \Tbar \phi, \Tbar \phi$ respectively, one gets the same structures as $S_{l12}^1, S_{l12}^4$.  Thus the estimate for $S_{l22}^1, S_{l22}^2$ will be just in an analogous to that for $S_{l12}^1, S_{l12}^4$. We obtain  
\begin{equation}\label{estimate-Sl22-12}
\begin{split}
& \int_{\Sigma_t}  |S_{l21}^1| + |S_{l22}^2| \lesssim  \left(  t^{-1} \mathcal{E}_3 (g, t) + \mathcal{E}^{\frac{3}{2}}_3 (g, t)\right) \mathcal{E}_4 (\phi, t).
\end{split}
\end{equation}
This is of course lower order than \eqref{estimate-Sl21-1-2}.
\end{itemize}  
Combining the results for Case $S_{l2}$-I  with  Case $S_{l2}$-II, we have for $S_{l2},$
\begin{equation}\label{estimate-S2}
 \int_{\Sigma_t}   |S_{l2}| \lesssim \left(  t^{-1} \mathcal{E}^{\frac{1}{2}}_l (g, t) + \mathcal{E}_l (g, t)\right) \mathcal{E}_4 (\phi, t).
\end{equation}

\section{The global existence}\label{sec-global-existence}
\subsection{Closing the bootstrap argument}\label{sec-close-b-t}
\subsubsection{For KG field and Weyl field}
In Lemma \ref{lemma-improved-0-energy-KG} and Lemma \ref{lemma-improved-l-energy-KG} and Lemma \ref{prop-energy-estimate-Bianchi}, we have proven that, for some universal constant $C$, and some constant $C_M$ depending linearly on $M$,
\begin{equation*}
\begin{split}
t E_0(\phi,t) \leq & C \varepsilon^2 I^2_{0}(\phi), \\
t E_l(\phi, t) \leq&C \left( \varepsilon^2 I^2_4(\phi) + \varepsilon^3 M^3 \right) t^{C_M \varepsilon}, \quad 1 \leq l \leq 4.
\end{split}
\end{equation*}
and 
\begin{equation*}
\begin{split}
t^2 E_l(W,t) \leq&  C\left( \varepsilon^2 I^2_{l}(W)+  \varepsilon^3 M^3 \right), \quad 0 \leq l \leq 2. \\
t^2 E_3(W,t) \leq& C\left( \varepsilon^2 I^2_{4}(\phi) + \varepsilon^2 I^2_{3}(W)+ \varepsilon^3 M^3 \right)t^{C_M \varepsilon}.
\end{split}
\end{equation*}
Then we choose $M$ satisfying $C( I^2_3(W) +I_4^2(\phi) )\leq \frac{M^2}{4}$ and $\varepsilon$ small enough such that $\varepsilon CM \leq \frac{1}{4}$, $C_M \varepsilon \leq \delta$. Hence, for $0\leq i \leq 4, 0 \leq j \leq 3,$
\begin{equation}\label{eq-closing-b-t-KG-Weyl}
t E_{i}(\phi,t) \leq \frac{\varepsilon^2 M^2}{2} t^{\delta}, \quad t^2 E_j(W,t) \leq \frac{\varepsilon^2 M^2}{2} t^{\delta},
\end{equation}
which improve the bootstrap assumptions \eqref{BT-KG}-\eqref{BT-E-H}.

\subsubsection{For the second fundamental form}\label{sec-closing-bt-k}
We will additionally take advantage of the transport equation for the second fundamental form to derive a bound for $\|\hat k\|_{L^2(\Sigma_t)}$. This helps to improve the weak bootstrap assumption for $\hat k$ \eqref{BT-k-N-L-infty}. 
\begin{lemma}[Improvement for $\hat k$]\label{lem-hat-k-estimate-original-0}
Fixing $0<\delta<\frac{1}{6},$  we have for some universal constant $C$,
\begin{equation}\label{eq-improved-hat-k-0}
 t^2 E_0(\hat k, t) \leq C   \varepsilon^2 I_0^2(\hat k)+ C \left( \varepsilon I_0(\phi) +  \varepsilon I_0(W) + \varepsilon^\frac{3}{2} M^\frac{3}{2} \right) \varepsilon M t^{\delta}.
\end{equation}
\end{lemma}

\begin{proof}
Reminding the evolution equation for $\hat k_{ij}$ \eqref{eq-evolution-2},
and noticing that
\begin{align*}
\bar R_{i\Tbar j \Tbar}  \hat k^{ij} 
=& \left( E_{ij} - \frac{1}{2}\bar R_{ij} \right) \hat k^{ij} =  \left( E_{ij} - \frac{1}{2} \nabla_i \phi \nabla_j \phi \right) \hat k^{ij},
\end{align*}
we multiply $\hat k^{ij}$ and integrate on $\Sigma_t,$ 
\begin{equation}\label{eq-transport-E0-hat-k-1} 
\begin{split}
\dtau E_0(\hat k, t) + 2 E_0(\hat k, t) =&\int_{\Sigma_t} \hat k^{ij}  \left( -2\nabla_i \nabla_j N + 2NE_{ij}  \right) \\
& - \hat N |\hat k|^2 + 2N |\hat k|^3-N \nabla_i \phi \nabla_j \phi \hat k^{ij} .
\end{split}
\end{equation}

Terms in the second line of \eqref{eq-transport-E0-hat-k-1} are cubic, and hence easier.
We make use of $\|\hat k\|_{L^\infty}$ \eqref{BT-k-L-infty-4-6} and $\|\hat N\|_{L^\infty}$ \eqref{BT-k-N-L-infty}, then
\begin{align*}
 \int_{\Sigma_t}  |\hat N| |\hat k|^2 + N |\hat k|^3 &\lesssim  \mathcal{E}^{\frac{1}{2}}_1(g, t) E_0(\hat k, t) \lesssim \varepsilon M t^{-1+2\delta}  E_0(\hat k, t).
\end{align*}
Rewrite $\nabla_i \phi * \nabla_j \phi* \hat k^{ij}$  as $(mt)^{-2} \nabla_i (\bar m (t) \phi)* \nabla_j (\bar m (t) \phi ) *\hat k^{ij}$ and perform $L^4, L^4, L^2$ on the three factors
\begin{align*}
\int_{\Sigma_t}  | \nabla_i \phi \nabla_j \phi \hat k^{ij}| 
 \lesssim & t^{-2}  \mathcal{E}_2(\phi, t) E^{\frac{1}{2}}_0(\hat k, t) \\
  \lesssim& \varepsilon M t^{-1+\delta}  E_0(\hat k, t) + \varepsilon^3 M^3 t^{-5+\delta}.
\end{align*}

The first line in \eqref{eq-transport-E0-hat-k-1} involves quadratic terms. We will make use of the refined $\mathcal{E}_2(\hat N, t)$  \eqref{N-improved-estimate-source}, and $E_0(W, t)$ \eqref{eq-energy-estimate-Bianchi-l-leq-2} to linearize them:
\begin{align*}
 \int_{\Sigma_t}  |\nabla_i \nabla_j N \hat k^{ij}|   
 &\lesssim E^{\frac{1}{2}}_2(\hat N, t) E^{\frac{1}{2}}_0(\hat k, t) \lesssim \varepsilon I_0(\phi) t^{-1}   E^{\frac{1}{2}}_0(\hat k, t), \\
\int_{\Sigma_t}  |E_{ij} \hat k^{ij}| 
&\lesssim  E^{\frac{1}{2}}_0(W, t) E^{\frac{1}{2}}_0(\hat k, t) \lesssim \left( \varepsilon I_0(W) + \varepsilon^\frac{3}{2} M^\frac{3}{2} \right) t^{-1} E^{\frac{1}{2}}_0(\hat k, t).
\end{align*}
Now, we substitute the bootstrap assumption for $E^\frac{1}{2}_0(\hat k, t)$ (see \eqref{BT-k-N-L-infty}) to achieve
\begin{equation*}
\begin{split}
\dtau E_0(\hat k, t) &+ 2 E_0(\hat k, t) \lesssim   \varepsilon M t^{-1+2\delta} E_0(\hat k, t) \\
&+  \left( \varepsilon I_0(\phi)+ \varepsilon I_0(W) + \varepsilon^\frac{3}{2} M^\frac{3}{2} \right)\varepsilon M t^{-2+\delta}.
\end{split}
\end{equation*}
Changing to $t$ and multiplying $t$ on the inequality, we achieve
\begin{equation}\label{ode-K-1}
\begin{split}
& \dt (t^2 E_0(\hat k, t)) \lesssim   \varepsilon M t^{-2+2\delta} (t^2 E_0(\hat k, t)) \\
&\quad +   \left( \varepsilon I_0(\phi)+ \varepsilon I_0(W) + \varepsilon^\frac{3}{2} M^\frac{3}{2} \right)\varepsilon M t^{-1+\delta}.
\end{split}
\end{equation}
where $0<\delta < \frac{1}{6}$. By the Gr\"{o}nwall inequality, 
\begin{equation*}
\begin{split}
 t^2 E_0(\hat k, t) \lesssim & \exp(\varepsilon M) \left(\varepsilon^2 I^2_0(\hat k) + \left( \varepsilon I_0(\phi) +  \varepsilon I_0(W) + \varepsilon^\frac{3}{2} M^\frac{3}{2} \right) \varepsilon M t^{\delta} \right).
\end{split}
\end{equation*}
$\varepsilon$ is chosen small enough such that $\varepsilon M <1$, thus it implies \eqref{eq-improved-hat-k-0}.

\end{proof}

$\hat k_{ij}$ being the trace free part of $k_{ij},$ satisfies the following elliptic system,
\begin{equation}\label{div-curl-k-hat}
\begin{split}
(\dive \hat k)_i &= - \bar R_{\Tbar i},\\
\curl \hat k_{ij} &= - H_{ij}, \\
\tr \hat k &=0.
\end{split}
\end{equation}
See Appendix for the definition of div and curl. 

\begin{lemma}[Improvement for derivative of $\hat k$]\label{lem-hat-k-estimate-original}
Fixing $0<\delta<\frac{1}{6},$ we have for $0 \leq i \leq 3$ and some universal constant $C$
\begin{equation}\label{estimate-nabla-1-4-hat-k}
t^2 E_{i+1}(\hat k, t)  \leq C  \varepsilon^2 I_0^2(\hat k)+ C\left( \varepsilon I_{3}(W) + \varepsilon I_4(\phi) +  \varepsilon^\frac{3}{2} M^\frac{3}{2} \right) \varepsilon M t^{\delta}.
\end{equation} 
\end{lemma}

\begin{proof}
The elliptic system \eqref{div-curl-k-hat} gives the identity (c.f. Proposition 4.4.1 in \cite{Christodoulou-K-93}), 
\begin{equation}\label{eq-div-curl-id-CK}
\int_{\Sigma_t} |\nabla \hat k|^2 + 3 R_{mn} \hat k^{im} \hat k_{i}^{\, n} - \frac{R}{2} |\hat k|^2 = \int_{\Sigma_t} |H_{ij}|^2 + \frac{1}{2} |\bar R_{\Tbar i}|^2.
\end{equation}
Noticing the formula for $R$ \eqref{Gauss-Ricci-trace-hat -k}, and the expansion \eqref{R-imjn-E} for $R_{imjn}$, we also get
\begin{equation}\label{R-ij-E}
R_{ij} \sim 1 \pm \hat k \pm |\hat k|^2 + E \pm |\D \phi|^2, \quad R  \sim 1 \pm |\hat k|^2 + E \pm |\D \phi|^2.
\end{equation}
Using Cauchy-Schwarz inequality, we  derive
\begin{equation}\label{ineq-nabla-hat-k-source}
\begin{split}
\int_{\Sigma_t}   |\nabla \hat k|^2 \lesssim & \int_{\Sigma_t}   \left(1+ |\D \phi|^2_{L^\infty} + |\hat k|_{L^\infty} + |\hat k|^2_{L^\infty}  +|E|_{L^\infty} \right) |\hat k|^2  \\
&+  \int_{\Sigma_t}  |H|^2 + |\Tbar \phi|^2_{L^\infty} |\nabla \phi|^2.
\end{split}
\end{equation}
That is,
\begin{equation*}
\begin{split}
\int_{\Sigma_t}   |\nabla \hat k|^2 \lesssim &  \left(1+ \mathcal{E}_2(\phi, t) + \mathcal{E}^{\frac{1}{2}}_2(g, t) + \mathcal{E}_2(g, t) \right) E_0(\hat k, t) \\
&+  E_0(W, t)  +\mathcal{E}_2(\phi, t) \mathcal{E}_0(\phi, t).
\end{split}
\end{equation*}
Using the bootstrap assumption \eqref{BT-KG}, the improved $E_0(W,t)$ \eqref{eq-energy-estimate-Bianchi-l-leq-2}, we prove \eqref{estimate-nabla-1-4-hat-k} with $i=0$, namely,
\begin{equation}\label{ineq-nabla-hat-k-source-1}
\begin{split}
\int_{\Sigma_t} t^2 |\nabla \hat k|^2 \lesssim &t^2 E_0(\hat k ,t)+  \varepsilon^2 I^2_{0}(W) +  \varepsilon^3 M^3.
\end{split}
\end{equation}

Proceeding to the higher order cases, we wish to prove that for $0\leq i \leq 3$,
\begin{equation}\label{estimate-nabla2-hat-k-i-induction}
\int_{\Sigma_t} t^2 |\nabla_{I_{i+1}} \hat k|^2 \lesssim t^2 E_0(\hat k, t)+ \left( \varepsilon^2 I^2_{i}(W) + \varepsilon^2 I^2_{4}(\phi) +  \varepsilon^3 M^3 \right) t^{C_M \varepsilon}.
\end{equation}
Now \eqref{ineq-nabla-hat-k-source-1} shows that \eqref{estimate-nabla2-hat-k-i-induction} holds for $i=0.$ Suppose \eqref{estimate-nabla2-hat-k-i-induction} holds for $i\leq l (l \leq 2),$ we need to prove that \eqref{estimate-nabla2-hat-k-i-induction} also holds for $i=l+1 (l \leq 2).$ 
 
In views of the commuting identity between $\nabla_{I_l}$ and $\Delta$ \eqref{commuting-nabla-laplacian-l-Psi-i-ij}-\eqref{Def-R-Psi-ij-commute-nabla-laplacian},  after an integration by parts, we calculate for $0 \leq l,$ 
\begin{align*}
 \int_{\Sigma_t}  |\nabla_{I_{l+2}} \hat k|^2 & = \int_{\Sigma_t} |\nabla_{I_l} \Delta \hat k|^2 \pm \mathcal{R}_{I_l}( \hat k_{ij} ) * \nabla_{I_l} \Delta \hat k  \pm \mathcal{R}_{I_{l+1}} ( \hat k_{ij} ) * \nabla_{I_{l+1}} \hat k.
\end{align*}
Apply the Cauchy-Schwarz inequality, then
\begin{align*}
 \int_{\Sigma_t}  |\nabla_{I_{l+2}} \hat k|^2 & \lesssim \int_{\Sigma_t} |\nabla_{I_l} \Delta \hat k|^2 + \big|\mathcal{R}_{I_l}( \hat k_{ij}) \big|^2 + \big|\mathcal{R}_{I_{l+1}} ( \hat k_{ij} )\nabla_{I_{l+1}} \hat k \big|.
\end{align*}
Meanwhile, 
\begin{align*}
\nabla_{I_l} \Delta \hat k_{ij} &= \nabla_{I_l} \left( \nabla^m \bar R_{mi \Tbar j} - \nabla_i \bar R_{\Tbar j} +R_{i}{}^{\!p} \hat k_{pj} - R_i{}^{\!p}{}_{\!j}{}^{\!q} \hat k _{pq} \right)
\end{align*}
and  $\nabla^m \bar R_{mi \Tbar j} = \epsilon_{mi}{}^{\!p} \nabla^m H_{pj} - \frac{1}{2} \nabla_j \bar R_{\Tbar i} +\frac{1}{2} g_{ij} \nabla^m \bar R_{m \Tbar},$
therefore, 
\begin{align*}
 \int_{\Sigma_t}  |\nabla_{I_{l+2}} \hat k|^2 \lesssim & \int_{\Sigma_t} |\nabla_{I_{l+1}} H|^2  +|\nabla_{I_{l+1}} \left(\bar R_{i \Tbar} \right)|^2 \\
 & + \big|\mathcal{R}_{I_l} (\hat k_{ij} ) \big|^2 + \big|\mathcal{R}_{I_{l+1}} ( \hat k_{ij}) \nabla_{I_{l+1}} \hat k \big|.
\end{align*}
Note that, $\bar R_{\Tbar j} = \nabla_j \phi   \Tbar \phi + \frac{\bar m^2(t)}{2}\phi^2 \bar g_{\Tbar j},$ $ \nabla_{I_l} (\Tbar^\mu \bar g_{\mu i}) =0,$  then $\nabla_{I_l} \bar R_{\Tbar j} = \nabla_{I_l} \left(  \nabla_j \phi  \Tbar \phi \right).$ And in views of the expansion for $R_{imjn}$ \eqref{R-imjn-E},
\begin{align*}
 \int_{\Sigma_t}  |\nabla_{I_{l+2}} \hat k|^2 \lesssim &\int_{\Sigma_t} | \nabla_{I_{l+1}} H|^2 +  \sum_{a+b =l+1} |\nabla_{I_{a+1}} \phi \nabla_{I_b} \Tbar \phi|^2\\
& \quad \quad \quad \quad + RK_1 + RK_2 + RK_3 + RK_4,
\end{align*}
where $RK_1 - RK_4$ are defined as below,
\begin{align*}
RK_1 =& \left( 1+ |\hat k|^2 + |\hat k|^4 +   |E|^2 + |\D\phi|^4 \right) |\nabla_{I_l} \hat k|^2, \\
 RK_2 =& \sum_{a+1+b=l} \big|\nabla_{I_{a+1}} \left( |\hat k| \pm |\hat k|^2 \pm E \pm |\D \phi |^2 \right)\big|^2  |\nabla_{I_b} \hat k|^2, \\
 RK_3 =& \left( 1+ |\hat k| + |\hat k|^2 +  |E| + |\D\phi|^2 \right) |\nabla_{I_{l+1}} \hat k|^2, \\
 RK_4 =& \sum_{a+1+b=l+1} \big|( t\nabla)_{I_{a+1}}  \left( |\hat k| \pm |\hat k|^2 \pm  E \pm  |\D \phi|^2 \right) \big| | \nabla_{I_b} \hat k \nabla_{I_{l+1}} \hat k|.
\end{align*}
Notice that $l \leq 2$.
First of all, 
\begin{equation}\label{esti-hat-k-l-H}
\int_{\Sigma_t} | \nabla_{I_{l+1}} H|^2 \lesssim \mathcal{E}_{l+1}(W, t), \quad l \leq 2. 
\end{equation}
Viewing $[\frac{l+1}{2}] +2 \leq 3 (l \leq 2)$, 
\begin{equation}\label{esti-hat-k-l-bar-R}
\begin{split}
 &\int_{\Sigma_t} \sum_{a+b =l+1} |\nabla_{I_{a+1}} \phi * \nabla_{I_b} \Tbar \phi|^2
 \lesssim \mathcal{E}_3(\phi, t) \mathcal{E}_{l+1}(\phi, t).
\end{split}
\end{equation}

For $RK_1$ and $RK_3,$ we apply $L^\infty, L^2, L^2$ to derive
\begin{equation}\label{esti-hat-k-l-RK13}
 \int_{\Sigma_t} RK_1 + RK_3
 \lesssim   \mathcal{E}_l(\hat k, t) +   \mathcal{E}_{l+1}(\hat k, t).
\end{equation}

For $RK_2,$ since $1 \leq a+1 \leq l \leq 2, b \leq l-1 \leq 1,$ we can always apply $L^4$ on each of the four factors (such as $ |\nabla_{I_{a+1}} E|^2  | \nabla_{I_b}\hat k|^2 $) or $L^6$ on each of the six factors (such as $ |\nabla_{I_{a+1}} (\D \phi)^2|^2   | \nabla_{I_b}\hat k|^2$). Then 
\begin{equation}\label{esti-hat-k-l-RK2}
\int_{\Sigma_t} RK_2 \lesssim \left( \mathcal{E}_{l+1}(g,t) + \mathcal{E}^2_{l+1}(\phi,t)  \right) \mathcal{E}_{l}(\hat k,t). 
\end{equation}

For $RK_4$, the argument is similar to that for $RK_2$. 
\begin{itemize}
\item If $a+1=l+1 \leq 3,$ then $b=0,$ we get $$\big|\nabla_{I_{l+1}}  \left( |\hat k| \pm |\hat k|^2 \pm  E \pm |\D \phi|^2  \right) \big| * |\hat k|  |\nabla_{I_{l+1}} \hat k|.$$ For three factors (such as $\nabla_{I_{l+1}} E * \hat k * \nabla_{I_{l+1}} \hat k$), we apply $L^2, L^\infty, L^2$; For four factors (such as $\nabla_{I_{l+1}} \hat k* \hat k * \nabla_{I_{l+1}} \hat k$), we apply $L^2, L^\infty, L^\infty, L^2$ or $L^\infty, L^2, L^\infty, L^2$, since $[\frac{l+1}{2}] +2 \leq 3(l\leq 2)$. Then we can deduce the bound
\begin{align*}
 \lesssim &\left( \mathcal{E}^{\frac{1}{2}}_{l+1}(W,t) +  \mathcal{E}^{\frac{1}{2}}_{l+1}(\phi,t) \mathcal{E}^{\frac{1}{2}}_{3}(\phi,t) \right) \mathcal{E}^{\frac{1}{2}}_{2}(\hat k,t) E^{\frac{1}{2}}_{l+1}(\hat k,t) \\
&+ \left( \mathcal{E}^{\frac{1}{2}}_{2}(\hat k,t) +  \mathcal{E}_{3}(\hat k,t) \right) \mathcal{E}_{l+1}(\hat k,t).
\end{align*}

\item If $1 \leq a+1 \leq l \leq 2,$ then $1 \leq b \leq l \leq 2,$ we apply $L^4, L^4, L^2$ or $L^6, L^6, L^6, L^2$ to derive the bound
\begin{align*}
& \left( \mathcal{E}^\frac{1}{2}_{l+1}(W,t)+ \mathcal{E}^\frac{1}{2}_{l+1}(\hat k ,t) \right) \mathcal{E}_{l+1}(\hat k,t) \\
& + \left( \mathcal{E}_{l+1}(\phi,t)  + \mathcal{E}_{l+1}(\hat k,t) \right) \mathcal{E}_{l+1}(\hat k,t).
\end{align*}
\end{itemize}
In summary, we have by Cauchy-Schwartz inequality,
\begin{equation}\label{esti-hat-k-l-RK4}
\begin{split}
\int_{\Sigma_t} RK_4 \lesssim & \left( \mathcal{E}^\frac{1}{2}_{l+1}(g,t) + \mathcal{E}_{l+1}(\phi,t) \right) \mathcal{E}_{l+1}(\hat k,t) \\
& +  \mathcal{E}_{l+1}(W,t) + \mathcal{E}_{3}(\phi,t)\mathcal{E}_{l+2}(\phi,t).
 \end{split}
\end{equation}

From \eqref{esti-hat-k-l-H}-\eqref{esti-hat-k-l-RK4}, we have for $ l \leq 2,$
\begin{equation}\label{eq-energy-hat-k-source}
\begin{split}
\int_{\Sigma_t}   |\nabla_{I_{l+2}} \hat k|^2 \lesssim & \mathcal{E}_{l+1}(W,t)+ \mathcal{E}_{l+1}(\hat k,t) + \mathcal{E}^2_4(\phi, t).
 \end{split}
\end{equation}
By induction and combining the improvement for $W, \phi$ with Lemma \ref{lem-hat-k-estimate-original-0}, we proved \eqref{estimate-nabla2-hat-k-i-induction} and hence \eqref{estimate-nabla-1-4-hat-k}.
\end{proof}

We choose $M$ additionally satisfying $C\left( I_0(\hat k) + I_3(W) + I_4(\phi) \right) \leq M/4$ and $\varepsilon$ sufficiently small such that $C \varepsilon^\frac{1}{2} M^\frac{1}{2} \leq 1/4$, then \eqref{eq-improved-hat-k-0} and \eqref{estimate-nabla-1-4-hat-k} imply
\begin{equation}\label{eq-energy-hat-k-improved}
t^2 E_{i}(\hat k,t ) \leq   \frac{1}{2}\varepsilon^2 M^2 t^{\delta}, \quad 0 \leq i \leq 4,
\end{equation}
which improves the bootstrap assumption for $\hat k$ in \eqref{BT-k-N-L-infty}.

\subsubsection{Closing the bootstrap assumption for lapse}
By the Sobolev inequalities, the improved energy for $\hat N$ \eqref{N-improved-estimate-source} (see Corollary \ref{lem-N-ineq-source-improved-2}) implies that 
 \begin{equation}\label{BT-N}
\|\hat N\|_{L^\infty} \leq C  \varepsilon  I_0(\phi)  t^{-1},
\end{equation}
for some universal constant $C$. We choose $C^2I^2_0(\phi) \leq M^2/2$, then
 \begin{equation}\label{ineq-hat-k-improved-infty}
 \|\hat N\|^2_{L^\infty} \leq \frac{\varepsilon^2 M^2}{2} t^{-2 + 4\delta}, \quad 0<\delta < \frac{1}{6},
\end{equation}
 which improves the bootstrap assumption for $\|\hat N\|_{L^\infty}$ in \eqref{BT-k-N-L-infty}.

We are in a position to state the global existence theorem.
\begin{theorem}\label{main-thm-global}
Let $(\Sigma, g_0, k_0, \phi_0, \phi_1)$ be CMC data for the EKG equations \eqref{eq-Einstein-source}-\eqref{def-energy-Mom-kg}, \eqref{eq-kg} (also see the $1+3$ forms \eqref{eq-1+3-bianchi-T-E}-\eqref{eq-1+3-bianchi-T-H} and \eqref{eq-rescale-kg-1+3-0}). Assume that $\Sigma$ is a compact hyperbolic $3$-manifold without boundary and $\tr_{g_0} k_0 = -3$, with initial time $t_0 > \max \{9 m^{-1}, 1\}$. Suppose that the initial data verifies the smallness condition: there is a constant $\varepsilon_0 >0$ such that for all $\varepsilon \leq \varepsilon_0$, 
\begin{equation}\label{initial-smallness}
t_0 \mathcal{E}_4(\phi, t_0) + t_0^2 \mathcal{E}_3(W, t_0) + t_0^2 \mathcal{E}_4(\hat k, t_0) \leq \varepsilon^2.
\end{equation}
Then there exists a unique global solution $(W, \phi)$ for $t \geq t_0$ with the property that for some universal constant $C$ and $0<\delta < \frac{1}{6},$
\begin{equation}\label{stronger-estimate-closing-bt}
t \mathcal{E}_4(\phi, t) + t^2 \mathcal{E}_3(W, t)   \lesssim \varepsilon^2 t^{C \varepsilon}, \quad  t^2 \mathcal{E}_4(\hat k, t) \lesssim \varepsilon^2 t^{\delta}.
\end{equation}
\end{theorem}
Finally, we can follow the argument presented in \cite{A-M-04} (Theorem 6.2) to prove the geodesic completeness.
       
 \begin{remark}\label{re-asympototic-state}
Let us the explore the asymptotic state of the rescaled metric $g$.
 As we can see from the Gauss equations \eqref{Gauss-Riem-hat-k}, the decay estimates for $\hat k, \D \phi, H$ yield that
 \begin{align*}
 & |R_{imjn} + \frac{1}{2}(g \odot g) _{imjn}| \lesssim \varepsilon t^{-4+\delta}.
\end{align*}
This shows that the sectional curvature $K$ of $g$ satisfies $K \rightarrow -1.$
\end{remark}
\begin{remark}\label{re-initial-data}
By the Gauss-Codazzi equation on the initial slice, the smallness condition for the data \eqref{intro-initial-data} given in Theorem \ref{main-thm-into} will enforce \eqref{initial-smallness}.
\end{remark}

\appendix

\section{Local in time development of CMC data}\label{sec-local-cmc}
In this section, we stick to the original spacetime $\breve g_{\mu \nu}$ and spatial metric $\tilde g_{ij}$, which are denoted by $\bar g_{\mu \nu}, g_{ij}$ (without tilde) as well. These notations (without tilde) applies to $\nabla, k_{ij}, R_{ij}, \bar R_{ij}, \cdots$ as well.

The EKG system \eqref{eq-Einstein-source}-\eqref{def-energy-Mom-kg} has the structure equations \eqref{eq-evolution-1}-\eqref{Codazzi-div-k} taking the following forms: The evolution equations
\begin{subequations}
\begin{equation}
\mathcal{L}_{\dt} g_{ij} = -2N k_{ij}, \quad \quad \quad \quad \quad \quad \quad \quad \quad \quad \quad \quad \quad \quad \quad \quad \quad \label{eq-evolution-1-1} 
\end{equation}
\begin{equation}
\lie_{\dt} k_{ij} = -\nabla_i \nabla_j N + N \left( R_{ij} - 2 k_{il} k_{j}^l + \tr k k_{ij} -  \bar{R}_{ij}(\phi)\right). \label{eq-evolution-2-1}
\end{equation}
\end{subequations}
\noindent The KG equation 
\begin{equation}\label{eq-kg-1}
\Box_{\bar g} \phi - m \phi =0.
\end{equation}
\noindent The equation for the lapse $N$
\begin{equation}\label{eq-lapse-1}
\Delta N + \dt \tr k - N \left( \bar R_{TT} (\phi) + |k|^2 \right)=0.
\end{equation}
\noindent The constraint equations
\begin{subequations}
\begin{equation}
\begin{split}
\tr k + 3t^{-1} &=0, \,\,\, \label{Constrain-CMC-gauge} 
\end{split}
\end{equation}
\begin{equation}
\begin{split}
\quad\quad R -|k|^2 + (\tr k)^2  &= 2\bar{R}_{TT}(\phi) + \bar{R}(\phi), \label{Constrain-Guass} 
\end{split}
\end{equation}
\begin{equation}
\begin{split}
\nabla^i k_{ij} - \nabla_j\tr k &= - \bar R_{Tj}(\phi). \label{Constrain-Codazzi}
\end{split}
\end{equation}
\end{subequations}
In \eqref{eq-evolution-1-1}-\eqref{Constrain-Codazzi}, the coupling $\bar R_{\alpha \beta}(\phi)$ is indeed the Ricci tensor $\bar R_{\alpha \beta}$ (associated to $\bar g$), which, via the EKG equation \eqref{eq-Einstein-source}-\eqref{def-energy-Mom-kg}, is related to the KG field by
\begin{equation}\label{eq-ricci-kg-phi}
\bar R_{\alpha \beta}(\phi) = D_\alpha \phi D_\beta \phi + \frac{m}{2}\phi^2 \bar g_{\alpha \beta}.
\end{equation}

As in \cite{Christodoulou-K-93} (P. 307), we introduce 
\begin{subequations}
\begin{equation}\label{conserved-gague-A}
A = \tr k +  3t^{-1}, \quad \quad \quad\quad \quad \quad \quad\quad \quad \quad \quad\quad \quad \quad \quad\quad \quad \quad \quad \quad 
\end{equation}
\begin{equation}\label{conserved-gague-B}
B = R -|k|^2 + (\tr k)^2  - 2\bar{R}_{TT}(\phi) - \bar{R}(\phi),  \quad \quad \quad\quad \quad  \quad  \quad \quad\,\,\,\,\,
\end{equation}
\begin{equation}\label{conserved-gague-C}
C_i = \nabla^j k_{ij} - \nabla_i \tr k + \bar R_{Ti}(\phi) + 1/2 \nabla_i A,  \quad \quad \quad\quad \quad\quad\quad \quad \quad \quad
\end{equation}
\begin{equation}\label{conserved-gague-D}
D_{ij}  = N^{-1} \lie_{\dt} k_{ij} + N^{-1} \nabla_i \nabla_j N - R_{ij} + 2 k_{ip} k_{j}^p - \tr k k_{ij} +  \bar{R}_{ij}(\phi).
\end{equation}
\end{subequations}
The definition of $A,B,C,D$ are almost identical (except the couplings) to the ones introduced in \cite{Christodoulou-K-93} and \cite{A-M-03-local}.

The EKG system \eqref{eq-evolution-1-1}-\eqref{eq-evolution-2-1},  \eqref{eq-kg-1}, \eqref{eq-lapse-1} and the constraint equations \eqref{Constrain-CMC-gauge}-\eqref{Constrain-Codazzi} imply $A=B=0, C_i=0, D_{ij} = 0$ and the following second-order hyperbolic equations for $k_{ij}$ and $\phi$:
\begin{equation}\label{eq-box-k-ij-kg-phi} 
\begin{split}
 - \left( N^{-1} \lie_{\dt} \right)^2 k _{ij} + \Delta k_{ij} &= F_{ij}, \\
 \Box_{\bar g} \phi -m \phi &=0, \\
\end{split}
\end{equation}
and
\begin{equation}\label{eq-svari-lapse} 
\begin{split}
-N^{-1} \lie_{\dt} g_{ij} &= 2 k_{ij}, \\
 \Delta N +\dt \tr k - N \left( \bar R_{TT}(\phi) + |k|^2 \right) &=0,
\end{split}
\end{equation}
and $N^\prime= \dtau N$,
\begin{equation}\label{eq-dt-N}
\begin{split}
& \Delta N^\prime - 3t^{-2}N^\prime
 =- 6t^{-3} \hat N   + N^\prime \left(\bar R_{TT}(\phi) + |\hat k|^2 \right) + N \dt\left(\bar R_{TT}(\phi) + |\hat k|^2 \right) \\
 & -t^{-1} |\nabla N |^2  -2 \hat k^{ij} \left( N\nabla_i\nabla_j N+ \nabla_i N \nabla_j N \right) +2t^{-1} N \Delta N + 2N \nabla^i N \bar R_{Ti} (\phi).
 \end{split}
\end{equation}
The inhomogeneous term $F_{ij}$ is 
\begin{equation}\label{def-F-ij-original} 
\begin{split}
 F_{ij} =&  N_{ij} - \nabla_i \bar R_{Tj}(\phi) -  \nabla_j  \bar R_{Ti} (\phi) \\
 & \quad + N^{-1} \lie_{\dt} \left( - \tr k k_{ij} + \bar R_{ij} (\phi) \right),
\end{split}
\end{equation}
where $N_{ij}$ and $H_{ij}$ are the same as that in \cite{Christodoulou-K-93} (Page 308):
\begin{equation}\label{def-N-L-ij} 
\begin{split}
N_{ij} =&L_{ij} - H_{ij},  \\
 N^2 L_{ij} = &\nabla_i \nabla_j N^\prime - N^{-1} N^\prime \nabla_i \nabla_j N - \lie_{\dt} \Gamma^p_{ij} \nabla_p N \\
 &+ 2 N \left( k_i{}^{\! p} \lie_{\dt} k_{jp} +k_j{}^{\! p} \lie_{\dt} k_{ip}  \right) + 4N^2 k^{pq} k_{ip} k_{jq},
\end{split}
\end{equation} 
and
\begin{equation}\label{def-H-I-ij} 
\begin{split}
N H_{ij} =&N I_{ij} +  \nabla^p N \left( 2 \nabla_p k_{ij} - \nabla_i k_{pj} - \nabla_j k _{p i} \right) \\
 &- \nabla_i N \nabla^p k_{pj} - \nabla_j N \nabla^p k_{pi} + \nabla_i N \nabla_j \tr k + \nabla_j N \nabla_i \tr k \\
 &- k_i{}^{\!p} \nabla_j \nabla_p N -   k_j{}^{\!p} \nabla_i \nabla_p N + \tr k \nabla_i \nabla_j N + k_{ij } \Delta N, \\
 I_{ij} =& -3 \left( k_i{}^{\!p} R_{jp} +  k_j{}^{\!p} R_{ip} \right) + 2 g_{ij} k^{pq}R_{pq} \\
 &+ 2 \tr k R_{ij} + \left(k_{ij} - g_{ij} \tr k \right) R.
\end{split}
\end{equation}

By a computation analogous to \cite{Christodoulou-K-93}, using the n+1-decomposition of the divergence of the energy-momentum tensor, $D^\mu T_{\mu \nu} = 0$, as presented in (\cite{Rendall-pde-gr}, P. 17), one derives the following system of evolution equation for $A, B ,C_i, D_{ij}$ which holds for  any given solution of the system \eqref{eq-box-k-ij-kg-phi}-\eqref{eq-svari-lapse}:
\begin{subequations}
\begin{equation}
\begin{split}
N^{-1} \dt A &= F:= \tr D + B, \quad \quad \quad\quad \quad \quad \quad\quad \quad \quad \quad \label{eq-time-derivative-A} 
\end{split}
\end{equation}
\begin{equation}
\begin{split}
N^{-1} \dt F  &=  \Delta A + 2N^{-1} \nabla^i N \nabla_i A - 4 N^{-1} \nabla^i N C_i, \quad   \label{eq-time-derivative-F-1} 
\end{split}
\end{equation}
\begin{equation}
\begin{split}
 \quad \,\, N^{-1} \dt C_i &= \left( \nabla^j D_{ji} - \frac{1}{2} \nabla_i \tr D\right)  + N^{-1} \nabla^j N D_{ij} \quad \quad \quad \\
& \quad \quad +  \tr k C_i - \frac{1}{2} \tr k \nabla_i A  -\frac{1}{2} N^{-1} \nabla_i N  F,  \label{eq-time-derivative-C} 
\end{split}
\end{equation}
\begin{equation}
\begin{split}
N^{-1} \dt D_{ij} &=   \nabla_i C_j + \nabla_j C_i. \quad \quad \quad\quad \quad \quad \quad  \quad \quad \quad\quad \quad  \label{eq-time-derivative-D}
\end{split}
\end{equation}
\end{subequations}
Generically, \eqref{eq-time-derivative-A}-\eqref{eq-time-derivative-D} always holds true and the proof is irrelevant to the matter field.

Therefore, as in \cite{Christodoulou-K-93}, one has
\begin{lemma}\label{lem-gauge-conserved}
Any solution of the coupled elliptic-hyperbolic system \eqref{eq-box-k-ij-kg-phi}-\eqref{eq-svari-lapse} whose initial data $g_{ij}, k_{ij}, \lie_{\dt} k_{ij}, \phi, T\phi$ verify the original constraint equations $A=0, B=0, C_i =0$ as well as $D_{ij}=0$ is a solution of the original systems  \eqref{eq-evolution-1-1}-\eqref{eq-evolution-2-1},  \eqref{eq-kg-1}, \eqref{eq-lapse-1}.
\end{lemma}

An existence result for the reduced system can be done roughly in the same way as in the proof of the local existence for the reduced wavelike system \cite{Bruhat} (c.f. \cite{Christodoulou-K-93}). The only differences lie in that $g_{ij}$ is tied to $k_{ij}$ and $N$ through the integral equation in \eqref{eq-svari-lapse}.  
One can proceed precisely as in \cite{Christodoulou-K-93}, estimating $k_{ij}, \phi$ by energy estimates, $g_{ij}$ by the integral equation in \eqref{eq-svari-lapse}  and $N, N^\prime$ by elliptic estimates. We also remark that in the elliptic equations for $N-1, N^\prime$, both of the $\p_t \tr k$ and  $- 3t^{-2} N^\prime$ admit good sign. They will contribute to the estimates for $\|t^{-2} ( N-1) \|_{L^2(\Sigma_t)}$ and $\|t^{-2} N^\prime \|_{L^2(\Sigma_t)}$. This can be seen in Section \ref{sec-N}.  Besides, we are working within the setting of compact manifold without boundary, hence we do not have to care about the asymptotic decay at spatially infinity as in \cite{Christodoulou-K-93}, and the regularity we are considering is one order high than that in \cite{Christodoulou-K-93}. In general, the local existence theorem in our setting is easier than that in \cite{Christodoulou-K-93}. Finally, we make the rescaling and Theorem \ref{thm-local-existence} follows.

\section{ODE estimate}\label{sec-ode}
The following Gr\"{o}nwall inequality is used throughout this paper.
\begin{lemma}[Differential form]\label{lemma-Gronwall-diff}
Let $a(s), b(s), \psi(s) $ are non-negative functions, and
\begin{equation*}
\psi^\prime(s) \leq b(s) + a(s) \psi(s),
\end{equation*}
then
\begin{equation}\label{eq-Gronwall-diff}
|\psi(s)| \leq  \exp\left( \int_{s_0}^{s} a(\tau) \di \tau\right) \left(\psi(s_0) + \int_{s_0}^s  b(\tau) \di \tau \right).
\end{equation}
\end{lemma}

The technical ODE estimate (c.f. Lemma 3.5 in \cite{Ma-Lefoch-16}) specializes to the situation considered as follows,
\begin{lemma}\label{lemma-ode}
Let $$\frac{\partial^2}{\partial ^2 s} \psi + c^2 \psi = F(s),$$ with the initial data $$\psi\big|_{s= s_0}= \psi_0, \quad \partial_s \psi \big|_{s=s_0} = \psi_1.$$ There is the inequality
\begin{equation}\label{eq-lemma-ode}
|\psi(s)| + |\partial_s \psi(s)| \lesssim |\psi_0| + |\psi_1| + \int_{s_0}^{s} |F(\tau)| \di \tau.
\end{equation}
\end{lemma}

\section{Basic geometric conventions and identities}\label{sec-app-id}
{\bf Curvature}. We clarify the conventions for curvature. The Riemann tensor of  $(\bar M, \bar g)$ is defined by
\begin{align*}
\bar R(X, Y) Z = D_{X} D_{Y} Z - D_{Y} D_{X} Z - D_{[X, Y]} Z
\end{align*}
for vector fields $X, Y, Z$. We define in a coordinate system
\begin{align*}
 \bar R_{\alpha \beta \mu \nu} = \bar g(e_\mu, \bar R(e_\alpha, e_\beta) e_\nu).
\end{align*}
The Riemann tensor satisfies the Bianchi identities
\begin{align*}
D_{[\sigma} \bar R_{\alpha \beta ]\mu \nu}= \frac{1}{3} \left( D_{\sigma} \bar R_{\alpha \beta \mu \nu} + D_{\alpha} \bar R_{\beta \sigma \mu \nu} + D_{\beta} \bar R_{\sigma \alpha \mu \nu} \right) =0.
\end{align*}

{\bf Volume element}. Let $\epsilon_{\mu \alpha \nu \beta}$ be the coefficient of the volume element of an $1+3$ dimensional Lorentz manifold $(\bar M, \bar g)$.  Then $\epsilon_{\mu \alpha \nu \beta}$ is a totally anti-symmetric tensor, and satisfies the identities (c.f. Page 136 in \cite{Christodoulou-K-93})
\begin{align*}
\epsilon^{\alpha_1 \alpha_2 \alpha_3 \alpha_4}\epsilon_{\alpha_1 \beta_2 \beta_3 \beta_4} &= -\det(\delta^{\alpha_i}_{\beta_j})_{i,j=2,3,4} \\
\epsilon^{\alpha_1 \alpha_2 \alpha_3 \alpha_4}\epsilon_{\alpha_1 \alpha_2 \beta_3 \beta_4} &= -2\det(\delta^{\alpha_i}_{\beta_j})_{i,j=3,4} \\
\epsilon^{\alpha_1 \alpha_2 \alpha_3 \alpha_4}\epsilon_{\alpha_1 \alpha_2 \alpha_3 \beta_4} &= -6\delta^{\alpha_4}_{\beta_4} \\
\epsilon^{\alpha_1 \alpha_2 \alpha_3 \alpha_4}\epsilon_{\alpha_1 \alpha_2 \alpha_3 \alpha_4} &= -24.
\end{align*}
In an O.N. frame adapted to a spacelike hypersurface $\Sigma$ in $(\bar M, \bar g)$,  the coefficient of the induced volume element is (c.f. Page 144 in \cite{Christodoulou-K-93})
\begin{align*}
\epsilon_{ijk}=\epsilon_{\Tbar ijk}.
\end{align*}

{\bf Identities}.
Let $(\Sigma, g_{ij})$ be the spatially $3$-dimensional Riemann manifold, and $\nabla$ is the covariant derivative with respect to $g_{ij}$. Let $\Gamma^a_{ij}$ be the connection coefficient associated to $\nabla$, the time derivative is a tensorfield:
\begin{equation}\label{dt-Gamma}
\begin{split}
\lie_{\dt} \Gamma^a_{ij}& =\frac{1}{2}g^{ab} \left( \nabla_i \lie_{\dt} g_{jb} + \nabla_j \lie_{\dt} g_{ib} -\nabla_b \lie_{\dt} g_{ij} \right).
\end{split}
\end{equation}
Then for any $(0, 2)$ tensor on $\Sigma$
\begin{equation}\label{commuting-lie-nabla}
\begin{split}
\lie_{\dt} \nabla_i A_{ab} =& \nabla_i \lie_{\dt} A_{ab} - \lie_{\dt} \Gamma^p_{ia} A_{pb} -  \lie_{\dt} \Gamma^p_{ib} A_{ap}.
\end{split}
\end{equation}

Define the following operations on symmetric $2$-tensors on $\Sigma,$
\begin{align}
A \cdot B &=A_{ij} B^{ij}, \label{def-dot} \\
(\dive  A)_i &= \nabla^j A_{ij}, \label{def-div} \\
\curl A_{ij} &= \frac{1}{2} \left( \epsilon_i{}^{\! pq}\nabla_q A_{pj} +\epsilon_j{}^{\! pq} \nabla_q A_{pi} \right). \label{def-curl} \\
(A \wedge B)_i & = \epsilon_i{}^{\! jp} A_j{}^{\! q} B_{qp}, \label{def-wedge} \\
(v \wedge A)_{ij} & = \epsilon_i{}^{\! pq} v_p A_{qj}  +  \epsilon_j{}^{\! pq} v_p A_{qi} , \label{def-wedge-1} \\
(A \times B)_{ij} & = \epsilon_i{}^{\! ab} \epsilon_j{}^{\! pq}  A_{ap} B_{bq} + \frac{1}{3} A \cdot B g_{ij} - \frac{1}{3} \tr A \tr B g_{ij}. \label{def-times} 
\end{align}
The operation $\wedge$ is skew symmetric, while $\times$ is symmetric.

We get some identities between $\dive$ and $\curl$. A computation shows that
\begin{equation}\label{wedge-curl-div}
\dive (A \wedge B) = - \curl A \cdot B + A \cdot \curl B.
\end{equation}
The $l^{\text{th}} (l \geq 1) $ order analogous of \eqref{wedge-curl-div} is given below (or Lemma \ref{lemma-div-curl}).
\begin{equation}\label{wedge-curl-div-general-app}
\begin{split}
& \nabla_{I_l} \left( \curl A \right)_{ij}  \nabla^{I_l} B^{ij} - \nabla_{I_l} \left( \curl B \right)_{ij} \nabla^{I_l} A^{ij} \\
=& \sum_{a \leq l}  \nabla_q    \left( \nabla_{I_a} A  * \nabla_{I_a} B  \right)^q \pm \hat{\mathcal{R}}_{I_{l-1}}(A) * \nabla_{I_l} B + \hat{\mathcal{R}}_{I_{l-1}}(B) * \nabla_{I_{l}} A .
\end{split}
\end{equation}
where $\hat{\mathcal{R}}_{I_l}(A)$ is defined as in \eqref{Def-R-Psi-ij-commute-nabla-laplacian}.

\begin{proof}[Proof for Lemma \ref{lemma-div-curl}]
As far as this paper is concerned, we can restrict our attention to the $l \leq 3$ case for simplicity. Denote
\begin{equation}\label{wedge-curl-div-general-app}
\begin{split}
C_{I_l} =& \nabla_{I_l} \left( \curl H \right)_{ij}  \nabla^{I_l} E^{ij} - \nabla_{I_l} \left( \curl E \right)_{ij} \nabla^{I_l} H^{ij}.
\end{split}
\end{equation}

We first note that, for any symmetric tensor $H_{ij}, E^{ij},$ $\nabla^{I_l} H_{ij}$ and $\nabla^{I_l} E^{ij}$ are still symmetric in $i,j$. Therefore,
\begin{align*}
C_{I_l} = &  \epsilon_i{}^{\! pq} \nabla_{I_l} \nabla_q H_{pj} \nabla^{I_l} E^{ij}  - \epsilon_i{}^{\! pq} \nabla_{I_l} \nabla_q E_{pj} \nabla^{I_l} H^{ij}.
\end{align*}
To prove the Lemma, it suffice to keep track of the principle part $R_{imjn}$ (without derivatives). Again, the principle part of $R_{imjn}$ is $-\left(g_{ij} g_{mn} -g_{in} g_{mj} \right).$ In what follows, we introduce the notation $\simeq$, which means equalling in the principle part. For example, $R_{imjn}\simeq -\left(g_{ij} g_{mn} -g_{in} g_{mj} \right).$

{\bf Case I}: $ \epsilon_i{}^{\! pq} \nabla_q \nabla_{I_l} H_{pj} \cdot \nabla^{I_l} E^{ij}  - \epsilon_i{}^{\! pq}  \nabla_q \nabla_{I_l} E_{pj} \cdot \nabla^{I_l} H^{ij}.$
\begin{align*}
&\epsilon_i{}^{\! pq}   \left(  \nabla_q \nabla_{i_1} \cdots \nabla_{i_l} H_{pj} \cdot\nabla^{i_1} \cdots \nabla^{i_l} E^{ij} \right)  \\
& \quad \quad -  \epsilon_i{}^{\! pq} \nabla_q \nabla_{i_1} \cdots \nabla_{i_l}  E_{pj} \cdot \nabla^{i_1} \cdots \nabla^{i_l} H^{ij}   \\
=&  \nabla_q \left( \epsilon_i{}^{\! pq}  \nabla_{i_1} \cdots \nabla_{i_l} H_{pj} \cdot\nabla^{i_1} \cdots \nabla^{i_l} E^{ij} \right)  \\
&\quad \quad + \epsilon^{piq} \nabla_q \nabla^{i_1} \cdots \nabla^{i_l} E_{ij} \cdot \nabla_{i_1} \cdots \nabla_{i_l} H_p^{j} \\
& \quad \quad -  \epsilon^{ipq} \nabla_q \nabla_{i_1} \cdots \nabla_{i_l}  E_{pj} \cdot \nabla^{i_1} \cdots \nabla^{i_l} H_i^{j}   \\
=&   \nabla_q \left( \epsilon_i{}^{\! pq}  \nabla_{i_1} \cdots \nabla_{i_l} H_{pj} \cdot\nabla^{i_1} \cdots \nabla^{i_l} E^{ij} \right),
\end{align*}
where we had exchange the indices $p$ and $i$ in the fourth line. 
Note that, on the right hand side of the first identity, $\epsilon^{ipq}$ is antisymmetric in $i$ and $p$, while $\nabla_{I_l} H_{pj} \nabla_q \nabla^{I_l} E_i^{j} +\nabla_q \nabla_{I_l}  E_{pj} \nabla^{I_l} H_i^{j}$ is symmetric in $i$ and $p$. The product vanishes.

{\bf Case II}: $ \epsilon_i{}^{\! pq} \nabla_{i_1} \nabla_q \nabla_{I_{l-1}} H_{pj} \cdot \nabla^{I_l} E^{ij}  - \epsilon_i{}^{\! pq} \nabla_{i_1} \nabla_q \nabla_{I_{l-1}} E_{pj} \cdot \nabla^{I_l} H^{ij}.$ In particular, this covers $C_{I_1}$.
 Combine with case I, case II reduces to $ \epsilon_i{}^{\! pq} [\nabla_{i_1}, \nabla_q] \nabla_{I_{l-1}} H_{pj} \cdot \nabla^{I_l} E^{ij}  - \epsilon_i{}^{\! pq} [\nabla_{i_1}, \nabla_q] \nabla_{I_{l-1}} E_{pj} \cdot \nabla^{I_l} H^{ij}$ as we can see below:
 \begin{align*}
=&  \epsilon_i{}^{\! pq}   \nabla_{i_1} \nabla_q \nabla_{I_{l-1}}  H_{pj} \cdot\nabla^{i_1} \cdots \nabla^{i_l} E^{ij} \\
& \quad \quad - \epsilon_i{}^{\! pq}  \nabla_{i_1} \nabla_q \nabla_{I_{l-1}}  E_{pj} \cdot \nabla^{i_1} \cdots \nabla^{i_l} H^{ij}  \\
=&\epsilon_i{}^{\! pq}  \nabla_q \nabla_{I_l}  H_{pj} \cdot\nabla^{I_l}  E^{ij}  -  \epsilon_i{}^{\! pq} \nabla_q \nabla_{I_l} E_{pj} \cdot \nabla^{I_l}  H^{ij}  \\
&+  \epsilon_i{}^{\! pq} [\nabla_{i_1}, \nabla_q] \nabla_{I_{l-1}} H_{pj} \cdot \nabla^{I_l} E^{ij}  - \epsilon_i{}^{\! pq} [\nabla_{i_1}, \nabla_q] \nabla_{I_{l-1}} E_{pj} \cdot \nabla^{I_l} H^{ij}.
\end{align*}
where the first line on the right hand side of the last equality is exactly case I. The last two commutators equal to
\begin{align*}
& \epsilon_i{}^{\! pq} [\nabla_{i_1}, \nabla_q] \nabla_{I_{l-1}} H_{pj} \cdot \nabla^{I_l} E^{ij}  - \epsilon_i{}^{\! pq} [\nabla_{i_1}, \nabla_q] \nabla_{I_{l-1}} E_{pj} \cdot \nabla^{I_l} H^{ij}\\
\simeq & \epsilon_i{}^{\! pq} \left( R_{i_1 q i_i}{}^{\! \lambda} \nabla_{i_2} \cdots \stackrel{i_i}{\nabla}_\lambda \cdots \nabla_{i_l} H_{pj} + R_{i_1 q p}{}^{\! \lambda} \nabla_{I_{l-1}} H_{\lambda j}  + R_{i_1 q j}{}^{\! \lambda} \nabla_{I_{l-1}} H_{p \lambda}  \right) \nabla^{I_l} E^{ij} \\
&- \epsilon_i{}^{\! pq} \left( R_{i_1 q i_i}{}^{\! \lambda} \nabla_{i_2} \cdots \stackrel{i_i}{\nabla}_\lambda \cdots \nabla_{i_l} E_{pj} + R_{i_1 q p}{}^{\! \lambda} \nabla_{I_{l-1}} E_{\lambda j}  + R_{i_1 q j}{}^{\! \lambda} \nabla_{I_{l-1}} E_{p \lambda}  \right) \nabla^{I_l} H^{ij}.
\end{align*}
Substitute the principle part of $R_{imjn}$, it equals to
\begin{align*}
=&-\epsilon_i{}^{\! pq} \left( g_{i_1 i_i} g_q^\lambda - g_{i_1}^\lambda g_{q i_i} \right) \nabla_{i_2} \cdots \stackrel{i_i}{\nabla}_\lambda \cdots \nabla_{i_l} H_{pj}   \nabla^{I_l} E^{ij} - \text{dual parts} \\
& - \epsilon_i{}^{\! pq} \left( g_{i_1 p} g_q^\lambda - g_{i_1}^\lambda g_{pq} \right) \nabla_{I_{l-1}} H_{\lambda j}  \nabla^{I_l} E^{ij} - \text{dual parts} \\
& - \epsilon_i{}^{\! pq} \left( g_{i_1 j }g_q^\lambda - g_{i_1}^\lambda g_{qj} \right) \nabla_{I_{l-1}} H_{p \lambda}  \nabla^{I_l} E^{ij} - \text{dual parts}.
\end{align*}
In views of the symmetry, we have that it equals to
\begin{equation}\label{eq-Case-2-principle}
\begin{split}
=& \sum_{i_i=i_2}^{i_l} -\epsilon_i{}^{\! p \lambda} g_{i_1 i_i} \nabla_{i_2} \cdots \stackrel{i_i}{\nabla}_\lambda \cdots \nabla_{i_l} H_{pj} \nabla^{I_l} E^{ij} \\
& + \epsilon_i{}^{\! p}{}_{\! i_i} \nabla_{i_2} \cdots \stackrel{i_i}{\nabla}_{i_1} \cdots \nabla_{i_l} H_{pj}   \nabla^{I_l} E^{ij}  \\
&+ \epsilon_i{}^{\! p \lambda} g_{i_1 i_i} \nabla_{i_2} \cdots \stackrel{i_i}{\nabla}_\lambda \cdots \nabla_{i_l} E_{pj}  \nabla^{I_l} H^{ij} \\
&  - \epsilon_i{}^{\! p}{}_{\! i_i} \nabla_{i_2} \cdots \stackrel{i_i}{\nabla}_{i_1} \cdots \nabla_{i_l} E_{pj}   \nabla^{I_l} H^{ij}  \\
&- \epsilon_{i i_1}{}^{\! \lambda} \nabla_{I_{l-1}} H_{\lambda j}  \nabla^{I_l} E^{ij} + \epsilon_{i i_1}{}^{\! \lambda}   \nabla_{I_{l-1}} E_{\lambda j} \nabla^{I_l} H^{ij},
\end{split}
\end{equation}
which are denoted by $S_i, i =1 \cdots 6.$
$S_1 + S_4$ equals to 
\begin{align*}
=&- \nabla_{i_i} \left( \epsilon_i{}^{\! p \lambda}   \nabla_{i_2} \cdots \stackrel{i_i}{\nabla}_\lambda \cdots \nabla_{i_l} H_{pj}  \nabla^{i_2} \cdots \nabla^{i_l} E^{ij}  \right)  \\
& + \epsilon_{ip}{}^{\! \lambda}    \nabla_{i_i}  \nabla_{i_2} \cdots \stackrel{i_i}{\nabla}_\lambda \cdots \nabla_{i_l} H^p_{j}  \nabla^{i_2} \cdots \nabla^{i_l} E^{ij}   \\
&  - \epsilon_{ip \lambda}  \nabla_{i_2} \cdots \stackrel{i_i}{\nabla}_{q} \cdots \nabla_{i_l} E^p_{j}  \nabla^q  \nabla^{i_2} \cdots \stackrel{i_i}{\nabla^{\lambda}} \cdots \nabla^{i_l} H^{ij}\\
=& div  + \epsilon_{ip}{}^{\! \lambda}    \nabla_{i_i}  \nabla_{i_2} \cdots \stackrel{i_i}{\nabla}_\lambda \cdots \nabla_{i_l} H^p_{j}  \nabla^{i_2} \cdots \nabla^{i_l} E^{ij}   \\
&  - \epsilon_{pi \lambda}  \nabla_{i_2} \cdots \stackrel{i_i}{\nabla}_{i_i} \cdots \nabla_{i_l} E^i_{j}  \nabla^{i_i}  \nabla^{i_2} \cdots \stackrel{i_i}{\nabla^{\lambda}} \cdots \nabla^{i_l} H^{pj}.
\end{align*}
By the reason of symmetry, the last two term cancelled. We have 
\begin{align*}
S_1 + S_4 =&- \nabla_{i_i} \left( \epsilon_i{}^{\! p \lambda}   \nabla_{i_2} \cdots \stackrel{i_i}{\nabla}_\lambda \cdots \nabla_{i_l} H_{pj}  \nabla^{i_2} \cdots \nabla^{i_l} E^{ij}  \right).
\end{align*}
Look at the $S_2 + S_3$, we change the index $i_i$ to be $\lambda$
\begin{align*}
=&\epsilon_i{}^{\! p \lambda} \nabla_{i_2} \cdots \stackrel{i_i}{\nabla}_{i_1} \cdots \nabla_{i_l} H_{pj}   \nabla^{i_1} \cdots \stackrel{i_i}{\nabla^\lambda} \cdots \nabla^{i_l} E^{ij}  \\
&+ \epsilon_i{}^{\! p \lambda} \nabla_{i_2} \cdots \stackrel{i_i}{\nabla}_\lambda \cdots \nabla_{i_l} E_{pj}  \nabla_{q}   \nabla^{i_2} \cdots \stackrel{i_i}{\nabla^q} \cdots \nabla^{i_l} H^{ij} \\
 &=\nabla^q \left( \epsilon_i{}^{\! p \lambda} \nabla_{i_2} \cdots \stackrel{i_i}{\nabla}_{q} \cdots \nabla_{i_l} H_{pj}  \nabla^{i_2} \cdots \stackrel{i_i}{\nabla^\lambda} \cdots \nabla_{i_l} E^{ij} \right) \\
 &-  \epsilon_i{}^{\! p \lambda} \nabla^q  \nabla_{i_2} \cdots \stackrel{i_i}{\nabla}_{q} \cdots \nabla_{i_l} H_{pj}  \nabla^{i_2} \cdots \stackrel{i_i}{\nabla^\lambda} \cdots \nabla_{i_l} E^{ij}  \\
 & + \epsilon_i{}^{\! p \lambda} \nabla_{i_2} \cdots \stackrel{i_i}{\nabla}_{\lambda} \cdots \nabla_{i_l} E_{pj}   \nabla_{q} \nabla^{i_2} \cdots \stackrel{i_i}{\nabla^q} \cdots \nabla^{i_l} H^{ij}.
\end{align*}
Again, by the symmetric reason, the last two terms cancelled. Thus,
\begin{align*}
S_2 + S_3=&\epsilon_i{}^{\! p \lambda} \nabla_{i_2} \cdots \stackrel{i_i}{\nabla}_{i_1} \cdots \nabla_{i_l} H_{pj}   \nabla^{i_1} \cdots \stackrel{i_i}{\nabla^\lambda} \cdots \nabla^{i_l} E^{ij}.
\end{align*}
Finally the last two terms $S_5+S_6$,
\begin{align*}
=&- \epsilon_{i q}{}^{\! \lambda} \nabla_{I_{l-1}} H_{\lambda j} \nabla^q \nabla^{I_{l-1}} E^{ij} + \epsilon_{i q}{}^{\! \lambda}  \nabla_{I_{l-1}} E_{\lambda j}   \nabla^q \nabla^{I_{l-1}} H^{ij},
\end{align*}
is exactly case I.
As a summary, we prove in case II:
\begin{align*}
&  \epsilon_i{}^{\! pq}   \nabla_{i_1} \nabla_q \nabla_{I_{l-1}}  H_{pj} \cdot\nabla^{I_l}  E^{ij} - \epsilon_i{}^{\! pq}  \nabla_{i_1} \nabla_q \nabla_{I_{l-1}}  E_{pj} \cdot \nabla^{I_l} H^{ij}  \\
\simeq &\nabla_i \left( \nabla_{I_l} H * \nabla_{I_l} E \right)^i+ \nabla_i \left( \nabla_{I_{l-1}} H * \nabla_{I_{l-1}} E \right)^i.
\end{align*}

{\bf Case III}: $ \epsilon_i{}^{\! pq} \nabla_{i_1} \nabla_{i_2} \nabla_q \nabla_{I_{l-2}} H_{pj} \cdot \nabla^{I_l} E^{ij}  - \epsilon_i{}^{\! pq} \nabla_{i_1} \nabla_{i_2} \nabla_q \nabla_{I_{l-2}} E_{pj} \cdot \nabla^{I_l} H^{ij}.$ Combined with case II, it reduces to consider $ \epsilon_i{}^{\! pq} \nabla_{i_0} [ \nabla_{i_1},  \nabla_q] \nabla_{I_{l-1}} H_{pj} \cdot \nabla^{i_0} \nabla^{I_l} E^{ij}  - \epsilon_i{}^{\! pq }\nabla_{i_0} [ \nabla_{i_1}, \nabla_q] \nabla_{I_{l-1}} E_{pj} \cdot \nabla^{i_0} \nabla^{I_l} H^{ij}.$ 
As calculated in case II \eqref{eq-Case-2-principle}, it has the principle part,
\begin{align*}
\simeq & - \epsilon_i{}^{\! p \lambda}  g_{i_1 i_i} \nabla_{i_0} \nabla_{i_2} \cdots \stackrel{i_i}{\nabla}_\lambda \cdots \nabla_{i_l} H_{pj}  \nabla^{i_0} \nabla^{I_l} E^{ij}\\
& + \epsilon_i{}^{\! p}{}_{\! i_i} \nabla_{i_0} \nabla_{i_2} \cdots \stackrel{i_i}{\nabla}_{i_1} \cdots \nabla_{i_l} H_{pj} \nabla^{i_0}\nabla^{I_l} E^{ij}  \\
& + \epsilon_i{}^{\! p \lambda} g_{i_1 i_i} \nabla_{i_0} \nabla_{i_2} \cdots \stackrel{i_i}{\nabla}_\lambda \cdots \nabla_{i_l} E_{pj} \nabla^{i_0} \nabla^{I_l} H^{ij} \\
& - \epsilon_i{}^{\! p}{}_{\! i_i} \nabla_{i_0} \nabla_{i_2} \cdots \stackrel{i_i}{\nabla}_{i_1} \cdots \nabla_{i_l} E_{pj} \nabla^{i_0}  \nabla^{I_l} H^{ij}  \\
& - \epsilon_{i i_1}{}^{\! \lambda} \nabla_{i_0} \nabla_{I_{l-1}} H_{\lambda j} \nabla^{i_0} \nabla^{I_l} E^{ij} + \epsilon_{i i_1}{}^{\! \lambda} \nabla_{i_0}  \nabla_{I_{l-1}} E_{\lambda j} \nabla^{i_0}  \nabla^{I_l} H^{ij} 
\end{align*}

We begin with the subcase {\bf Case III-1}: $l=2$, i.e.
$ \nabla_{I_2} \left( \curl H \right)_{ij} \cdot  \nabla^{I_2} E^{ij} - \nabla_{I_2} \left( \curl E \right)_{ij} \cdot  \nabla^{I_2} H^{ij}$.
 That is, 
\begin{align*}
=&   \epsilon_i{}^{\! p}{}_{\! \lambda} \nabla_{i_1} \nabla_{i_2} E_{pj} \nabla^{i_1} \nabla^{i_2} \nabla^\lambda H^{ij}  - \epsilon_i{}^{\! p}{}_{\! \lambda} \nabla_{i_1} \nabla_{i_2} H_{pj}  \nabla^{i_1} \nabla^{i_2} \nabla^\lambda E^{ij}\\
= & \epsilon_i{}^{\! p}{}_{\! \lambda} \nabla_{i_1} \nabla_{i_2} E_{pj} \nabla^{i_1} \nabla^\lambda \nabla^{i_2}  H^{ij}  - \epsilon_i{}^{\! p}{}_{\! \lambda} \nabla_{i_1} \nabla_{i_2} H_{pj}  \nabla^{i_1}\nabla^\lambda  \nabla^{i_2} E^{ij}\\
& + \epsilon_i{}^{\! p}{}_{\! \lambda} \nabla_{i_1} \nabla_{i_2} E_{pj} \nabla^{i_1} [ \nabla^{i_2}, \nabla^\lambda] H^{ij}  - \epsilon_i{}^{\! p}{}_{\! \lambda} \nabla_{i_1} \nabla_{i_2} H_{pj}  \nabla^{i_1} [\nabla^{i_2}, \nabla^\lambda] E^{ij}.
\end{align*}  
The line $ \epsilon_i{}^{\! p}{}_{\! \lambda} \nabla_{i_1} \nabla_{i_2} E_{pj} \nabla^{i_1} \nabla^\lambda \nabla^{i_2}  H^{ij}  - \epsilon_i{}^{\! p}{}_{\! \lambda} \nabla_{i_1} \nabla_{i_2} H_{pj}  \nabla^{i_1}\nabla^\lambda  \nabla^{i_2} E^{ij}$ is reduced to case II. As calculated in case II \eqref{eq-Case-2-principle}, the principle part of the second line is
\begin{align*}
\simeq &  \epsilon^{q p}{}_{\! \lambda} \nabla_{i_1} \nabla_{q} E_{pj} \nabla^{i_1}  H^{\lambda j}  - \epsilon^{q p}{}_{\! \lambda} \nabla_{i_1} \nabla_{q} H_{pj}  \nabla^{i_1} E^{\lambda j},
\end{align*} 
which reduces to case II with $l=1$. 

As we only need $l \leq 3$ in this paper, we will focus on $l=3$ (the conclusion will be true for $l=2$ as well): $ \epsilon_i{}^{\! pq} \nabla_{i_1} [\nabla_{i_2}, \nabla_q] \nabla_{i_{3}} H_{pj} \cdot \nabla^{I_3} E^{ij}  - \epsilon_i{}^{\! pq} \nabla_{i_1} [\nabla_{i_2}, \nabla_q] \nabla_{i_{3}} E_{pj} \cdot \nabla^{I_3} H^{ij},$ whose principle part is
 \begin{align*}
\simeq &- \epsilon_i{}^{\! p \lambda}  g_{i_2 i_3} \nabla_{i_1} \nabla_\lambda H_{pj}  \nabla^{I_3} E^{ij}  + \epsilon_i{}^{\! p}{}_{\! i_3} \nabla_{i_1} \nabla_{i_2} H_{pj}  \nabla^{I_3} E^{ij}  \\
& + \epsilon_i{}^{\! p \lambda} g_{i_2 i_3} \nabla_{i_1} \nabla_\lambda E_{pj} \nabla^{I_3} H^{ij} - \epsilon_i{}^{\! p}{}_{\! i_3} \nabla_{i_1} \nabla_{i_2} E_{pj} \nabla^{I_3} H^{ij}  \\
& - \epsilon_{i i_2}{}^{\! \lambda} \nabla_{i_1} \nabla_{i_{3}} H_{\lambda j} \nabla^{I_3} E^{ij} + \epsilon_{i i_2}{}^{\! \lambda} \nabla_{i_1}  \nabla_{i_{3}} E_{\lambda j}  \nabla^{I_3} H^{ij} 
\end{align*} 
which can be rewritten as
 \begin{align*}
\simeq &- \epsilon_i{}^{\! p \lambda} \nabla_{i_1} \nabla_\lambda H_{pj}  \nabla^{i_1} \nabla^q \nabla_q E^{ij} \\
& + \epsilon_i{}^{\! p}{}_{\! \lambda} \nabla_{i_1} \nabla_{i_2} H_{pj}  \nabla^{i_1} \nabla^{i_2} \nabla^\lambda E^{ij}  \\
&+ \epsilon_i{}^{\! p \lambda}  \nabla_{i_1} \nabla_\lambda E_{pj} \nabla^{i_1} \nabla^q \nabla_q H^{ij}  \\
& - \epsilon_i{}^{\! p}{}_{\! \lambda} \nabla_{i_1} \nabla_{i_2} E_{pj} \nabla^{i_1} \nabla^{i_2} \nabla^\lambda H^{ij}  \\
&- \epsilon_{i p}{}^{\! \lambda} \nabla_{i_1} \nabla_{i_{3}} H_{\lambda j} \nabla_{i_1} \nabla^p \nabla_{i_{3}} E^{ij} \\
& + \epsilon_{i p}{}^{\! \lambda} \nabla_{i_1}  \nabla_{i_{3}} E_{\lambda j} \nabla_{i_1} \nabla^p \nabla_{i_{3}} H^{ij}. 
\end{align*} 
We index the above six terms by $-A_1,  \cdots, - A_6.$ Note that $A_5+A_6$ reduces to case II,
while $A_1, A_3$ can be further manipulated as 
 \begin{align*}
A_1 =& \epsilon_i{}^{\! p \lambda} \nabla_{i_1} \nabla_\lambda H_{pj} \nabla^q  \nabla^{i_1} \nabla_q E^{ij} +  \epsilon_i{}^{\! p \lambda} \nabla_{i_1} \nabla_\lambda H_{pj} [ \nabla^{i_1}, \nabla^q] \nabla_q E^{ij} \\
 =& A_{11} + A_{12}; \\
A_3 = & - \epsilon_i{}^{\! p \lambda}  \nabla_{i_1} \nabla_\lambda E_{pj} \nabla^q \nabla^{i_1}  \nabla_q H^{ij}   - \epsilon_i{}^{\! p \lambda}  \nabla_{i_1} \nabla_\lambda E_{pj} [\nabla^{i_1}, \nabla^q] \nabla_q H^{ij}\\
= & A_{31} + A_{32}.
\end{align*} 
$A_{11}, A_{31}$ are further made into:
 \begin{align*}
A_{11} =& \nabla^q \left( \epsilon_i{}^{\! p \lambda} \nabla_{i_1} \nabla_\lambda H_{pj} \nabla^{i_1} \nabla_q E^{ij}  \right) - \epsilon_i{}^{\! p \lambda}  \nabla^q  \nabla_{i_1} \nabla_\lambda H_{pj} \nabla^{i_1} \nabla_q E^{ij}\\
=& div - \epsilon_i{}^{\! p \lambda}  \nabla_{i_1}  \nabla^q  \nabla_\lambda H_{pj} \nabla^{i_1} \nabla_q E^{ij} - \epsilon_i{}^{\! p \lambda}  [\nabla^q,  \nabla_{i_1}] \nabla_\lambda H_{pj} \nabla^{i_1} \nabla_q E^{ij}\\
A_{31} = & - \nabla^q \left( \epsilon_i{}^{\! p \lambda}  \nabla_{i_1} \nabla_\lambda E_{pj} \nabla^{i_1}  \nabla_q H^{ij} \right) + \epsilon_i{}^{\! p \lambda} \nabla^q   \nabla_{i_1} \nabla_\lambda E_{pj} \nabla^{i_1}  \nabla_q H^{ij}  \\
=& div  + \epsilon_i{}^{\! p \lambda}  \nabla_{i_1} \nabla^q  \nabla_\lambda E_{pj} \nabla^{i_1}  \nabla_q H^{ij}   + \epsilon_i{}^{\! p \lambda} [\nabla^q, \nabla_{i_1}] \nabla_\lambda E_{pj} \nabla^{i_1}  \nabla_q H^{ij}.
\end{align*}
Note that the second term in $A_{11}$ will be cancelled with $A_4$ and the second term in $A_{31}$ is cancelled with $A_2$, hence
 $A_1 + \cdots A_4$ is equals to
  \begin{align*}
&+  \epsilon_i{}^{\! p \lambda} \nabla_{i_1} \nabla_\lambda H_{pj} [ \nabla^{i_1}, \nabla^q] \nabla_q E^{ij} - \epsilon_i{}^{\! p \lambda}  [\nabla^q,  \nabla_{i_1}] \nabla_\lambda H_{pj} \nabla^{i_1} \nabla_q E^{ij} \\
& - \epsilon_i{}^{\! p \lambda}  \nabla_{i_1} \nabla_\lambda E_{pj} [\nabla^{i_1}, \nabla^q] \nabla_q H^{ij} + \epsilon_i{}^{\! p \lambda} [\nabla^q, \nabla_{i_1}] \nabla_\lambda E_{pj} \nabla^{i_1}  \nabla_q H^{ij},
\end{align*} 
whose principle part is, after substituting the principle part of curvature $R_{imjn}$,
\begin{align*}
&+  \epsilon_{ip}{}^{\! \lambda}  \nabla_\lambda H_{qj} \left( \nabla^{p} \nabla^q - \nabla^q \nabla^p \right) E^{ij} -  \epsilon_{ip}{}^{\! \lambda}  \nabla_\lambda E_{qj} \left( \nabla^{p} \nabla^q - \nabla^q \nabla^p \right) H^{ij} \\
& + \epsilon_i{}^{\! p \lambda} \nabla_\lambda H_{pq} \left( \nabla_{j} \nabla^q - \nabla^q \nabla_j \right) E^{ij} - \epsilon_i{}^{\! p \lambda} \nabla_\lambda E_{pq} \left( \nabla_{j} \nabla^q - \nabla^q \nabla_j \right) H^{ij} \\
&+ \epsilon^{i p \lambda} \nabla_\lambda H_{pj} \left(\nabla_q \nabla^i - \nabla^{i}\nabla_q \right) E^{qj} - \epsilon^{i p \lambda} \nabla_\lambda E_{pj} \left(\nabla_q \nabla^i - \nabla^{i}\nabla_q \right) H^{qj} \\
&+ \epsilon_i{}^{\! p \lambda} \nabla_\lambda H_{pj} \left(\nabla_q \nabla^j - \nabla^{j}\nabla_q \right) E^{iq} -  \epsilon_i{}^{\! p \lambda} \nabla_\lambda E_{pj} \left(\nabla_q \nabla^j - \nabla^{j}\nabla_q \right) H^{iq}.
\end{align*} 
Again, we calculate the principle part:
\begin{align*}
&+  \epsilon_{ip}{}^{\! \lambda}  \nabla_\lambda H_{ij}  E^{pj} -    \epsilon_{ip}{}^{\! \lambda}  \nabla_\lambda E_{ij}  H^{pj} \\
& + \epsilon_i{}^{\! p \lambda} \nabla_\lambda H_{qp} E^{iq} - \epsilon_i{}^{\! p \lambda} \nabla_\lambda E_{qp} H^{iq}  \\
& + \epsilon_j{}^{\! p \lambda} \nabla_\lambda H_{qp} E^{qj} - \epsilon_j{}^{\! p \lambda} \nabla_\lambda E_{qp} H^{qj}  \\
& +2 \epsilon_i{}^{\! p \lambda} \nabla_\lambda H_{qp} E^{iq} - 2\epsilon_i{}^{\! p \lambda} \nabla_\lambda E_{qp} H^{iq}  \\
& + \epsilon_q{}^{\! p \lambda} \nabla_\lambda H_{pj} E^{qj} - \epsilon_q{}^{\! p \lambda} \nabla_\lambda E_{pj} H^{qj}  \\
& +2 \epsilon_q{}^{\! p \lambda} \nabla_\lambda H_{pi} E^{iq} - 2\epsilon_q{}^{\! p \lambda} \nabla_\lambda E_{pi} H^{iq}  \\
& + \epsilon_j{}^{\! p \lambda} \nabla_\lambda H_{pq} E^{qj} - \epsilon_j{}^{\! p \lambda} \nabla_\lambda E_{pq} H^{qj}  \\
& +2 \epsilon_i{}^{\! p \lambda} \nabla_\lambda H_{pq} E^{iq} - 2\epsilon_i{}^{\! p \lambda} \nabla_\lambda E_{pq} H^{iq}. 
\end{align*} 
 Up to some divergence forms, they are
\begin{align*}
&-  \epsilon_{ip}{}^{\! \lambda}  H_{ij}  \nabla_\lambda E^{pj} -   \epsilon_{ip}{}^{\! \lambda}  \nabla_\lambda E_{ij}  H^{pj} (=0) \\
& - \epsilon_i{}^{\! p \lambda}  H_{qp}\nabla_\lambda E^{iq} - \epsilon_i{}^{\! p \lambda} \nabla_\lambda E_{qp} H^{iq} (=0)  \\
& - \epsilon_j{}^{\! p \lambda}  H_{qp}\nabla_\lambda E^{qj} - \epsilon_j{}^{\! p \lambda} \nabla_\lambda E_{qp} H^{qj} (=0) \\
& - 2 \epsilon_i{}^{\! p \lambda} H_{qp} \nabla_\lambda E^{iq} - 2\epsilon_i{}^{\! p \lambda} \nabla_\lambda E_{qp} H^{iq} (=0) \\
& - \epsilon_q{}^{\! p \lambda}  H_{pj} \nabla_\lambda E^{qj} - \epsilon_q{}^{\! p \lambda} \nabla_\lambda E_{pj} H^{qj} (=0) \\
& - 2 \epsilon_q{}^{\! p \lambda} H_{pi}  \nabla_\lambda E^{iq} - 2\epsilon_q{}^{\! p \lambda} \nabla_\lambda E_{pi} H^{iq} (=0) \\
& - \epsilon_j{}^{\! p \lambda}  H_{pq} \nabla_\lambda E^{qj} - \epsilon_j{}^{\! p \lambda} \nabla_\lambda E_{pq} H^{qj} (=0) \\
& - 2 \epsilon_i{}^{\! p \lambda} H_{pq} \nabla_\lambda  E^{iq} - 2\epsilon_i{}^{\! p \lambda} \nabla_\lambda E_{pq} H^{iq} (=0). 
\end{align*} 
 Hence, we proved in case III ($l \leq 3$):
\begin{align*}
& \epsilon_i{}^{\! pq} \nabla_{i_1} \nabla_{i_2} \nabla_q \nabla_{I_{l-2}} H_{pj} \cdot \nabla^{I_l} E^{ij}  - \epsilon_i{}^{\! pq} \nabla_{i_1} \nabla_{i_2} \nabla_q \nabla_{I_{l-2}} E_{pj} \cdot \nabla^{I_l} H^{ij} \\
\simeq &\sum_{a\leq l} \nabla_i \left( \nabla_{I_a} H * \nabla_{I_a} E \right)^i.
\end{align*}
  
Let us come to the final {\bf Case IV}: $ \epsilon_i{}^{\! pq} \nabla_{I_3} \nabla_q H_{pj} \cdot \nabla^{I_3} E^{ij}  - \epsilon_i{}^{\! pq} \nabla_{I_3} \nabla_q E_{pj} \cdot \nabla^{I_3} H^{ij}.$
Combined with the previous cases III, it reduces to consider
\begin{align*}
& \epsilon_i{}^{\! pq} \nabla_{I_2} [ \nabla_{i_3}, \nabla_q] H_{pj} \cdot \nabla^{I_3} E^{ij}  - \epsilon_i{}^{\! pq} \nabla_{I_2} [\nabla_{i_3}, \nabla_q] E_{pj} \cdot \nabla^{I_3} H^{ij} \\
\simeq&\epsilon_i{}^{\! p \lambda} \nabla_{I_2} H_{\lambda j} \cdot \nabla^{I_2} \nabla_p E^{ij}  - \epsilon_i{}^{\! p \lambda} \nabla_{I_2} E_{\lambda j} \cdot \nabla^{I_2} \nabla_p H^{ij}\\
 =&\epsilon_i{}^{\! p \lambda} \nabla_{I_2} H_{\lambda j} \cdot \nabla^{i_1} \nabla_p \nabla^{i_2}  E^{ij}  - \epsilon_i{}^{\! p \lambda} \nabla_{I_2} E_{\lambda j} \cdot \nabla^{i_1} \nabla_p \nabla^{i_2}  H^{ij}\\
 &+\epsilon_i{}^{\! p \lambda} \nabla_{I_2} H_{\lambda j} \cdot \nabla^{i_1} [\nabla^{i_2}, \nabla_p] E^{ij}  - \epsilon_i{}^{\! p \lambda} \nabla_{I_2} E_{\lambda j} \cdot\nabla^{i_1} [\nabla^{i_2}, \nabla_p] H^{ij}.
\end{align*}
 Note that the last line is already covered by case III.  The remaining part equals to
 \begin{align*}
 =&\epsilon_i{}^{\! p \lambda} \nabla_{I_2} H_{\lambda j} \cdot \nabla_p \nabla^{I_2}  E^{ij}  - \epsilon_i{}^{\! p \lambda} \nabla_{I_2} E_{\lambda j} \cdot \nabla_p \nabla^{I_2} H^{ij}\\
 &+ \epsilon_i{}^{\! p \lambda} \nabla_{I_2} H_{\lambda j} \cdot [\nabla^{i_1}, \nabla_p] \nabla^{i_2}  E^{ij}  - \epsilon_i{}^{\! p \lambda} \nabla_{I_2} E_{\lambda j} \cdot [\nabla^{i_1}, \nabla_p] \nabla^{i_2}  H^{ij}.
\end{align*}
Again, the last line above is covered by case II. Hence we are left with
   \begin{align*}
&\epsilon_i{}^{\! p \lambda} \nabla_{I_2} H_{\lambda j} \cdot \nabla_p \nabla^{I_2}  E^{ij}  - \epsilon_i{}^{\! p \lambda} \nabla_{I_2} E_{\lambda j} \cdot \nabla_p \nabla^{I_2} H^{ij},
\end{align*}
 which is covered by case I ($l=2$). Hence, we prove in case IV ($l = 3$):
\begin{align*}
& \epsilon_i{}^{\! pq} \nabla_{I_3} \nabla_q H_{pj} \cdot \nabla^{I_3} E^{ij}  - \epsilon_i{}^{\! pq} \nabla_{I_3} \nabla_q E_{pj} \cdot \nabla^{I_3} H^{ij} \\
\simeq &\sum_{a \leq 3} \nabla_i \left( \nabla_{I_a} H * \nabla_{I_a} E \right)^i.
\end{align*}

\end{proof}

\end{document}